\newif\iffull\fulltrue
\newif\ifacm\acmfalse
\newif\iftwocol\twocolfalse

\documentclass{lmcs}
\keywords{model checking, probabilistic programs, higher-order programs, termination probabilities}

\usepackage{booktabs}   \usepackage{subcaption} \usepackage{multirow}
\usepackage{color,multicol}
\usepackage{bcprules,proof,myproof}
\usepackage{qtree}
\usepackage{alltt}
\usepackage[all]{xy}
\newenvironment{asparaenum}{\begin{enumerate}}{\end{enumerate}}
\usepackage{graphicx}
\usepackage{underscore}
\usepackage{hyperref}
\usepackage{amsthm}

\theoremstyle{plain}
\newtheorem{definition}{Definition}[section]
\newtheorem{ex}{Example}[section]
\newenvironment{example}{\begin{ex}\rm}{\end{ex}}
\newtheorem{theorem}{Theorem}[section]
\newtheorem{corollary}[theorem]{Corollary}
\newtheorem{lemma}[theorem]{Lemma}
\newtheorem{proposition}[theorem]{Proposition}
\newenvironment{proofn}[1]{\paragraph{#1.}}{\qed\\}
\newtheorem{remark}[theorem]{Remark}
\usepackage{stmaryrd}
\usepackage{amssymb,amsmath}
\usepackage{listings}
\definecolor{purple}{RGB}{0,40,40}

\newcommand\emptyTE{\emptyset}
\newcommand\pN{\p_{\NONTERMS}}
\newcommand\comment[1]{}
\newcommand\hquad{\hspace*{0.8em}}
\newcommand\PHORSSET{\mathcal{P}}
\newcommand\fgroup{\mathtt{fgrp}}
\newcommand\evexp[1]{\mathtt{event}\ #1;}
\newcommand\SETofPHORS{\Psi}

\newcommand\listgen{\textit{Listgen}}
\newcommand\listlistgen{\textit{ListgenList}}
\newcommand\listboolgen{\textit{ListgenBool}}
\newcommand\listgene{\textit{ListgenE}}
\newcommand\listgeno{\textit{ListgenO}}
\newcommand\treegen{\textit{Treegen}}
\newcommand\boolgen{\textit{Boolgen}}
\newcommand\deter{\textit{Determinize}}
\newcommand\one{\textit{One}}
\newcommand\zero{\textit{Zero}}
\newcommand\forallp{\textit{ForallP}}
\newcommand\Avg{\textit{Avg}}

\newcommand\TRUE{\mathtt{true}}
\newcommand\FALSE{\mathtt{false}}
\newcommand{\commentout}[1]{}
\newcommand\sem[1]{\llbracket{#1}\rrbracket}
\newcommand\seme[2]{\llbracket{#1}\rrbracket_{#2}}
\newcommand\EQref[1]{\E^{\mathtt{ref}}_{#1}}
\newcommand\Sub[2]{\{#2/#1\}}
\newcommand\Subs[1]{\{#1\}}
\newcommand\tr{\leadsto}
\newcommand\To{\Rightarrow}
\newcommand\E{\mathcal{E}}
\newcommand\seql[1]{|\seq{#1}|}
\newcommand\EQtwo[1]{\mathcal{E}_{2,\GRAM}}
\newcommand\trT[1]{{#1}^\dagger}
\newcommand\trTp[1]{{#1}^{\dagger'}}
\newcommand\trTn[2]{{#1}^{\dagger+{#2}}}
\newcommand\arity{\mathtt{ar}}
\newcommand\Hole{[\,]}
\newcommand\F{\mathcal{F}}
\newcommand\GRAMAST[1]{\GRAM^{P,Q}_{#1}}
\newcommand\R{\mathbf{R}}
\newcommand\LFP{\mathbf{lfp}}
\newcommand\abs[1]{\widehat{#1}}
\newcommand\Leq{\sqsubseteq} \newcommand\Lub{\bigsqcup} \newcommand\Dir{d} \newcommand\lambdaY{\(\lambda Y\)}
\newcommand\pHORS{PHORS}
\newcommand\pHORSs{\pHORS{}}
\newcommand\ceil[1]{\lceil{#1}\rceil}
\newcommand\floor[1]{\lfloor{#1}\rfloor}
\newcommand\Nat{\mathbf{Nat}}

\newcommand\NONTERMS{\mathcal{N}}
\newcommand\NTE{\NONTERMS}
\newcommand\GRAM{\mathcal{G}}

\newcommand\RULES{\mathcal{R}}

\newcommand\redp[2]{\xrightarrow{#1}_{#2}}
\newcommand\newredp[1]{\xrightarrow{#1}_{\texttt{es}}}
\newcommand\newredpg[2]{\xrightarrow{#1}_{\texttt{es},#2}}
\newcommand\newredsp[1]{\xRightarrow{#1}_{\mathtt{es}}}
\newcommand\newredspg[2]{\stackrel{#1}{\Red}_{\mathtt{es},#2}}
\makeatletter
\newcommand{\xRightarrow}[2][]{\ext@arrow 0359\Rightarrowfill@{#1}{#2}}
\makeatother
\newcommand\redsp[2]{\xRightarrow{#1}_{#2}}
\newcommand\Red{\Longrightarrow}
\newcommand\sty{\kappa}
\newcommand\STE{\mathcal{K}}
\newcommand\T{\mathtt{o}}
\newcommand\Reals{\mathbf{R}}
\newcommand\realp{\mathbf{R}_\infty}
\newcommand\realt{\mathtt{R}}
\newcommand{\qv}[2]{#1^\#_{#2}}
\newcommand{\metalambda}{\mathop{\rlap{$\lambda$}\mkern2mu
    \raisebox{.275ex}{$\lambda$}}}
\newcommand{\envsubst}[3]{#1[#2\leftarrow #3]}
\newcommand\Te{\mathtt{e}}
\newcommand\p{\vdash}
\newcommand\COL{\mathbin{:}}
\newcommand\dom{\mathit{dom}}
\newcommand\set[1]{\{#1\}}
\newcommand\C[1]{\,\oplus_{#1}\,}
\newcommand\Prob{\mathcal{P}}
\newcommand\ProbES{\mathcal{P}_{\mathtt{es}}}
\newcommand\cs{\pi}

\newcommand\order{\mathit{order}}
\newcommand\TtoP[1]{{#1}^{\#}}
\newcommand\seq[1]{\widetilde{#1}}

\newcommand{\midd}{\; \; \mbox{\Large{$\mid$}}\;\;} 
 \newif\ifdraft\draftfalse
\ifdraft
\definecolor{DarkGreen}{RGB}{0,100,0}
\newcommand\nk[1]{\textcolor{red}{[#1 -nk]}}
\newcommand\udl[1]{\textcolor{green}{[#1 -udl]}}

\else
\newcommand\nk[1]{}
\newcommand\udl[1]{}

\fi

\ifdraft
\newcommand\cg[1]{\textcolor{blue}{[#1 -cg]}}
\else
\newcommand\cg[1]{}
\fi

\newenvironment{varitemize}
{
\begin{list}{\labelitemi}
{\setlength{\itemsep}{0pt}
 \setlength{\topsep}{0pt}
 \setlength{\parsep}{0pt}
 \setlength{\partopsep}{0pt}
 \setlength{\leftmargin}{15pt}
 \setlength{\rightmargin}{0pt}
 \setlength{\itemindent}{0pt}
 \setlength{\labelsep}{5pt}
 \setlength{\labelwidth}{10pt}
}}
{
 \end{list} 
}

\newcounter{number}

\title[On the Termination of Probabilistic Higher-Order Programs]{On the Termination Problem\\ for Probabilistic Higher-Order Recursive Programs}
\author[N.~Kobayashi]{Naoki Kobayashi}
\address{The University of Tokyo, Japan}
\author[U. Dal Lago]{Ugo Dal Lago}
\address{University of Bologna, Italy}
\author[C. Grellois]{Charles Grellois}
\address{Aix Marseille Univ, Universit\'e de Toulon, CNRS, LIS, Marseille, France}

\begin{document}

\begin{abstract}
  In the last two decades, there has been much progress on model
  checking of both \emph{probabilistic} systems and \emph{higher-order}
  programs. In spite of the emergence of higher-order probabilistic
  programming languages, not much has been done to combine those two
  approaches. 
\iftwocol
As a first step towards model checking of probabilistic higher-order programs, 
we introduce \pHORS{}, a probabilistic
extension of higher-order recursion schemes (HORS), as a model of
probabilistic higher-order programs.
\else
In this paper, we initiate a study on the
  probabilistic higher-order model checking problem, by giving some
  first theoretical and experimental results.  
As a first step towards our goal, we introduce \pHORS{}, a probabilistic
extension of higher-order recursion schemes (HORS), as a model of
probabilistic higher-order programs.
\fi
The model of \pHORS{} may alternatively be viewed as a higher-order
  extension of recursive Markov chains.  We then investigate the
  probabilistic termination problem --- or, equivalently, the
  probabilistic reachability problem. We prove that almost sure
  termination of order-2 \pHORS{} is undecidable.  We also provide a
  fixpoint characterization of the termination probability of
  \pHORS{}, and develop a sound (although possibly incomplete) procedure
  for approximately computing the termination probability.  
\iftwocol\else
We have
  implemented the procedure for order-2 \pHORS{}, and confirmed that
  the procedure works well through preliminary experiments.
\fi
\end{abstract}
 \maketitle

\section{Introduction}
\label{sec:intro}
\label{SEC:INTRO}
Computer science has interacted with probability theory in many
fruitful ways, since the very early
days~\cite{shannon-schapiro:computability-proba-machines}.
Probability theory enables \emph{state abstraction}, reducing in this
way the state space's cardinality.  It has also led to a
new \emph{model of computation}, used for instance in randomized
computation~\cite{MotwaniRandomized} or in
cryptography~\cite{GoldwasserMicali}.  The trend of a rise of
probability theory's importance in computer science has been followed
by the programming language community, up to the point that
probabilistic programming is nowadays a very active research
area. Probabilistic choice can be modeled in various ways in
programming, and fair binary probabilistic choice is for instance
perfectly sufficient to obtain universality if the underlying
programming language is universal
itself~\cite{Santos69,dallagozorzi2012}. This has been the path
followed in probabilistic
$\lambda$-calculi~\cite{SahebDjahromi,JonesPlotkin1989,danosharmer,dallagozorzi2012,ehrhardtassonpagani2014,DLSA14,HKSY17}.

In the present paper, we are interested in the analysis of probabilistic,
higher-order recursive programs. Model checking of probabilistic
finite state systems has been a very active research field
(see~\cite{PrinciplesMC,HandbookMC} for a survey).  Over the last two
decades, there has also been much interest and progress in model
checking of probabilistic \emph{recursive}
programs~\cite{Etessami09,DBLP:journals/jacm/EtessamiY15,DBLP:journals/fmsd/BrazdilEKK13,DBLP:journals/jcss/BrazdilBFK14},
which \emph{cannot} be modeled as finite state systems, and thus
escape the classic model checking framework and algorithms. None of
the proposals in the literature on probabilistic model checking,
however, is capable of handling \emph{higher-order} functions, which
are a natural feature in functional languages. This is in sharp
contrast with what happens for \emph{non-probabilistic} higher-order
programs, for which model checking techniques can be fruitfully
employed for proving both reachability and safety properties, as shown
in the extensive literature on the subject
(e.g.~\cite{Ong06LICS,Hague08LICS,Kobayashi13JACM,KO09LICS,KSU11PLDI,DBLP:conf/csl/GrelloisM15,DBLP:conf/mfcs/GrelloisM15,DBLP:conf/csl/TsukadaO14,Salvati11ICALP}).
There have been some studies on the termination of probabilistic
higher-order programs~\cite{DBLP:conf/esop/LagoG17},
but to our knowledge, they have not provided a procedure for
\emph{precisely} computing the termination probability,
nor  discussed whether it is possible at all:
see Section~\ref{sec:related} for more details.
Summing up, little has been known about the decidability and
 complexity of model checking of {probabilistic higher-order}
 programs, and even less about the existence of practical procedures
 for approximately solving model checking problems.

One may think that probabilistic \emph{and} higher-order computation
is rather an exotic research topic, but it is important for precisely
modeling and verifying any higher-order functional programs that
interact with a probabilistic environment.  As a simple example,
consider the following (non-higher-order) \textsf{OCaml}-like program,
which uses a primitive
\texttt{flip} for generating \texttt{true} or \texttt{false} with probability \(\frac{1}{2}\).
\begin{center}
\small
\begin{lstlisting}[basicstyle=\ttfamily]
 let rec f() = if flip() then () else f() in f()
\end{lstlisting}
\end{center}
The program almost surely terminates (i.e., terminates with
probability 1), but if we ignore the probabilistic aspects and
model \texttt{flip()} as a \emph{non-deterministic} (rather
than \emph{probabilistic}) primitive, then we would conclude that the
program can \emph{may} diverge. The following program makes use of an
interesting combination of probabilistic choice and higher-order
functions:
\begin{center}
\small
\begin{lstlisting}[basicstyle=\ttfamily]
 let boolgen() = flip()
 let rec listgen f () = 
   if flip() then [] else f()::listgen f ()
 in listgen (listgen boolgen) ()
\end{lstlisting}
\end{center}
The function \texttt{listgen} above takes a generator \texttt{f} of
elements as an argument, and creates a list of elements, each of them
obtained by calling
\(f\). Thus, the whole program generates a list of lists of
Booleans. The length of such a list of lists is randomized, and
distributed according to the geometric distribution. We may then wish
to ask, for example, (i) whether it almost surely terminates, and (ii)
what is the probability that a list of even length is
generated. Generating random data structures like the
one produced by \texttt{listgen} is not an artificial task, being central to, e.g., random test
generation~\cite{DBLP:conf/icst/PouldingF17a,DBLP:journals/corr/abs-1808-01520}.

As a model of probabilistic higher-order programs, we first
introduce \pHORS{}, a probabilistic extension of higher-order
recursion schemes (HORS)~\cite{Knapik02FOSSACS,Ong06LICS}.  Our model
of \pHORS{} is expressive enough to accurately model probabilistic
higher-order functions, but the underlying non-probabilistic language
(i.e., HORS, obtained by removing probabilistic choice) is \emph{not}
Turing-complete; thus, we can hope for the existence of algorithmic
solutions to some of the verification problems. As an example, we can
decide indeed whether the termination probability of \pHORS{}
is \(0\), by reduction to a model checking problem for
non-probabilistic HORS.

Through the well-known correspondence between HORS and (collapsible) higher-order
pushdown automata~\cite{Knapik02FOSSACS,Hague08LICS},
\pHORS{} can be considered a higher-order extension of probabilistic pushdown 
systems~\cite{DBLP:journals/fmsd/BrazdilEKK13,DBLP:journals/jcss/BrazdilBFK14}
and of recursive Markov chains~\cite{DBLP:conf/qest/YannakakisE05},
the computation models used in previous work on model checking of
probabilistic recursive programs.  We can also view \pHORS{} as an
extension of
the \lambdaY{}-calculus~\cite{DBLP:journals/apal/Statman04} with
probabilities, just like HORS can be viewed as an alternative
presentation of the \lambdaY{}-calculus.  The correspondence between
HORS and the \lambdaY{}-calculus has been useful for transferring
techniques for typed \(\lambda\)-calculi (most notably, game
semantics~\cite{Ong06LICS}, intersection
types~\cite{Kobayashi09POPL,KO09LICS} and Krivine
machines~\cite{Salvati11ICALP}) to HORS; thus, we expect similar
benefits in using \pHORS{} (rather than probabilistic higher-order
pushdown automata) as models of probabilistic higher-order programs.

As a first step towards understanding the nature of the
model checking problem for probabilistic higher-order programs,
the present paper studies the problem of computing the termination (or equivalently, reachability)
probabilities of \pHORS{}. Note that, as in a non-probabilistic setting,
one can easily reduce a safety property verification problem to a 
may-termination problem (i.e. the problem of checking whether a program may terminate),
by encoding safety violation as termination. We can also verify certain
liveness properties, by encoding a good event as a termination and checking that
the termination probability is \(1\).
As we will see in Section~\ref{sec:problem},
the two questions (i) and (ii) mentioned earlier on the \texttt{listgen} program 
can also be reduced to
the problem of computing the termination probability of a \pHORS{}.
Note also that computing the termination (or equivalently, reachability)
probability has been a key to solving more general model checking problems
(such as LTL/CTL model checking) for 
recursive programs~\cite{DBLP:conf/qest/YannakakisE05,DBLP:journals/jcss/BrazdilBFK14}.

As the first result on the problem of computing termination probabilities, we prove that 
the almost sure termination problem, i.e., whether a given \pHORS{} terminates with
probability \(1\), is undecidable at order-2 or higher. This contrasts with
the case of recursive Markov chains, for which the almost sure termination problem can
be decided in PSPACE~\cite{Etessami09}. The proof of undecidability is based on
a reduction from the undecidability of
Hilbert's tenth problem (i.e. unsolvability of Diophantine equations)~\cite{Diophantine}.
The undecidability result also implies that it is not possible to
compute the \emph{exact} termination probability. More precisely, for any rational
number \(r\in(0,1]\),
the set \(\set{\GRAM \mid \Pr(\GRAM)\geq r}\) 
(where \(\Pr(\GRAM)\) denotes the termination probability of \(\GRAM\))
is not recursively enumerable (in other words, the set is 
\(\Pi^0_1\)-hard in the arithmetical hierarchy).
Note, however, that this negative result does not preclude the possibility to
compute the termination probability with arbitrary precision; there may exist
an algorithm that, given a \pHORS{} \(\GRAM\) and \(\epsilon>0\) as inputs, finds \(r\) such that
the termination probability of \(\GRAM\) belongs to \((r,r+\epsilon)\).
The existence of such an approximation algorithm remains open.

As a positive result towards approximately computing the termination probability, we show that 
the termination probability of order-\(n\) \pHORS{} can be characterized by fixpoint equations
on order-(\(n-1\)) functions on real numbers. 
The fixpoint characterization of the termination probability of recursive Markov chains~\cite{Etessami09}
can be viewed as a special case of our result where \(n=1\).
The fixpoint characterization immediately
provides a semi-algorithm for 
 the lower-bound problem: ``Given a \pHORS{} \(\GRAM\) and a rational number \(r\in[0,1]\), does \(\Pr(\GRAM)>r\) hold?''
Recall, however, that \(\set{\GRAM \mid \Pr(\GRAM)\geq r}\) is \emph{not} recursively enumerable, so there is no semi-algorithm for the variation:
``Given a \pHORS{} \(\GRAM\) and a rational number \(r\in[0,1]\), does \(\Pr(\GRAM)\geq r\) hold?''

The remaining question is whether an upper-bound on the termination
probability can be computed with arbitrary precision. We have not
settled this question yet, but propose a procedure for \emph{soundly
estimating} an upper-bound of the termination probability of
order-2 \pHORS{} by using the fixpoint characterization above, \emph{\`a la} FEM (finite element method). 
 We have
implemented the procedure, and conducted preliminary experiments to confirm
that the procedure works fairly well in practice:
combined with the lower-bound computation
based on the fixpoint characterization, the procedure was able to instantly compute the termination
probabilities of (small but) non-trivial examples with precision \(10^{-2}\).
We also briefly discuss how to generalize the procedure to deal with \pHORS{} of arbitrary orders.

The contributions of this article can thus be summarized as follows:
\begin{enumerate}[(i)]
\item
  A formalization of probabilistic higher-order recursion schemes
(\pHORS{}) and their termination probabilities. This is in
Section~\ref{sec:problem}.
\item
  A proof of undecidability of the almost sure termination problem
for \pHORS{} (of order 2 or higher), which can be found in
Section~\ref{sec:undecidability}.
\item
  A fixpoint characterization of the termination probability
of \pHORS{}, which immediately yields the semi-decidability of the
lower-bound problem. This is in Section~\ref{sec:fixpoint}.
\item
  A sound procedure for computing an upper-bound to the termination
probability of order-2 \pHORS{} (which is described in
Section~\ref{sec:upperbound}) accompanied by an implementation and
preliminary experiments with promising results, reported in
Section~\ref{sec:exp}.
\end{enumerate}
We also discuss related work in Section~\ref{sec:related},
and conclude the article in Section~\ref{sec:conc}.
A preliminary summary of this article appeared in Proceedings of LICS 2019~\cite{KDG19LICS}.
 
\section{Probabilistic Higher-Order Recursion Schemes (\pHORS{}) and Termination Probabilities}
\label{sec:problem}
\label{SEC:PROBLEM}
This section introduces \emph{probabilistic higher-order recursion
schemes} (\pHORS{}\footnote{We write \pHORS{} for both singular and plural forms.}), an extension of higher-order recursion
schemes~\cite{Knapik02FOSSACS,Ong06LICS}
in which programs
can at any evaluation step perform a discrete probabilistic choice, then proceeding according to
its outcome. Higher-order recursion schemes are usually treated as
generators of infinite trees, but as we are only interested in the
termination probability, we consider only nullary tree
constructors \(\Te\) and \(\Omega\), which represent termination
and divergence respectively.

We first define types and applicative terms. The set of \emph{types},
ranged over by \(\sty\), is given by:
\[\sty ::= \T \mid \sty_1\to\sty_2.\]
Intuitively, \(\T\) describes the unit value, and \(\sty_1\to\sty_2\) describes functions from
\(\sty_1\) to \(\sty_2\).
As usual, the \emph{order} of a type \(\sty\) is defined by:
\begin{align*}
\order(\T)&=0\\ 
\order(\sty_1\to\sty_2)&=\max(\order(\sty_1)+1,\order(\sty_2)).
\end{align*}
We often write \(\T^\ell\to\T\) for \(\underbrace{\T\to\cdots\to\T}_\ell\to\T\),
and abbreviate \(\sty_1\to\cdots\to\sty_k\to\sty\) to \(\seq{\sty}\to\sty\).
\noindent
The set of \emph{applicative terms}, ranged over by \(t\), is given by:
\[ t ::= \Te\ \midd \Omega\midd \ x\ \midd\ t_1t_2,\]
where \(\Te\) and \(\Omega\) are (the only) constants of type \(\T\) and \(x\) ranges over a set of variables. Intuitively, \(\Te\) and \(\Omega\) denote termination and
divergence respectively (the latter can be defined
as a derived form, but assuming it as a primitive
is convenient for Section~\ref{sec:fixpoint}).
We consider the following standard simple type system for applicative terms, where
\(\STE\), called a \emph{type environment}, is a map from a finite set of variables to the set of types.
\iftwocol
$$
\infer{\STE\p \Te:\T}{}
\qquad\qquad
\infer{\STE\p \Omega:\T}{}
\qquad\qquad
\infer{\STE\p x\COL\sty}{\STE(x)=\sty}
$$
$$
\infer{\STE\p t_1t_2:\sty}{\STE\p t_1:\sty_2\to\sty\andalso \STE\p t_2:\sty_2}
$$
\else
$$
\infer{\STE\p \Te:\T}{}
\qquad\qquad
\infer{\STE\p \Omega:\T}{}
\qquad\qquad
\infer{\STE\p x\COL\sty}{\STE(x)=\sty}
\qquad\qquad
\infer{\STE\p t_1t_2:\sty}{\STE\p t_1:\sty_2\to\sty\andalso \STE\p t_2:\sty_2}
$$
\fi
\begin{definition}[\pHORS{}]
A \emph{probabilistic higher-order recursion scheme} (\pHORS{}) 
is a triple \(\GRAM=(\NONTERMS,\RULES,S)\),
where:
\begin{asparaenum}
\item \(\NONTERMS\) is a map from a finite set of variables (called \emph{non-terminals}
and typically denoted $F,\,G,\,\ldots$)
to the set of types.
\item \(\RULES\) is a map from \(\dom(\NONTERMS)\) to 
terms of the form \(\lambda x_1.\cdots\lambda x_k.t_L\C{p}t_R\),
where \(p\in[0,1]\) is a rational number, and \(t_L,\,t_R\) are applicative terms.
If \(\NONTERMS(F)=\sty_1\to\cdots\to\sty_k\to\T\), \(\RULES(F)\) must be of the form
\(\lambda x_1.\cdots\lambda x_k.t_L\C{p}t_R\), where
\(\NONTERMS,x_1\COL\sty_1,\ldots,x_k\COL\sty_k \p t_L:\T\) and
\(\NONTERMS,x_1\COL\sty_1,\ldots,x_k\COL\sty_k \p t_R:\T\).
\item \(S\in \dom(\NONTERMS)\), called the \emph{start symbol}, is a distinguished non-terminal that satisfies \(\NONTERMS(S)=\T\).
\end{asparaenum}
The \emph{order} of a \pHORS{} \((\NONTERMS,\RULES,S)\)
is \iftwocol\else \(\max_{F\in\dom(\NONTERMS)} \order(\NONTERMS(F))\), i.e., \fi
the highest order of the types of its non-terminals.
We write \(\PHORSSET_k\) for the set of order-\(k\) \pHORS{}.
\end{definition}

\noindent
When \(\RULES(F)= \lambda x_1.\cdots\lambda x_k.t_L\C{p}t_R\),
we often write $F\,x_1\,\cdots\,x_k = t_L\C{p}t_R$, 
and specify \(\RULES\) as a set of such  equations.
The rule $F\,x_1\,\cdots\,x_k = t_L\C{p}t_R$ 
intuitively means that
\(F\,t_1\,\cdots\,t_k\) is reduced to \([t_1/x_1,\ldots,t_k/x_k]t_L\)
and \([t_1/x_1,\ldots,t_k/x_k]t_R\) with probabilities \(p\) and \(1-p\), respectively.
We often write
just \(F\,x_1\,\cdots\,x_k = t_L\)
for \(F\,x_1\,\cdots\,x_k = t_L\C{1}t_R\).

\begin{definition}[Operational Semantics and Termination Probability of \pHORS{}]
Given a \pHORS{} \(\GRAM=(\NONTERMS,\RULES,S)\), 
the rewriting relation \(\redp{\Dir,p}{\GRAM}\) (where \(\Dir\in \set{L,R}\) and \(p\in[0,1]\))
is defined by:
\iftwocol
\infrule{\RULES(F)=\lambda x_1.\cdots\lambda x_k.t_L\C{p}t_R}{F\,t_1\,\cdots\,t_k \redp{L,p}
  {\GRAM} [t_1/x_1,\ldots,t_k/x_k]t_{L}}

\infrule{\RULES(F)=\lambda x_1.\cdots\lambda x_k.t_L\C{p}t_R}
{F\,t_1\,\cdots\,t_k \redp{R,1-p}
  {\GRAM} [t_1/x_1,\ldots,t_k/x_k]t_{R}}
\else
$$
\infer{F\,t_1\,\cdots\,t_k \redp{L,p}
  {\GRAM} [t_1/x_1,\ldots,t_k/x_k]t_{L}}{\RULES(F)=\lambda x_1.\cdots\lambda x_k.t_L\C{p}t_R}
\qquad\qquad
\infer{F\,t_1\,\cdots\,t_k \redp{R,1-p}
  {\GRAM} [t_1/x_1,\ldots,t_k/x_k]t_{R}}{\RULES(F)=\lambda x_1.\cdots\lambda x_k.t_L\C{p}t_R}
$$
\fi
We write \(\redsp{\cs,p}{\GRAM}\) for the relational composition of
\(\redp{\Dir_1,p_1}{\GRAM},\ldots, \redp{\Dir_n,p_n}{\GRAM}\), when \(\cs=\Dir_1\cdots \Dir_n\) and 
\(p = \prod_{i=1}^np_i\).
Note that \(n\) may be \(0\), so that we have \(t_1\redsp{\epsilon,1}{\GRAM}t_2\) iff \(t_1=t_2\).
By definition, for each \(\cs\in\set{L,R}^*\),
there exists at most one \(p\) such that
\(S\redsp{\cs,p}{\GRAM}\ \! \Te\). For an applicative term $t$, we define
\(\Prob(\GRAM,t,\cs)\) by: \\[-2ex]
\[
\Prob(\GRAM,t,\cs)=\left\{\begin{array}{ll}
  p & \mbox{if $t\redsp{\cs,p}{\GRAM}\Te$}\\
  0 & \mbox{if $t\redsp{\cs,p}{\GRAM}\Te$ does not hold for any \(p\)}
\end{array}
\right..
\]
\iftwocol
We write \(\Prob(\GRAM,t)\) for
\(
\displaystyle\sum_{\cs\in\set{L,R}^*}\ \ \Prob(\GRAM,t,\cs)
\).
\noindent
Finally, we set \(\Prob(\GRAM) \,=\,\Prob(\GRAM,S)\) and call it the termination probability of \(\GRAM\).
\else
The \emph{partial} and \emph{full} \emph{termination probabilities}, written 
\(\Prob(\GRAM,t,n)\) and \(\Prob(\GRAM,t)\), are defined by:
\[
\Prob(\GRAM,t,n) = \displaystyle\sum_{\cs\in\set{L,R}^{\leq n}}\ \ \Prob(\GRAM,t,\cs)
\qquad \text{and}\qquad
\Prob(\GRAM,t) = \displaystyle\sum_{\cs\in\set{L,R}^*}\ \ \Prob(\GRAM,t,\cs).
\]
\noindent
Finally, we set \(\Prob(\GRAM,n) \,=\,\Prob(\GRAM,S,n)\) and \(\Prob(\GRAM) \,=\,\Prob(\GRAM,S)\).
\fi

\end{definition}
We often omit the subscript \(\GRAM\) below and just write
\(\redp{d,p}{}\) and \(\redsp{\cs,p}{}\) for
\(\redp{d,p}{\GRAM}\) and \(\redsp{\cs,p}{\GRAM}\) respectively.
\iftwocol\else
The \emph{termination probability of} $\GRAM$ refers to 
its full termination probability \(\Prob(\GRAM)\).
\fi

\begin{example}
\label{ex:random-walk}
Let \(\GRAM_1\) be the order-1 \pHORS{} \((\NONTERMS_1, \RULES_1, S)\),
where:
\[
\begin{array}{l}
\NONTERMS_1=\set{S\mapsto \T, F\mapsto \T\to\T}\\
\RULES_1 = \set{S\ =\ F\,\Te\C{1}\Omega,\quad
F\,x\ =\ x \C{p} F(F\,x)}.
\end{array}
\]
The start symbol \(S\) can be reduced, for example, as follows.
\[
S \redp{L,1}{} F(\Te) \redp{R,1-p}{} F(F\,\Te )\redp{L,p}{} F\,\Te\redp{L,p}{} \Te.
\]
Thus, we have \(S\redsp{LRLL, p^2(1-p)}\,\! \Te\).
As we will see in Section~\ref{sec:fixpoint}, the termination probability
\(\Prob(\GRAM_1)\) is the least solution for \(r\) of the fixpoint equation:
\( r = p + (1-p)r^2\).
Therefore, \(\Prob(\GRAM_1) = \frac{p}{1-p}\) if \(0\leq p<\frac{1}{2}\)
and \(\Prob(\GRAM_1) = 1\) if \(\frac{1}{2}\leq p\).
The corresponding example of a recursive Markov chain is shown in Figure~\ref{fig:rmc},
using the notational conventions from~\cite{Etessami09}. $\GRAM_1$ can be seen as realizing a binary, random walk on the natural
numbers, starting from $1$.
\qed
\end{example}
\begin{figure}
\begin{center}
\includegraphics[scale=1.05]{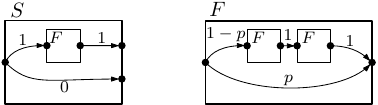}
\end{center}
\caption{A Recursive Markov Chain Modeling $\GRAM_1$.}\label{fig:rmc}
\end{figure}
As the previous example suggests, there is a mutual translation between recursive Markov chains
and order-1 \pHORSs{};
see the Appendix~\ref{sec:encoding-rmc} for details.
\begin{example}
\label{ex:order2-phors}
Let \(\GRAM_2\) be the order-2 \pHORS{} \((\NONTERMS_2,
\RULES_2,S)\)
where:
\begin{align*}
\NONTERMS_2=&\;\set{S\mapsto \T, H\mapsto \T\to\T, F\mapsto (\T\to\T)\to\T, \\
    &\hquad D\mapsto(\T\to\T)\to\T\to\T}\\
\RULES_2=&\;\set{S =\ (F\,H) \C{1}\Te,\ 
H\,x = x \C{\frac{1}{2}} \Omega,\\
    &\hquad F\,g = (g\,\Te) \C{\frac{1}{2}} (F(D\,g)),\ 
D\,g\,x = (g\,(g\,x))\C{1}\Omega}.\end{align*}
The start symbol \(S\) can be reduced, for example, as follows.
\begin{align*}
S &\redp{L,1}{} F\,H \redp{R,\frac{1}{2}}{} F(D\,H)
\redp{L,\frac{1}{2}}{} D\,H\,\Te
\redp{L,1}{} H(H\,\Te)\\
&\redp{L,\frac{1}{2}}{} H\,\Te
\redp{L,\frac{1}{2}}{}\Te.
\end{align*}
Contrary to $\GRAM_1$, it is quite hard to find an RMC which models the behavior
of $\GRAM_2$. In fact, this happens for very good reasons, as we will see in Section~\ref{sec:undecidability}.
\qed
\end{example}
The following result is obvious from the definition of \(\Prob(\GRAM)\).
\begin{theorem}
\label{th:lower-bound}
For any rational number
\(r\in[0,1]\), 
the set \(\set{\GRAM\mid \Prob(\GRAM)>r}\) is recursively enumerable.
\end{theorem}
\begin{proof}
This follows immediately from the facts that \(\Prob(\GRAM)>r\) if and only if
\(\Prob(\GRAM,n)=\sum_{\cs\in\set{L,R}^{\leq n}}\,\Prob(\GRAM,S,\cs)>r\) for some \(n\), and that \(\Prob(\GRAM,n)\) is computable.
\end{proof}
\noindent
In other words, 
whether \(\Prob(\GRAM)>r\) is semi-decidable, i.e., there exists a
procedure that eventually answers ``yes'' whenever \(\Prob(\GRAM)>r\).
As we will see in Section~\ref{sec:undecidability}, however, for every \(r\in (0,1]\),
\(\set{\GRAM\mid \Prob(\GRAM)\geq r}\) is \emph{not} recursively enumerable.

\begin{remark}
\label{rem:phors-vs-hors}
Given a \pHORS{} \(\GRAM\),
replacing each probabilistic operator \(\C{p}\) s.t. \(0<p<1\) with
a binary tree constructor \(\texttt{br}\) and replacing
\(t_L\C{1}t_R\) (\(t_L \C{0} t_R\), resp.) with \(t_L\) (\(t_R\), resp.),
we obtain an ordinary HORS \(\GRAM^{\mathit{ND}}\). 
Then \(\Prob(\GRAM)=0\) if and only if
the tree generated by \(\GRAM^{\mathit{ND}}\) has no finite path to \(\Te\).
Thus, by \cite{KO11LMCS} (see the paragraph below the proof of Theorem~4.5
about the complexity of the reachability problem), 
whether \(\Prob(\GRAM)=0\) is decidable, and \((n-1)\)-EXPTIME complete.
Note, on the other hand, that there is no clear correspondence between
the almost sure termination problem \(\Prob(\GRAM)\stackrel{?}{=}1\)
and a model checking problem for \(\GRAM^{\mathit{ND}}\).
If the tree of \(\GRAM^{\mathit{ND}}\) has neither \(\Omega\) nor an infinite path
(which is decidable), then
\(\Prob(\GRAM){=}1\), but the converse does not hold.
\end{remark}

\begin{remark}
The restriction that a probabilistic choice may occur only at the top-level of each rule 
is \emph{not} a genuine restriction. Indeed, whenever we wish to write a rule of the form
\(F\;\seq{x}= C[t_L\C{p}t_R]\), we can normalize it to 
\(F\;\seq{x}= C[G\;\seq{x}]\), where \(G\) is defined by \(G\;\seq{x}=t_L\C{p}t_R\).
Keeping this in mind, we sometimes allow probabilistic choices to occur inside terms.
In fact, a \pHORS{} can be considered as a term (of type \(\T\)) of 
a probabilistic extension of the (call-by-name) 
\lambdaY{}-calculus~\cite{DBLP:journals/apal/Statman04}. We define the set of probabilistic \lambdaY{} terms by:
\[ s ::= \Te \midd \Omega\midd x \midd \lambda x.s \midd s_1s_2 \midd Y(\lambda f.\lambda x.s)\midd s_1\C{p}s_2.\]
Here, \(\C{p}\) is a probabilistic choice operator of type \(\T\to\T\to\T\), and other terms are simply-typed
in the usual way. Then, \pHORSs{} and probabilistic \lambdaY{} terms can be converted to each other.
We use \pHORS{} in the present paper for the convenience of
 the fixpoint characterizations discussed in Section~\ref{sec:fixpoint}.
\end{remark}

\begin{remark}
We adopt the call-by-name semantics, and allow probabilistic choices only on terms of type \(\T\).
The call-by-value semantics, as well as probabilistic choices at higher-order types can be modeled by
applying a standard CPS transformation. Moreover, a \pHORS{} does not have data other than functions, but
as in ordinary HORS~\cite{Kobayashi13JACM}, 
elements of
a finite set (such as Booleans) can be modeled by using Church encoding.
\end{remark}

We provide a few more examples of \pHORS{} below.
\begin{example}
\label{ex:listgen}
Recall the list generator example in Section~\ref{sec:intro}, whose
 termination is equivalent to that of the following program,
obtained by replacing the output of each function with the unit value \texttt{()}.
\begin{center}
\begin{lstlisting}[basicstyle=\ttfamily]
  let boolgen() = flip() in
  let rec listgen f ()= 
    if flip() then () else (f(); listgen f ()) 
  in listgen (listgen boolgen) ()
\end{lstlisting}
\end{center}
With a kind of CPS transformation, termination of the above program is reduced to that
of the following \pHORS{} \(\GRAM_3\):
\begin{align*}
S &= \listgen\;(\listgen\,\boolgen)\;\Te\\
 \boolgen\;k &= k \\ \listgen\;f\;k &= k\C{\frac{1}{2}}f(\listgen\,f\,k)
\end{align*}
It is not difficult to confirm that \(\Prob(\GRAM_3)=1\) (using the fixpoint characterization 
given in Section~\ref{sec:fixpoint}).
\end{example}

\begin{example}
\label{ex:listgen-variant}
The following is a variation of the list generator example (Example~\ref{ex:listgen}), 
which generates ternary trees instead of lists:
\begin{lstlisting}[basicstyle=\ttfamily]
 let boolgen() = flip() in
 let rec treegen f = 
    if flip() then Leaf 
    else Node(f(), treegen f, treegen f, treegen f) in
 treegen(boolgen)
\end{lstlisting}
The following \pHORS{} \(\GRAM_4\) captures the termination probability of the aforementioned program:
\begin{align*}
S &= \treegen\;\boolgen\;\Te\\
\boolgen\;k &= k \\ \treegen\;f\;k &= k\C{\frac{1}{2}}(f(\treegen\,f\,(\treegen\,f\,(\treegen\,f\,k))))
\end{align*}
Interestingly, \(\GRAM_4\) is \emph{not} almost surely terminating, since \(\Prob(\GRAM_4) = \frac{\sqrt{5}-1}{2}\).

To increase the chance of termination, let us change the original program as follows:
\begin{lstlisting}[basicstyle=\ttfamily,escapechar=!]
 let boolgen() = flip() in
 let rec treegen p f = 
   if flipp(p) then Leaf 
   else Node(f(), treegen !\(\frac{\texttt{p}+1}{2}\)! f, treegen !\(\frac{\texttt{p}+1}{2}\)! f, treegen !\(\frac{\texttt{p}+1}{2}\)! f) in
 treegen !\(\frac{1}{2}\)! boolgen
\end{lstlisting}
where \texttt{flipp} is the natural generalization of \texttt{flip}.
Here, \texttt{treegen} is parameterized with probability \texttt{p}, which is increased upon each recursive call.
We assume that \texttt{flipp(p)} returns \(\TRUE\) with probability \texttt{p} and \(\FALSE\) with \(1-\texttt{p}\).
The corresponding \pHORS{} \(\GRAM_5\) is:
\begin{align*}
S &= \treegen\;H\;\boolgen\;\Te\\
\boolgen\;k &= k\\ H\;x\;y &= x\C{\frac{1}{2}}y\\
G\;p\;x\;y &= x\C{\frac{1}{2}}(p\;x\;y)\\
\treegen\;p\;f\;k &= p\;k\; (f(\treegen\,(G\;p)\,f\,(\treegen\,(G\;p)\,f\,(\treegen\,(G\;p)\,f\,k))))
\end{align*}
The function \(\treegen\) is parameterized by a probabilistic choice function \(p\), which is initially set to
the function \(H\) (that chooses the first argument with probability \(\frac{1}{2}\)). The function
\(G\) takes a probabilistic choice function \(p\), and returns a probabilistic function
\(\lambda x.\lambda y.x\C{\frac{1}{2}}(p\;x\;y)\), which chooses the first argument with 
probability \(\frac{\texttt{p}+1}{2}\) where \(\texttt{p}\) is the probability
that \(p\) chooses the first argument. As expected, \(\GRAM_5\) is almost surely terminating. \qed
\end{example}

\begin{example}
\label{ex:listgen-even}
Recall the list generator example again. Suppose that we wish to compute the probability
that \texttt{listgen(boolgen)} generates a list of even length. It can be reduced to
the problem of computing the termination probability of the following program:

\begin{lstlisting}[basicstyle=\ttfamily,escapechar=!]
 let boolgen() = flip() in
 let rec loop() = loop() in
 let rec listgenE f () = 
   if flip() then () else (f();listgenO f ())
 and listgenO f () = 
   if flip() then loop() else (f();listgenE f ()) in
 listgenE boolgen ()
\end{lstlisting}
Here, we have duplicated \texttt{listgen} to \texttt{listgenE} and 
\texttt{listgenO}, which are expected to simulate the generation of even and odd lists respectively.
Thus, the then-branches of \texttt{listgenE} and \texttt{listgenO} have been replaced by
termination and divergence respectively.
As in the previous example, the above program can further be translated to the following \pHORS{}
\(\GRAM_6\):
\begin{align*}
S &= \listgene\;\boolgen\;\Te\\
\boolgen\;k &= k\\ \listgene\;f\;k &= k\C{\frac{1}{2}}(f(\listgeno\,f\,k))\\
\listgeno\;f\;k &= \Omega\C{\frac{1}{2}}(f(\listgene\,f\,k)).
\end{align*}
The termination probability of the \pHORS{} is
\[ \frac{1}{2} + \frac{1}{2}\cdot \frac{1}{4} + \frac{1}{2}\cdot \left(\frac{1}{4}\right)^2 + \cdots =
\frac{1}{2}\cdot\sum_{i=0}^\infty\frac{1}{4^i}=\frac{2}{3}.\]
Thus, the probability that the original program generates an even list is also \(\frac{2}{3}\).

Let us also consider the problem of computing the probability that 
\texttt{listgen(boolgen)} generates a list containing an even number of \(\TRUE\)'s.
It can be reduced to the termination probability of the following \pHORS{}.
\begin{align*}
S &= \listgene\;\boolgen\;\Te\\
\boolgen\;k_1\;k_2 &= k_1\C{\frac{1}{2}}k_2\\
\listgene\;f\;k &= k\C{\frac{1}{2}}(f(\listgeno\,f\,k)(\listgene\,f\,k))\\
\listgeno\;f\;k &= \Omega\C{\frac{1}{2}}(f(\listgene\,f\,k)(\listgeno\,f\,k)).
\end{align*}
The function \(\boolgen\) now takes two continuations \(k_1\) and \(k_2\) as arguments,
and calls \(k_1\) or \(k_2\) according to whether \(\TRUE\) or \(\FALSE\) is generated in the original
program. The function \(\listgene\) (\(\listgeno\), resp.) is called when
the number of \(\TRUE\)'s generated so far is even (odd, resp.).
The termination probability of the \pHORS{} above is \(\frac{3}{4}\).
\qed
\end{example}

In the following example, a standard program transformation for
randomized algorithms is captured as a \pHORS{}. More specifically, a higher-order function
is defined, which turns any Las-Vegas
algorithm that sometimes declares not to be able to provide the correct answer into
one that \emph{always} produces the correct answer. (For more details about the
use of the scheme above, please refer to~\cite{DBLP:series/txtcs/Hromkovic05}).
\begin{example}
\label{ex:palgo}
Consider a probabilistic function \(f\), which takes a value of type \(A\), and returns a value of type
\(B\) with probability \(p\) and \texttt{Unknown} with probability \(1-p\),
where \(p>0\).
The following higher-order function \texttt{determinize} takes such a function \(f\) as an argument, 
and generates a function from \(A\) to \(B\).
\begin{lstlisting}[basicstyle=\ttfamily,escapechar=!]
 type 'b pans = Ans of 'b | Unknown
 let rec determinize(f:'a->'b pans)(x:'a)=
    match f x with
        Ans(r) -> r
      | Unknown -> determinize f x
\end{lstlisting}
To confirm that \(\texttt{determinize}\ f\) almost surely terminates and returns a value of type \(B\),
it suffices to check that the \pHORS{} term \(\deter\; g\) almost surely terminates for 
\(g=\lambda y.\lambda z.y\C{p}z\), where \(\deter\) is defined by:
\[
\deter\;g = g\;\Te\;(\deter\;g).
\]
Here, the first argument of \(g\) corresponds to the body of the clause \(\texttt{Ans(r)->}\cdots\),
while the second argument corresponds to that of the clause \(\texttt{Unknown->}\cdots\).
Almost sure termination of \(\deter(\lambda y.\lambda z.y\C{p}z)\) for any \(p>0\) can further
be encoded as that of the following \pHORS{} \(\GRAM_7\):
\begin{align*}
S &= (\deter\;\one) \C{\frac{1}{2}} (\forallp\;\zero\;\one)\\
\one\;y\;z &= y\\
\zero\;y\;z &= z\\
\Avg\;p\;q\;y\;z &= (p\;y\;z)\C{\frac{1}{2}}(q\;y\;z)\\
\forallp\;p\;q &= (\deter\;(\Avg\;p\;q))\C{\frac{1}{2}}\\
  &\qquad((\forallp\;p\;(\Avg\;p\;q))\C{\frac{1}{2}}(\forallp\;(\Avg\;p\;q)\;q))
\end{align*}
It runs \(\deter\;(\lambda y.\lambda z.y\C{p}z)\) for every \(p\; (0<p\leq 1)\) of the form
\(\frac{k}{2^n}\) with non-zero probability. Thus, \(\Prob(\GRAM_7)=1\) 
if \(\deter(\lambda y.\lambda z.y\C{p}z)\) almost surely terminates for every \(p>0\).
Conversely, by the continuity of the termination probability of 
\(\deter\;(\lambda y.\lambda z.y\C{p}z)\) except at \(p=0\) (which we omit to discuss formally),
 \(\Prob(\GRAM_7)=1\)  implies that
 \(\deter(\lambda y.\lambda z.y\C{p}z)\) almost surely terminates for every \(p>0\). 
 \qed
\end{example}

\begin{remark}
Although \pHORS{} do not have probabilities as first-class values, as demonstrated in
the examples above, certain operations on probabilities can be realized by encoding 
a probability \(p\) into a probabilistic function \(\lambda x.\lambda y.x\C{p}y\).
The function \(\Avg\) in Example~\ref{ex:palgo} realizes the average operation \(\frac{p_1+p_2}{2}\).
The multiplication \(p_1p_2\) can be represented by \(\textit{Mult}\;p_1\;p_2\), where
\(\textit{Mult}\;p_1\;p_2\;x\,y = p_1\;(p_2\;x\;y)\;y\).
\end{remark}

\section{Undecidability of Almost Sure Termination of Order-2 \pHORS{}}
\label{sec:undecidability}
\label{SEC:UNDECIDABILITY}
\newcommand\CheckLt{\mathit{CheckLt}}
\newcommand\CheckLtP{\mathit{CheckLtPr}}
\newcommand\NatToP{\mathit{NatToPr}}
\newcommand\tonat[1]{[#1]}
\newcommand\CheckHalf{\mathit{CheckHalf}}

We prove in this section that the almost sure termination problem,
i.e., whether 
the termination probability \(\Prob(\GRAM)\) of a given \pHORS{} \(\GRAM\) is \(1\),
is undecidable
even for order-2 \pHORSs.
The proof is by reduction from the undecidability of
Hilbert's 10th problem~\cite{Diophantine}
(i.e. unsolvability of Diophantine equations).
Note that almost sure termination of an order-1 \pHORS{} is decidable,
as order-1 \pHORSs{} are essentially equi-expressive with
 probabilistic pushdown
systems and recursive Markov chains~\cite{Etessami09,DBLP:journals/jacm/EtessamiY15,DBLP:journals/fmsd/BrazdilEKK13,DBLP:journals/jcss/BrazdilBFK14}.
In fact, by the fixpoint characterization given in Section~\ref{sec:order-n-1-equation}, the termination probability of an order-1 \pHORS{}
can be expressed as the least solution of fixpoint equations over reals,
which can be solved as discussed in \cite{Etessami09}.
Thus, our undecidability result for order-2 \pHORSs{} is optimal.

We start by giving an easy reformulation of the unsolvability of
Diophantine equations in terms of polynomials with \emph{non-negative}
coefficients, which follows immediately from the original result.
\begin{lemma}
\label{lem:diophantine}
Given two polynomials \(P(x_1,\ldots,x_k)\) and
\(Q(x_1,\ldots,x_k)\) with non-negative integer coefficients, whether
\(P(x_1,\ldots,x_k)<Q(x_1,\ldots,x_k)\) for some \(x_1,\ldots,x_k\in \Nat\)
is undecidable. More precisely,
the set of pairs of polynomials:
\(\set{(P(x_1,\ldots,x_k),Q(x_1,\ldots,x_k)) \mid
\exists x_1,\ldots,x_k\in \Nat.P(x_1,\ldots,x_k)< Q(x_1,\ldots,x_k)}\)
is \(\Sigma^0_1\)-complete in the arithmetical hierarchy.
\end{lemma}
\begin{proof}
  Let \(D(x_1,\ldots,x_k)\) be a multivariate polynomial with integer coefficients.
  Then, for all natural numbers \(x_1,\ldots,x_k\in \Nat\),
  \(D(x_1,\ldots,x_k)=0\) if and only if \((D(x_1,\ldots,x_k))^2-1<0\).
Any such polynomial \((D(x_1,\ldots,x_k))^2-1\) may
be rewritten as \(P(x_1,\ldots,x_k)-Q(x_1,\ldots,x_k)\), where
\(P(x_1,\ldots,x_k)\) and
\(Q(x_1,\ldots,x_k)\) have only non-negative integer coefficients.
Then, \(D(x_1,\ldots,x_k)=0\) 
if and only if \(P(x_1,\ldots,x_k)<Q(x_1,\ldots,x_k)\).
Since whether \(D(x_1,\ldots,x_k)=0\) for some 
\(x_1,\ldots,x_k\in \Nat\) is
undecidable~\cite{Diophantine}, it is also undecidable whether
\(P(x_1,\ldots,x_k)<Q(x_1,\ldots,x_k)\) for some 
\(x_1,\ldots,x_k\in \Nat\).
Furthermore, since the set of sastisfiable Diophantine equations
is \(\Sigma^0_1\)-complete, the set 
\(\set{(P(x_1,\ldots,x_k),Q(x_1,\ldots,x_k)) \mid
\exists x_1,\ldots,x_k\in \Nat.P(x_1,\ldots,x_k)< Q(x_1,\ldots,x_k)}\) is 
\(\Sigma^0_1\)-hard. The set is also obviously 
recursively enumerable, hence belongs to \(\Sigma^0_1\).
\end{proof}

Roughly, the idea of our undecidability proof is to show that for every
$P$ and $Q$ as above, one can effectively construct
an order-$2$ \pHORS{} that
\emph{does not} almost surely terminate if and only if \(P(x_1,\ldots,x_k)<Q(x_1,\ldots,x_k)\)
\emph{for some} $x_1,\ldots,x_k$.
Henceforth, we say \(t\) is \emph{non-AST} if
\(t\) is not almost surely terminating.
For ease of understanding, we first construct
an order-\(3\) \pHORS{} \(\GRAMAST{3}\) that satisfies the property above
in Section~\ref{sec:order-3-undecidability}
and then refine the construction to obtain an order-\(2\)
\pHORS{} \(\GRAMAST{2}\) with the same property
in Section~\ref{sec:order-2-undecidability}.
\subsection{Construction of the Order-3 \pHORS{} \(\GRAMAST{3}\)}
\label{sec:order-3-undecidability}

\newcommand\TestAll{\mathit{Loop}}
\newcommand\TestAllP{\mathit{LoopPr}}
\newcommand\TestLt{\mathit{Lt}}
\newcommand\TestLtP{\mathit{LtPr}}
Let \(P(x_1,\ldots,x_k)\) and \(Q(x_1,\ldots,x_k)\) be, as above, polynomials with
non-negative coefficients. We give the construction of \(\GRAMAST{3}\)
in a top-down manner.
We let \(\GRAMAST{3}\) enumerate all the tuples of natural numbers
\((n_1,\ldots,n_k)\), and for each tuple,
spawn a process
\(\TestLt\,(P(n_1,\ldots,n_k))\,(Q(n_1,\ldots,n_k))\) with non-zero probability,
where \(\TestLt\,m_1\,m_2\) is a process that
is non-AST if and only if
\(m_1<m_2\). Thus, we define the start
symbol \(S\) of \(\GRAMAST{3}\)  by:
\newcommand\Zero{Zero}
\newcommand\Succ{\mathit{Succ}}
\newcommand\Add{\mathit{Add}}
\newcommand\Mult{\mathit{Mult}}
\newcommand\OneP{\mathit{OnePr}}
\newcommand\ZeroP{\mathit{ZeroPr}}
\newcommand\SuccP{\mathit{SuccPr}}
\newcommand\AddP{\mathit{AddPr}}
\begin{align*}
S =&\;\TestAll\;\Zero\;\cdots\; \Zero.\\
\TestAll\,x_1\,\cdots\,x_k =&\;
(\TestLt\,(P\,x_1\,\cdots\,x_k)\,(Q\,x_1\,\cdots\,x_k)) 
\iftwocol
\\\qquad 
\else
\\\qquad \qquad \qquad \qquad 
\fi
&\vspace{-20pt}\C{\frac{1}{2}} (\TestAll\,(\Succ\,x_1)\,\cdots\,x_k)
\C{\frac{1}{2}}\cdots 
\iftwocol \\\qquad \else\fi
&\vspace{-20pt}\C{\frac{1}{2}} (\TestAll\,x_1\,\cdots\,(\Succ\,x_k)).
\end{align*}
\iftwocol
Here, for readability, we have extended the righthand sides of rules to
$n$-ary probabilistic choices: 
\(t_1\C{p_1} t_2\C{p_2}\cdots \C{p_{n-1}}t_n\) (where \(\C{p}\) is left-associative).
\else
Here, for readability, we have extended the righthand sides of rules to
$n$-ary probabilistic choices:
\[t_1\C{p_1} t_2\C{p_2}\cdots \C{p_{n-1}}t_n.\]
These can be expressed as \(t_1\C{p_1} (F_2\,x_1\,\cdots\,x_k)\),
where auxiliary non-terminals are defined by:
\[F_2\,x_1\,\cdots\,x_k = t_2\C{p_2}(F_3\,x_1\,\cdots\,x_k)\qquad
\cdots\qquad F_{n-1}\,x_1\,\cdots\,x_k = t_{n-1}\C{p_{n-1}}t_n.\]
\fi
We can express natural numbers and operations on them
by using Church encoding:
\newcommand{\CT}{\mathsf{CT}}
\begin{align*}
\Zero\;s\;z&=z  &
\Succ\;n\;s\;z&=s\;(n\;s\;z)\\
\Add\;n\;m\;s\;z &=n\;s\;(m\;s\;z) & \Mult\;n\;m\;s\;z&=n\;(m\;s)\;z.
\end{align*}
Here, the types of non-terminals above are given by:
\begin{align*}
\NONTERMS(\Zero)&=\CT\\
\qquad\qquad\NONTERMS(\Succ)&=\CT\to\CT\\
\NONTERMS(\Add)=\NONTERMS(\Mult)&=\CT\to\CT\to\CT,
\end{align*}
where \(\CT = (\T\to\T)\to\T\to\T\) is the usual type of Church numerals.
Note that the order of \(\CT\) is \(2\), while
that of \(\NONTERMS(\Succ)\), \(\NONTERMS(\Add)\), and \(\NONTERMS(\Mult)\) is \(3\).
By using the just introduced operators, we can easily define \(P\) and \(Q\) as order-3 non-terminals.
By abuse of notation, we often
use symbols \(P\) and \(Q\) to denote
both polynomials and the representations of them as non-terminals;
similarly for natural numbers.

It remains to define an order-3 non-terminal \(\TestLt\), so that
\(\TestLt\,m_1\,m_2\) is non-AST
if and only if \(m_1<m_2\). Since
\(\GRAMAST{3}\) runs
\(\TestLt\,(P\,n_1\,\cdots\,n_k)\,(Q\,n_1\,\cdots\,n_k)\)
for each tuple of Church numerals \((n_1,\ldots,n_k)\) with non-zero
probability,
\(\GRAMAST{3}\) is non-AST if and only if
\(P(n_1,\ldots,n_k)<Q(n_1,\ldots,n_k)\) for \emph{some} natural numbers
\(n_1,\ldots,n_k\).
The key ingredient used for the construction of \(\TestLt\) is the function
\(\CheckHalf\) of type \((\T\to\T\to\T)\to\T\), defined as follows:
\[
\CheckHalf\; g = F'\; g\; \Te\qquad
F'\; g\; x = g\; x\; (F'\,g\,(F'\,g\,x)).\]
Here, \(F'\) above is a parameterized version of \(F\) from Example~\ref{ex:random-walk}: \(F'\; \C{p}\) (where \(\C{p}\) is treated as a function of type
\(\T\to\T\to\T\), which chooses the first argument with probability \(p\)
and the second one with \(1-p\)) corresponds to \(F\).
As discussed in Example~\ref{ex:random-walk},
\(F\,\Te\) is non-AST if and only if \(p<\frac{1}{2}\).
Thus, \(\CheckHalf\;g= F'\;g\;\Te\)
(which is equivalent to \(F\,\Te\) when \(g=\C{p}\)) is non-AST if and only if 
 the probability that \(g\) chooses the first argument is smaller than \(\frac{1}{2}\).
Let \(\CheckLt\) (which will be defined shortly)
be a function which takes Church numerals \(m_1\)
and \(m_2\), and returns a function of type \(\T\to\T\to\T\) that chooses
the first argument with probability smaller than \(\frac{1}{2}\)
if and only if \(m_1<m_2\).
Then, \(\TestLt\) can be defined as:
\[
\TestLt\;m_1\;m_2 = \CheckHalf (\CheckLt\;m_1\;m_2).
\]
Finally, \(\CheckLt\) can be defined by:
\begin{align*}
\CheckLt\;m_1\;m_2\;x\;y =
(\NatToP\;m_1\;x\;y)\C{\frac{1}{2}}(\NatToP\;m_2\;y\;x).\\
\NatToP\;m\;x\;y = m\;(H\;x)\;y.\qquad
H\;x\;y = x\C{\frac{1}{2}}y.
\end{align*}
Let us write \(\tonat{m}\) for the natural number represented
by a Church numeral \(m\).
For a Church numeral \(m\),
\(\NatToP\;m\;x\;y\) (which is equivalent to
\((H\,x)^{\tonat{m}}y\)) chooses \(x\) with probability \(1-\frac{1}{2^{\tonat{m}}}\)
and \(y\) with probability
\(\frac{1}{2^{\tonat{m}}}\).
Thus, the probability that \(\CheckLt\;m_1\;m_2\;x\;y\) chooses \(x\) is
\[
\frac{1}{2}\cdot\left(1-\frac{1}{2^{\tonat{m_1}}}\right)+
\frac{1}{2}\cdot\frac{1}{2^{\tonat{m_2}}}
= \frac{1}{2} +\frac{1}{2}\cdot
\left(\frac{1}{2^{\tonat{m_2}}}- \frac{1}{2^{\tonat{m_1}}}\right),
\]
which is smaller than \(\frac{1}{2}\) if and only if \(\tonat{m_1}<\tonat{m_2}\),
as required.
This completes the construction of \(\GRAMAST{3}\).
See Figure~\ref{fig:g3} for the whole rules of \(\GRAMAST{3}\).
From the discussion above, it should be trivial that
\(\GRAMAST{3}\) is non-AST if and only if
\(P(x_1,\ldots,x_k)<Q(x_1,\ldots,x_k)\) holds for
some \(x_1,\ldots,x_k\in\Nat\).

\begin{figure}
\fbox{
\begin{minipage}{.97\textwidth}
\begin{align*}
S &= \TestAll\;\Zero\;\cdots\; \Zero.\\
\TestAll\,x_1\,\cdots\,x_k &=
(\TestLt\,(P\,x_1\,\cdots\,x_k)\,(Q\,x_1\,\cdots\,x_k)) \\
&\C{\frac{1}{2}} (TestAll\,(\Succ\,x_1)\,\cdots\,x_k)
\C{\frac{1}{2}}\cdots \C{\frac{1}{2}} (TestAll\,x_1\,\cdots\,(\Succ\,x_k)).\\
\TestLt\;m_1\;m_2 &= \CheckHalf (\CheckLt\;m_1\;m_2).\\
\CheckHalf\; y &= F'\; y\; \Te.\\
F'\; g\; x &= g\; x\; (F'\,g\,(F'\,g\,x)).\\
\CheckLt\;m_1\;m_2\;x\;y &=
(\NatToP\;m_1\;x\;y)\C{\frac{1}{2}}(\NatToP\;m_2\;y\;x).\\
\NatToP\;m\;x\;y &= m\;(H\;x)\;y.\\
H\;x\;y &= x\C{\frac{1}{2}}y.\\
\Zero\,s\,z &= z.\\
\Succ\,n\,s\,z &= s\,(n\,s\,z).\\
\Add\,n\,m\,s\,z &= n\,s\,(m\,s\,z).\\
\Mult\,n\,m\,s\,z&=n\,(m\,s)\,z.\\
P\,x_1\,\cdots\,x_k &= t_P.\\
Q\,x_1\,\cdots\,x_k &= t_Q.
\end{align*}
\end{minipage}}
\caption{The rules of \(\GRAMAST{3}\), where \(t_P\) and \(t_Q\) are terms encoding the
polynomials \(P\) and \(Q\) by way of \(\Zero\), \(\Succ\), \(\Add\), and \(\Mult\).}
\label{fig:g3}
\end{figure}

\subsection{Decreasing the Order}
\label{sec:order-2-undecidability}
We now refine the construction of \(\GRAMAST{3}\) to obtain
an order-2 \pHORS{} \(\GRAMAST{2}\) that satisfies the same property.
The idea is, instead of passing around a Church numeral \(m\),
to pass a probabilistic function equivalent to \(\NatToP\,m\),
which takes two arguments and chooses the first and second arguments
with probabilities \(1-\frac{1}{2^{\tonat{m}}}\) and \(\frac{1}{2^{\tonat{m}}}\),
respectively.
Note that a Church numeral \(m\) has an order-2 type
\(\CT=(\T\to\T)\to\T\to\T\), whereas
\(\NatToP\,m\) has an order-1 type \(\T\to\T\to\T\). This 
ultimately allows us to decrease the order of the \pHORS{}.

Based on the idea above, we replace \(\TestLt\) with \(\TestLtP\),
which now takes probabilistic functions of type \(\T\to\T\to\T\)
as arguments:
\begin{align*}
 \TestLtP\;g_1\;g_2 &= \CheckHalf (\CheckLtP\;g_1\;g_2).\\
 \CheckLtP\;g_1\;g_2\;x\;y &= (g_1\;x\;y)\C{\frac{1}{2}}(g_2\;y\;x).
\end{align*}
Here, \(\CheckLtP\) is an analogous version of \(\CheckLt\),
and \(\CheckHalf\) is as before: \(\CheckHalf\;g\) is non-AST if and only if
the probability that \(g\) chooses the first argument is smaller than \(\frac{1}{2}\).
Then, \(\TestLtP\;(\NatToP\;(P\,n_1\,\cdots\,n_k))\;(\NatToP\;(Q\,n_1\,\cdots\,n_k))\)
is non-AST if and only if \(P(n_1,\ldots,n_k)<Q(n_1, \ldots,n_k)\).

It remains to modify the top-level loop \(\TestAll\), so that
we can enumerate (terms equivalent to)
\(\TestLtP\;(\NatToP\;(P\,n_1\,\cdots\,n_k))\;(\NatToP\;(Q\,n_1\,\cdots\,n_k))\)
for all \(n_1,\ldots,n_k\in\Nat\),
without explicitly constructing Church numerals.
Instead of using Church encodings, we can encode
natural numbers and operations on them (except multiplication)
into probabilistic functions as follows.
\begin{align*}
\ZeroP\;x\;y&=y  &\SuccP\;g\;x\;y&=x\C{\frac{1}{2}}(g\;x\;y)\\
\OneP\;x\;y&=x\C{\frac{1}{2}}y  &
\AddP\;g_1\;g_2\;x\;y &=g_1\;x\;(g_2\;x\;y).
\end{align*}
Basically, a natural number \(m\) is encoded as a probabilistic function
of type \(\T\to\T\to\T\), which chooses the first and second arguments
with probabilities \(1-\frac{1}{2^m}\) and \(\frac{1}{2^m}\) respectively.
Notice that
\(\AddP\;(\NatToP\;m_1)\;(\NatToP\;m_2)\) is equivalent to
\(\NatToP\;(\Add\;m_1\;m_2)\), because
the probability that
\(\AddP\;(\NatToP\;m_1)\;(\NatToP\;m_2)\;x\;y\) chooses \(y\)
is \(\frac{1}{2^{[m_1]}} \cdot \frac{1}{2^{[m_2]}} =
\frac{1}{2^{[m_1]+[m_2]}}\).
We call this encoding the \emph{probabilistic function encoding}, or
\emph{PF encoding} for short.

The multiplication cannot, however, be directly encoded. To compensate for
the lack of the multiplication operator, instead of passing around
just \(n_1,\ldots,n_k\) in the top-level loop,
we pass around the PF encodings of the values of
\(n_1^{i_1}\cdots n_k^{i_k}\) for each \(i_1\leq d_1,\ldots,i_k\leq d_k\),
where 
\(d_1,\ldots,d_k\) respectively are the largest degrees of 
\(P(x_1,\ldots,x_k)+Q(x_1,\ldots,x_k)\) in \(x_1,\ldots,x_k\).
We thus define the start symbol \(S\) of \(\GRAMAST{2}\) by:
\newcommand\Inc{\mathit{Inc}}
\begin{align*}
S &= \TestAllP\;\OneP\;\underbrace{\ZeroP\,\cdots\,\ZeroP}_{(d_1+1)\cdots(d_k+1)-1\mbox{ times}}.\\
\TestAllP\;\seq{x} &= (\TestLtP\;(P'\;\seq{x})\;\;(Q'\;\seq{x}))\\\ 
&\C{\frac{1}{2}} (\TestAllP\,(\Inc_{1,(0,\ldots,0)}\,\seq{x})\,\cdots\,
(\Inc_{1,(d_1,\ldots,d_k)}\,\seq{x}))\C{\frac{1}{2}}\cdots\\\ 
&\C{\frac{1}{2}} (\TestAllP\,(\Inc_{k,(0,\ldots,0)}\,\seq{x})\,\cdots\,
(\Inc_{k,(d_1,\ldots,d_k)}\,\seq{x})).
\end{align*}
Here, \(\seq{x}\) denotes the sequence of \((d_1+1)\cdots(d_k+1)\) variables
 \(x_{(0,\ldots,0)},\ldots,x_{(d_1,\ldots,d_k)}\),
 consisting of \(x_{(i_1,\ldots,i_k)}\) for each \(i_1\in\set{0,\ldots,d_1},
 \ldots, i_k\in\set{0,\ldots,d_k}\).
 Each variable \(x_{(i_1,\ldots,i_k)}\) holds (the PF encoding of)
 the value of  \(n_1^{i_1}\cdots n_k^{i_k}\).

Moreover, the functions \(P'\) and \(Q'\) are the PF encodings of the polynomials
\(P\) and \(Q\). Since \(P\) and \(Q\) can be represented as linear combinations of 
monomials \(x_1^{i_1}\cdots x_k^{i_k}\) for \(i_1\leq d_1,\ldots,i_k\leq d_k\),
\(P'\) and \(Q'\) can be defined using \(\ZeroP\) and \(\AddP\).
For example, if \(P(x_1,x_2) = x_1^2+2x_1x_2\), then \(P'\) is defined by:
\(P'\;\seq{x}\;y\;z = \AddP\;x_{(2,0)}\;(\AddP\;x_{(1,1)}\;x_{(1,1)})\;y\;z\).

The function \(\Inc_{j,(i_1,\ldots,i_k)}\,\seq{x}\) represents the
PF encoding of \(n_1^{i_1}\cdots (n_j+1)^{i_j}\cdots n_k^{i_k}\),
assuming that \(\seq{x}\) represents (the PF encoding of)
the values \(n_1^{0}\cdots n_k^{0}, \ldots, n_1^{d_1}\cdots n_k^{d_k}\).
Note that \(\Inc_{j,(i_1,\ldots,i_k)}\) can also be defined by using \(\ZeroP\) and \(\AddP\),
since \(x_1^{i_1}\cdots (x_j+1)^{i_j}\cdots x_k^{i_k}\)
can be expressed as a linear combination of
monomials \(x_1^{0}\cdots x_k^{0}, \ldots,
x_1^{d_1}\cdots x_k^{d_k}\). For example, if \(k=2\), then \(\Inc_{2,(1,2)}\) can be
defined by
\( \Inc_{2,(1,2)}\,\seq{x}\,y\,z =\AddP\;x_{(1,2)}\;(\AddP\;x_{(1,1)}\;(\AddP\;x_{(1,1)}\;x_{(1,0)}))\,y\,z\),
because \(x_1(x_2+1)^2 = x_1x_2^2 + 2x_1x_2+x_1\). \qed

This completes the construction of \(\GRAMAST{2}\).
See Figure~\ref{fig:g2} for the list of all rules of \(\GRAMAST{2}\).
\begin{figure}[tbp]
\fbox{
\begin{minipage}{.97\textwidth}
\begin{align*}
S &= \TestAllP\;\OneP\;\underbrace{\ZeroP\,\cdots\,\ZeroP}_{(d_1+1)\cdots(d_k+1)-1}.\\
\TestAllP\;\seq{x} &= (\TestLtP\;(P'\;\seq{x})\;\;(Q'\;\seq{x})) \\& \C{\frac{1}{2}} (\TestAllP\,(\Inc_{1,(0,\ldots,0)}\,\seq{x})\,\cdots\,
(\Inc_{1,(d_1,\ldots,d_k)}\,\seq{x})) \\
&\C{\frac{1}{2}}\cdots \C{\frac{1}{2}} (\TestAllP\,(\Inc_{k,(0,\ldots,0)}\,\seq{x})\,\cdots\,
(\Inc_{k,(d_1,\ldots,d_k)}\,\seq{x})).\\
 \TestLtP\;g_1\;g_2 &= \CheckHalf (\CheckLtP\;g_1\;g_2).\\
 \CheckLtP\;g_1\;g_2\;x\;y &= (g_1\;x\;y)\C{\frac{1}{2}}(g_2\;y\;x).\\
\CheckHalf\; y &= F'\; y\; \Te.\\
F'\; g\; x &= g\; x\; (F'\,g\,(F'\,g\,x)).\\
\ZeroP\;x\;y&=y.\\
\OneP\;x\;y&=x\C{\frac{1}{2}}y.\\
\SuccP\;g\;x\;y&=x\C{\frac{1}{2}}(g\;x\;y).\\
\AddP\;g_1\;g_2\;x\;y&=g_1\;x\;(g_2\;x\;y).\\
P'\,\seq{x}&=t'_P\\ Q'\,\seq{x}&=t'_Q\\   \Inc_{j,(i_1,\ldots,i_k)}&= t_\Inc^{j(i_1,\ldots,i_k)} \end{align*}
\end{minipage}}
  \caption{The rules of \(\GRAMAST{2}\), where the terms $t'_P$, $t'_Q$ and $t_\Inc^{j,(i_1,\ldots,i_k)}$ are defined based
  on \(\ZeroP\), \(\OneP\), and \(\SuccP\) and \(\AddP\).}
\label{fig:g2}
\end{figure}
By the discussion above, we have:

\begin{theorem}
\label{th:undecidability-ast}
The almost sure termination
of order-2 \pHORS{} is undecidable.
More precisely, the set \(\set{\GRAM\mid \Prob(\GRAM)=1, \mbox{$\GRAM$
is an order-2 \pHORS{}}}\) is \(\Pi^0_1\)-hard.
\end{theorem}
\begin{proof}
By the construction of \(\GRAMAST{2}\) above, \(\Prob(\GRAMAST{2})=1\) if and only if
\(P(x_1,\ldots,x_k)\geq Q(x_1,\ldots,x_k)\) holds for all \(x_1,\ldots,x_k\in\Nat\).
By Lemma~\ref{lem:diophantine},
the set of pairs \((P,Q)\) that satisfy the latter is \(\Pi^0_1\)-complete,
hence the set
\(\set{\GRAM\mid \Prob(\GRAM)=1, \mbox{$\GRAM$
is an order-2 \pHORS{}}}\) is \(\Pi^0_1\)-hard. \hfill
\end{proof}
As a corollary, we also have:
\begin{theorem}
\label{th:undecidability}
For any rational number \(r\in (0,1]\), the followings are undecidable:
\begin{enumerate}
\item whether a given order-2 \pHORS{} \(\GRAM\) satisfies 
\(\Pr(\GRAM){\geq}r\).
\item whether a given order-2 \pHORS{} \(\GRAM\) satisfies 
\(\Pr(\GRAM){=}r\).
\end{enumerate}
More precisely, the sets
\(\set{\GRAM\in \PHORSSET_2 \mid \Pr(\GRAM){\geq}r}\) and
\(\set{\GRAM\in\PHORSSET_2 \mid \Pr(\GRAM){=}r}\) are \(\Pi^0_1\)-hard.
\end{theorem}
\begin{proof}
Let \(\GRAM\) be an order-2 \pHORS{} with the start symbol \(S\).
Define \(\GRAM'\) as the \pHORS{} obtained by replacing the start symbol with
\(S'\) and adding the rules \(S'= S\C{r}\Omega\). Then \(\Pr(\GRAM')\geq r\) if and only if
\(\Pr(\GRAM')= r\) if and only if \(\Pr(\GRAM)= 1\).
Thus, the result follows from Theorem~\ref{th:undecidability-ast}. \hfill
\end{proof}

\iftwocol
\else
\fi

\begin{remark}
\label{rem:approximation}
Let us write \(\SETofPHORS_{\sim r}\) for the set of order-2 \pHORS{} \(\GRAM\) such that
\(\Pr(\GRAM)\sim r\) where \(\sim \in\set{<,\leq, =, \geq, >}\).
By Theorem~\ref{th:lower-bound} and Theorem~\ref{th:undecidability-ast}, we have: \begin{enumerate}[(i)]
\item For any rational number \(r\in [0,1]\), \(\SETofPHORS_{>r}\) is recursively enumerable
(or, belongs to \(\Sigma^0_1\)).
\item For any rational number \(r\in (0,1]\), \(\SETofPHORS_{\ge r}\) is \(\Pi^0_1\)-hard
(whereas \(\SETofPHORS_{\ge 0}\) is obviously recursive).
\item For any rational number \(r\in (0,1]\), \(\SETofPHORS_{=r}\) is \(\Pi^0_1\)-hard
(whereas \(\SETofPHORS_{=0}\) is recursive; recall Remark~\ref{rem:phors-vs-hors}).
\end{enumerate}
It is open whether the following propositions hold or not.\footnote{This open question has recently been settled by ChatGPT; see Appendix~\ref{sec:chatgpt}.}
\begin{enumerate}[(i)]
\setcounter{enumi}{3}
\item \(\SETofPHORS_{<r}\) is recursively enumerable for every rational number \(r\).
\item \(\SETofPHORS_{\le r}\) is recursively enumerable for every rational number \(r\).
\item There exists an algorithm that takes an order-2 \pHORS{} \(\GRAM\) and 
a rational number \(\epsilon>0\) as inputs,
and returns a rational number \(r\) such that \(|\Pr(\GRAM)-r|<\epsilon\).
\end{enumerate}
\iffull
Statements (iv) and (vi) are equivalent. In fact, if (iv) is true, 
we can construct an algorithm for (vi) as follows. First, test whether \(\Pr(\GRAM)=0\) (which is decidable). 
If so,
output \(r=0\). Otherwise, pick a natural number \(m\) such that 
\(\frac{1}{m}<\frac{1}{2}\epsilon\), and 
divide the interval \((0,1+\frac{1}{2}\epsilon)\) to
\(m\) (overlapping) intervals 
\[
\begin{array}{l}
\left(0,\frac{1}{m}+\frac{1}{2}\epsilon\right),
\left(\frac{1}{m}, \frac{2}{m}+\frac{1}{2}\epsilon\right),\ldots,
\left(\frac{m-2}{m}, \frac{m-1}{m}+\frac{1}{2}\epsilon\right),
\iftwocol\\\fi
\left(\frac{m-1}{m}, 1+\frac{1}{2}\epsilon\right).
\end{array}
\]
By using procedures for (i) and (iv), 
one can enumerate all the order-2 \pHORSs{} whose termination probabilities
 belong to each interval.
Thus, \(\GRAM\) is eventually enumerated for one of the intervals 
\((\frac{i}{m},\frac{i+1}{m}+\frac{1}{2}\epsilon)\); one can then output \(\frac{i}{m}\) as \(r\).
Conversely, suppose that we have an algorithm for (vi). For each order-2 \pHORS{} \(\GRAM\),
repeatedly run the algorithm for 
\(\epsilon = \frac{1}{2}, \frac{1}{4},\frac{1}{8},\ldots\), and output \(\GRAM\) if the output \(r'\) for
\((\GRAM,\epsilon)\) satisfies \(r'+\epsilon<r\). Then, \(\GRAM\) is eventually output just if
\(\Pr(\GRAM)<r\) (note that if \(\Pr(\GRAM)<r\), then 
 \(\epsilon\) eventually becomes smaller than \(\frac{1}{2}(r-\Pr(\GRAM))\); at that point, the output \(r'\)
satisfies 
\(r'+\epsilon < (\Pr(\GRAM)+\epsilon)+\epsilon < r\)).

Proposition (v) implies (iv) (and hence also (vi)). If there is a procedure for (v), one can enumerate
all the elements of \(\SETofPHORS_{<r}\) by running the procedure for enumerating 
\(\SETofPHORS_{\le r-\epsilon}\) for \(\epsilon=\frac{1}{2}, \frac{1}{4},\frac{1}{8},\ldots\)
\else
Propositions (iv) and (vi) are equivalent, and
Proposition (v) implies (iv) (and hence also (vi))~\cite{KDG19LICSfull}.
\fi
\hfill \qed
\end{remark}

\begin{remark}
Table~\ref{tab:complexity} summarizes the hardness of termination problems in terms of
the arithmetical hierarchy for 
recursive Markov chains (RMC), \pHORS{}, and a probabilistic language whose underlying
(non-probabilistic) language is Turing-complete. The results for
RMC and the Turing-complete language come from \cite{Etessami09} and \cite{Kaminski18}.
As seen in the table, the results on \pHORS{} are not tight, except for the problem \(\Prob(\GRAM)>0\). Since the expressive power of \pHORS{} is between those of RMC and the Turing complete language,
the hardness of each problem is between those of the two models.
Theorem~\ref{th:undecidability} shows \(\Sigma^0_1\)-hardness of \(\Prob(\GRAM)<r\), but
we do not know yet whether the problem is \(\Sigma^0_1\)-complete or \(\Sigma^0_2\)-complete,
or lies between the two classes.\footnote{This has been settled; see the previous footnote.}
\end{remark}
\begin{table}
\caption{Hardness of the termination problems in terms of the arithmetical hierarchy.
For recursive sets (i.e. those in \(\Delta^0_1=\Sigma^0_1\cap \Pi^0_1\)), 
more precise computational complexities of the membership problems are given.
``\pHORS{}'' means order-\(k\) \pHORS{} where \(k\ge 2\).}
\label{tab:complexity}
\begin{tabular}{|l|l|l|l|l|}
\hline
\multicolumn{2}{|l|}{Models}  & \(\SETofPHORS_{>0}\) & \(\SETofPHORS_{>r}\) 
(\(r\in (0,1)\))& \(\SETofPHORS_{<r}\) (\(r\in (0,1]\))\\\hline
\multicolumn{2}{|l|}{RMC}  & P  & PSPACE & PSPACE\\
\hline
\pHORS{} & \textsl{Hardness} &
\((k-1)\)-EXPTIME &
\((k-1)\)-EXPTIME  & \( \Sigma^0_1\)\\ \cline{2-5}&\textsl{Containment}  & 
\((k-1)\)-EXPTIME & \( \Sigma^0_1 \) & \( \Sigma^0_2\)\\
\hline
\multicolumn{2}{|l|}{Turing-complete language}  & \( \Sigma^0_1\)-complete  & \( \Sigma^0_1\)-complete  & \( \Sigma^0_2\)-complete \\
\hline
\end{tabular}
\end{table}

\begin{remark}
Theorem~\ref{th:undecidability-ast} implies that,
in contrast to the decidability of LTL model checking of recursive Markov
chains~\cite{DBLP:journals/fmsd/BrazdilEKK13,DBLP:journals/tocl/EtessamiY12}, 
the corresponding problem for order-2 \pHORS{} 
(of
computing the probability that an infinite transition sequence satisfies a given LTL property)
is undecidable
and there are even no precise approximation algorithms. Let us extend terms with events:
\[ t ::= \cdots \midd \evexp{a}t \]
where \(\evexp{a}{t}\) raises an event \(a\) and evaluates \(t\).
Consider the problem of, given an order-2 \pHORS{} \(\GRAM\),
computing the probability \(\Prob_{a^\omega}(\GRAM)\) that
\(a\) occurs infinitely often. Then there is no algorithm to
compute \(\Prob_{a^\omega}(\GRAM)\) with arbitrary precision, in
the sense of (vi) of Remark~\ref{rem:approximation}. To see this,
notice that
by parametric \(\GRAMAST{2}\) with \(\Te\),
we can define a nonterminal
\(F\COL\T\to\T\) such that \(F\;x\) almost surely reduces to \(x\) if and only if
there exist no \(n_1,\ldots,n_k\) such that \(P(n_1,\ldots,n_k)<Q(n_1,\ldots,n_k)\).
Consider the (extended) \pHORS{} \(\GRAM^{P,Q,a^\omega}\) whose start symbol
\(S\) is defined by \(S = \evexp{a}{F(S)}\). Then
\iftwocol
\[
\Prob_{a^\omega}(\GRAM^{P,Q,a^\omega})=\left\{\begin{array}{ll}
0 & \mbox{if there exists \(n_1,\ldots,n_k\)} \\ & \mbox{ such that \(P(n_1,\ldots,n_k)<Q(n_1,\ldots,n_k)\)}\\
1 & \mbox{otherwise}.
\end{array}
\right.
\]
\else
\[
\Prob_{a^\omega}(\GRAM^{P,Q,a^\omega})=\left\{\begin{array}{ll}
0 & \mbox{if there exists \(n_1,\ldots,n_k\)}  \mbox{ such that \(P(n_1,\ldots,n_k)<Q(n_1,\ldots,n_k)\)}\\
1 & \mbox{otherwise}.
\end{array}
\right.
\]
\fi
Thus, there is no algorithm to approximately
compute \(\Prob_{a^\omega}(\GRAM)\) even within the precision of
\(\epsilon=\frac{1}{2}\).
\hfill \qed
\end{remark}

\begin{remark}
The \pHORS{} \(\GRAMAST{2}\) obtained above satisfies the so called ``safety'' restriction~\cite{Knapik01TLCA,Salvati15OI}. Thus, based on the correspondence between safe grammars and pushdown systems~\cite{Knapik01TLCA}, the undecidability result above would also hold for probabilistic second-order pushdown systems
(without collapse operations~\cite{Hague08LICS}).
\end{remark}

\section{Fixpoint Characterization of Termination Probability}
\label{sec:fixpoint}
\label{SEC:FIXPOINT}
Although, as observed in the previous section, there is no general algorithm for
exactly computing the termination probability of \pHORS{}, 
there is still hope that we can \emph{approximately} compute the termination probability.
As a possible route towards this goal,
this section shows that the termination probability of any \pHORS{} $\GRAM$ can be characterized as
the least solution of fixpoint equations on higher-order
functions over \([0,1]\). As mentioned in Section~\ref{sec:intro},
the fixpoint characterization immediately yields a procedure for computing
lower-bounds of termination probabilities, and also serves as a justification
for the method for computing upper-bounds discussed in Section~\ref{sec:upperbound}.
We first introduce higher-order fixpoint equations in Section~\ref{sec:ho-fixpoint}.
We then characterize the termination probability
of an order-\(n\)
\pHORS{} in terms of fixpoint equations on order-\(n\)
functions over \([0,1]\) (Section \ref{sec:order-n-equation}), and
then improve the result by characterizing
the same probability in terms of order-(\(n-1\)) fixpoint equations
for the case \(n\geq 1\) (Section~\ref{sec:order-n-1-equation}).
The latter characterization can be seen as a generalization of the characterization
of termination probabilities of recursive Markov chains as polynomial equations~\cite{Etessami09},
which served as a key step in the analysis of recursive Markov chains (or probabilistic pushdown 
systems)~\cite{Etessami09,DBLP:journals/jacm/EtessamiY15,DBLP:journals/fmsd/BrazdilEKK13,DBLP:journals/jcss/BrazdilBFK14}.

\subsection{Higher-order Fixpoint Equations}
\label{sec:ho-fixpoint}
We define the syntax and semantics of fixpoint equations that are commonly
used in Sections~\ref{sec:order-n-equation} and \ref{sec:order-n-1-equation}.
We first define the syntax of fixpoint equations.
\iftwocol
A system \(\E\) of fixpoint equations  is a set of 
function definitions of the form 
\(f\,(\seq{x}_{1})\,\cdots\, (\seq{x}_{\ell}) = e\),
where the syntax of expressions \(e\) is given by
\(
  e  ::= r \midd x \midd f \midd e_1+e_2 \midd e_1\cdot e_2 \midd e_1e_2
\).
\else
\[
\begin{array}{l}
  \E \mbox{ (equations)} ::= \set{f_1\,(\seq{x}_{1,1})\,\cdots\, (\seq{x}_{1,\ell_1}) = e_1,
    \ldots, f_m\,(\seq{x}_{m,1})\,\cdots\, (\seq{x}_{m,\ell_m}) = e_m};\\
  e \mbox{ (expressions)} ::= r \midd x \midd f \midd e_1+e_2 \midd e_1\cdot e_2 \midd e_1e_2
  \mid (e_1,\ldots,e_k).
\end{array}
\]
\fi
Here,
\(r\) ranges over the set of real numbers in \([0,1]\),
and \((\seq{x})\) represents a tuple of variables \((x_1,\ldots,x_k)\).
In the set \(\E\) of equations,
we require that each function symbol occurs at most once on the lefthand side.
The expression \(e_1\cdot e_2\) represents the multiplication
of the values of \(e_1\) and \(e_2\), whereas \(e_1e_2\) represents a function application;
however, we sometimes omit \(\cdot\) when there is no confusion (e.g.,
we write \(0.5x\) for \(0.5\cdot x\)).
\iftwocol
Expressions must be well-typed under a simple type system;
as it is standard, we defer it to \iffull
Appendix~\ref{app:sec4}.
\else
the longer version~\cite{KDG19LICSfull}.
\fi
\else
Expressions must be well-typed under the type system
given in Figure~\ref{fig:typing-eq}. 
\fi
The \emph{order of a system of
  fixpoint equations}
\(\E\) is the largest order of the types of functions in \(\E\),
where the order of the type \(\realt\) of reals is \(0\),
and the order of a function type is defined analogously to the order of types for \pHORS{}
in Section~\ref{sec:problem}.

\begin{example}
\label{ex:eq}
The following is a system of order-2 fixpoint equations:
\[
\begin{array}{c}
\set{ f_1 = f_2\;f_3\;(0.5,0.5), 
f_2\;g\;(x_1,x_2) = g(x_1+x_2),
f_3\,x = 0.3\,x+0.7f_3(f_3\,x)}.
\end{array}
\]
It is well-typed under \(f_1\COL\realt,
f_2\COL(\realt\to\realt)\to(\realt\times\realt)\to\realt,
f_3\COL\realt\to\realt\). \hfill\qed
\end{example}

\iftwocol
\else

\iftwocol
\begin{figure*}
\else
\begin{figure}
\fi
\fbox{
\begin{minipage}{0.97\textwidth}
\small
\[ \tau \mbox{ (types)}::= \realt \mid \tau_1\to\tau_2 \mid \tau_1\times \cdots \times \tau_n.\]
\vspace{3pt}
$$
\infer{\Gamma\p r: \realt}{r\in \Reals}
\qquad\qquad
\infer{\Gamma\p e_1+e_2:\realt}{\Gamma\p e_1:\realt\andalso \Gamma\p e_2:\realt}
\qquad\qquad
\infer{\Gamma\p e_1\cdot e_2:\realt}{\Gamma\p e_1:\realt\andalso \Gamma\p e_2:\realt}
$$
\vspace{3pt}
$$
\infer{\Gamma\p x\COL\tau}{\Gamma(x)=\tau}
\qquad\quad
\infer{\Gamma\p e_1e_2:\tau}{\Gamma\p e_1:\tau_2\to\tau\andalso \Gamma\p e_2:\tau_2}
\qquad\quad
\infer{\Gamma\p (e_1,\ldots,e_k):\tau_1\times \cdots \times \tau_k}{\Gamma\p e_i:\tau_i\mbox{ for each $i\in\set{1,\ldots,k}$}}
$$          
\vspace{3pt}      
$$
\infer{\Gamma\p \set{f_i\,(\seq{x}_{i,1})\,\cdots\,(\seq{x}_{i,\ell_i})=e_i
                    \mid i\in\set{1,\ldots,m}}}
      {\Gamma, (\seq{x}_{i,1})\COL\tau_{i,1},\ldots,(\seq{x}_{i,\ell_i})\COL\tau_{i,\ell_ii}
                      \p e_i
& \Gamma(f_i)=\tau_{i,1}\to\cdots\to\tau_{i,\ell_i}\to\realt
\mbox{ (for each $i\in\set{1,\ldots,m}$)}}
$$                
\end{minipage}}
\caption{Type system for fixpoint equations, where \((x_1,\ldots,x_k)\COL \tau\) denotes
\(x_1\COL\tau_1,\ldots,x_k\COL\tau_k\) whenever \(\tau=\tau_1\times\cdots \times\tau_k\).}
\label{fig:typing-eq}                
\iftwocol
\end{figure*}
\else
\end{figure}
\fi
 \fi
The semantics of fixpoint equations is defined in an obvious manner.
Let \(\realp\) be the set consisting of non-negative real numbers and \(\infty\).
We extend addition and multiplication by:
\(x+\infty=\infty+x=\infty\),  \(0\cdot \infty=\infty\cdot 0=0\), and
 \(x\cdot \infty=\infty\cdot x=\infty\) if \(x\neq 0\).
Note that \((\realp, \leq, 0)\) forms an \(\omega\)-cpo,
where \(\leq\) is the extension of the usual inequality on
reals with \(x\leq \infty\) for every \(x\in\realp\).
For each type \(\tau\), we interpret \(\tau\) as
the cpo \(\sem{\tau}=(X_\tau, \Leq_\tau, \bot_\tau)\), defined by
induction on \(\tau\): \begin{align*}
X_\realt &= \realp\\
\Leq_\realt &= \leq\\
\bot_\realt &= 0\\
X_{\tau_1\to\tau_2} &= \set{f\in X_{\tau_1}\to X_{\tau_2} \mid \mbox{ \(f\) is monotonic and \(\omega\)-continuous}}\\
 \Leq_{\tau_1\to\tau_2}&=\set{(f_1,f_2)\in X_{\tau_1\to\tau_2}\times
 X_{\tau_1\to\tau_2} \mid \forall x\in X_{\tau_1}. f_1(x)\Leq_{\tau_2} f_2(x)}\\\
 \bot_{\tau_1\to\tau_2} &= \metalambda x\in X_{\tau_1}.\bot_{\tau_2} \\ X_{\tau_1\times\cdots\times \tau_k} &=X_{\tau_1}\times \cdots \times X_{\tau_k}\\
\Leq_{\tau_1\times\cdots\times \tau_k} &=
\set{((x_1,\ldots,x_k),(y_1,\ldots,y_k))\mid x_i\Leq_{\tau_i}y_i
\mbox{ for each $i\in\set{1,\ldots,k}$}}\\
\bot_{\tau_1\times\cdots\times \tau_k} &= (\bot_{\tau_1},\ldots,\bot_{\tau_k}).
\end{align*}
By abuse of notation, we often write \(\sem{\tau}\) also for \(X_\tau\).
We also often omit the subscript \(\tau\) and just write
\(\Leq\) and \(\bot\) for \(\Leq_\tau\) and \(\bot_\tau\) respectively.
The interpretation of base type \(\realt\) can actually be restricted to
\([0,1]\), but for technical convenience (to make the existence of
a fixpoint trivial) we have defined \(X_{\realt}\) as \(\realp\).

For a type environment \(\Gamma\), we write \(\sem{\Gamma}\) for
the set of functions that map each \(x\in\dom(\Gamma)\) to an element of
\(\sem{\Gamma(x)}\). Given \(\rho\in\sem{\Gamma}\) and 
\(e\) such that \(\Gamma\p e:\tau\), its semantics
\(\seme{e}{\rho}\in\sem{\tau}\) is defined by: \begin{align*}
\seme{r}{\rho}&=r\\
\seme{x}{\rho}&=\rho(x)\\
\seme{f}{\rho}&=\rho(f)\\
\seme{e_1+e_2}{\rho}&=\seme{e_1}{\rho}+\seme{e_2}{\rho}\\
\seme{e_1\cdot e_2}{\rho}&=\seme{e_1}{\rho}\cdot \seme{e_2}{\rho}\\
\seme{e_1e_2}{\rho}
&=(\seme{e_1}{\rho})(\seme{e_2}{\rho})\\
\seme{(e_1,\ldots,e_k)}{\rho}
&=(\seme{e_1}{\rho},\ldots,\seme{e_k}{\rho}).
\end{align*}
Given \(\E\) such that \(\Gamma\p \E\),
we write \(\rho_{\E}\) for the least solution of \(\E\),
i.e., the least \(\rho\in\sem{\Gamma}\) such that
\(\seme{f(\seq{x}_1)\cdots (\seq{x}_\ell)}{\rho\set{\seq{x}_{1}\mapsto \seq{\bf y}_{1},\ldots,
  \seq{x}_{\ell}\mapsto \seq{\bf y}_{\ell}}}=\seme{e}{\rho\set{\seq{x}_{1}\mapsto \seq{\bf y}_{1},\ldots,
  \seq{x}_{\ell}\mapsto \seq{\bf y}_{\ell}}}\) for
  every equation \(f(\seq{x}_1)\cdots (\seq{x}_\ell)=e\in\E\) and
\(  (\seq{\bf y}_1)\in \sem{\tau_1},\ldots,
(\seq{\bf y}_\ell)\in \sem{\tau_\ell}\)
with \(\Gamma(f)=\tau_1\to\cdots\to\tau_\ell\to\realt\).
Note that \(\rho_{\E}\) always exists, and is given by:
\(
\rho_{\E} = \LFP (\F_\E) = \bigsqcup_{i\in \omega}\F_\E^i(\bot_{\sem{\Gamma}})\),
where \(\F_\E\in \sem{\Gamma}\to \sem{\Gamma}\) is defined as the map
such that
\[\F_\E(\rho)(f) = 
  \metalambda (\seq{\bf y}_{1})\in
  \sem{\tau_1}.\ldots\metalambda(\seq{\bf y}_{\ell})\in\sem{\tau_\ell}.
  \seme{e}{\rho\set{\seq{x}_{1}\mapsto \seq{\bf y}_{1},\ldots,
  \seq{x}_{\ell}\mapsto \seq{\bf y}_{\ell}}}
  \]
  for each \(f(\seq{x}_1)\cdots (\seq{x}_\ell)=e\in \E\) with \(\Gamma(f)=\tau_1\to\cdots\to\tau_\ell\to\realt\).
Note that \(\F_\E\) is continuous in the \(\omega\)-cpo \(\sem{\Gamma}\).

\begin{example}
Let \(\E\) be the system of equations in Example~\ref{ex:eq}.
Then, \(\rho_{\E}\) is:
\[
\left\{f_1\mapsto \frac{3}{7}, f_2\mapsto \metalambda g\in \R\to\R.\metalambda (x_1,x_2)\in\R\times\R.g(x_1+x_2), f_3\mapsto \metalambda x\in\R.\frac{3}{7}x\right\}.
\quad\qed\]
\end{example}  
\subsection{Order-\(n\) Fixpoint Characterization}
\label{sec:order-n-equation}
We now give a translation from an order-\(n\) \pHORS{} \(\GRAM\)
to a system of order-\(n\) fixpoint equations \(\E\), so that \(\Prob(\GRAM, S)=\rho_{\E}(S)\).
The translation is actually straightforward: we just need to replace
\(\Te\) and \(\Omega\) with the termination probabilities \(1\) and \(0\), and 
probabilistic choices with summation and multiplication of probabilities.
The translation function \(\TtoP{(\cdot)}\) is defined by:
\begin{align*}
\TtoP{(\NONTERMS,\RULES,S)} &= (\TtoP{\RULES}, S)\\
\RULES^\# &= \set{F\,\seq{x}=p\cdot (t_L)^\# +(1-p)\cdot (t_R)^\# \mid \RULES(F)=\lambda\seq{x}.t_L\C{p}t_R}\\
\Te^\# &= 1 \qquad \Omega^\# = 0\qquad
x^\# = x \qquad (st)^\#=s^\# t^\#.
\end{align*}
We write \(\E_\GRAM\) for \(\RULES^\#\). 
We define the translation of types and type environments by:
\begin{align*}
  \T^\# &= \realt\\
  (\sty_1\to\sty_2)^\# &= \sty_1^\#\to \sty_2^\#\\
  (x_1\COL\sty_1,\ldots,x_n\COL\sty_n)^\# &=
   x_1\COL\sty_1^\#,\ldots, x_n\COL\sty_n^\#.
\end{align*}
\iftwocol\else
The following lemma states that the output of the translation is
well-typed. 
\begin{lemma}
  \label{lem:tr-n-wf}
  Let \(\GRAM=(\NONTERMS,\RULES,S)\) be an order-\(n\) \pHORS.
  Then \(\NONTERMS^\# \p \E_{\GRAM}\) and \(\NONTERMS^\#\p S:\realt\).
\end{lemma}
By the above lemma and the definition of the translation
of type environments, it follows that  for an order-\(n\) \pHORS{} \(\GRAM\),
the order of \(\E_{\GRAM}\) is also \(n\).
\fi
The following theorem states the correctness of the translation
(see \iffull Appendix~\ref{sec:proof-tr-n-1} \else \cite{KDG19LICSfull} \fi for a proof).

\begin{theorem}
\label{prop:order-k-fixpoint}
\label{PROP:ORDER-K-FIXPOINT}
  Let \(\GRAM\) be an order-\(n\) \pHORS.
  Then \(\Prob(\GRAM) = \rho_{\E_{\GRAM}}(S)\).
\end{theorem}

\begin{example}
\label{ex:order-n-equation}
Recall \(\GRAM_1=(\NONTERMS_1, \RULES_1, S)\) from Example~\ref{ex:random-walk}:
\begin{align*}
\NONTERMS_1&=\set{S\mapsto \T, F\mapsto \T\to\T};\\
\RULES_1&=\set{S\ =\ F\,\Te\C{1}\Omega,\quad
F\,x\ =\ x \C{p} F(F\,x)}.
\end{align*}
\(\NONTERMS_1^\# = \set{S\mapsto \R, F\mapsto \R\to\R}\), and
\(\E_{\GRAM_1}\) consists of: \(
S = 1\cdot F(1)\) and \( F\,x = p\cdot x +(1-p)\cdot F(F\,x)\).
The least solution \(\rho_{\E_{\GRAM_1}}\) is
\iftwocol
\(S=\frac{p}{1-p}\) and \(F=\metalambda{\bf x}.
\frac{p}{1-p}\cdot{\bf x}\) if \(0\leq p<\frac{1}{2}\),
and \(S=1\) and \(F=\metalambda{\bf x}.{\bf x}\)
if \(\frac{1}{2}\leq p\leq 1\).
\else
\[
\begin{array}{l}
S = \left\{\begin{array}{ll}
 \frac{p}{1-p} & \mbox{if $0\leq p<\frac{1}{2}$}\\
 1 & \mbox{if $\frac{1}{2}\leq p\leq 1$}
\end{array}\right.
\qquad
F=\metalambda{\bf x}.
\left\{\begin{array}{ll}
 \frac{p}{1-p}\cdot{\bf x} & \mbox{if $0\leq p<\frac{1}{2}$}\\
 {\bf x} & \mbox{if $\frac{1}{2}\leq p\leq 1$.}
\end{array}\right.
\end{array}
\]
\fi
\hfill\qed
\end{example}

\begin{example}
\label{ex:equation-listgen}
Recall \(\GRAM_3\) from Example~\ref{ex:listgen}:
\[
\begin{array}{l}
S = \listgen\;(\listgen\,\boolgen)\;\Te\quad
 \boolgen\;k = k\quad
\listgen\;f\;k = k\C{\frac{1}{2}}(f(\listgen\,f\,k)).
\end{array}
\]
The corresponding fixpoint equations are:
\begin{align*}
S &= \listgen\;(\listgen\,\boolgen)\;1\\
 \boolgen\;k &= k\\
\listgen\;f\;k &= \frac{1}{2}k+\frac{1}{2}(f(\listgen\,f\,k)).
\end{align*}
By specializing \(\listgen\) for the cases \(f=\listgen\,\boolgen\)
and \(f=\boolgen\), we obtain:
\begin{align*}
S &= \listlistgen\;1\\
\boolgen\;k &= k\\
\listlistgen\;k &= \frac{1}{2}k+\frac{1}{2}(\listboolgen(\listlistgen\,k))\\
\listboolgen\;k &= \frac{1}{2}k+\frac{1}{2}(\boolgen(\listboolgen\,k)).
\end{align*}
The least solution is:
\[ S=1\quad \boolgen\;k=\listlistgen\;k=\listboolgen\;k=k.\]
\qed
\end{example}

\subsection{Order-(\(n-1\)) Fixpoint Characterization}
\label{sec:order-n-1-equation}
\label{SEC:ORDER-N-1-EQUATION}
We now characterize the termination probability of order-\(n\) \pHORS{} (where
\(n>0\)) in terms of
order-(\(n-1\)) equations, so that the fixpoint equations are easier
to solve.  When \(n=1\), the characterization yields polynomial
equations on probabilities; thus the result below may be considered as a
generalization of the now classic result on the reachability problem for recursive
Markov chains~\cite{Etessami09}.

The basic observation (that is also behind the fixpoint characterization for
recursive Markov chains~\cite{Etessami09}) is that
the termination behavior of an order-1 function of type
\(\T^\ell\to\T\) can be represented by a tuple of \emph{probabilities}
\((p_0,p_1,\ldots,p_\ell)\), where (i) \(p_0\) is the probability that
the function terminates \emph{without} using any of its arguments, and (ii) \(p_i\)
is the probability that the function uses the \(i\)-th argument.
To see why, consider a term \(f\,t_1\,\cdots\,t_\ell\) of type \(\T\),
where \(f\) is  an order-1 function of type \(\T^\ell\to\T\).
In order for \(f\,t_1\,\cdots\,t_\ell\) to terminate, the only possibilities
are: (i) \(f\) terminates without calling any of the arguments,
or (ii) \(f\) calls \(t_i\) for some \(i\in\set{1,\ldots,\ell}\), and
\(t_i\) terminates (notice, in this case, that none of the other \(t_j\)'s
are called: since \(t_i\) is of type \(\T\), once
\(t_i\) is called from \(f\), the control cannot go back to \(f\)). Thus, the probability that \(f\,t_1\,\cdots\,t_\ell\) terminates
can be calculated by \(p_0+p_1q_1+\cdots p_\ell q_\ell\), where
each \(q_i\) denotes the probability that \(t_i\) terminates. 
The termination probability is, therefore, independent of
the precise internal behavior of \(f\); only \((p_0,p_1,\ldots,p_\ell)\)
matters. Thus, information about an order-1 function can be represented as
a tuple of real numbers, which is order 0. By generalizing this observation,
we can represent information about an order-\(n\) function
as an order-(\(n-1\)) function on (tuples of) real numbers. Since the general
translation is quite subtle and requires a further insight, however, 
let us first confirm the above idea by revisiting Example~\ref{ex:random-walk}.

\begin{example}
\label{ex:random-walk-eq}
Recall \(\GRAM_1\) from Example~\ref{ex:random-walk},
consisting of:
\(S = F\;\Te\) and 
\(F\,x = x \C{p} F(F\,x)\).
Here, we have two functions: \(S\) of type \(\T\) and \(F\) of type \(\T\to\T\).
Based on the observation above, their behaviors can be represented by
\(S_0\) and \((F_0,F_1)\) respectively, where \(S_0\) (\(F_0\), resp.)
 denotes the probability
 that \(S\) (\(F\), resp.) terminates, and 
 \(F_1\) represents the probability that \(F\) uses the argument.
 Those values are obtained as the least solutions for the following
system of equations.
\begin{align*}
S_0 &= F_0 + F_1\cdot 1\\
F_0 &= p\cdot 0 + (1-p)(F_0 + F_1\cdot F_0)\\
F_1 &= p\cdot 1 + (1-p)(F_1\cdot F_1\cdot 1).
\end{align*}
To understand the last equation, note that the possibilities that
\(x\) is used are: (i) \(F\) chooses the left branch (with probability \(p\)) and
then uses \(x\) with probability \(1\), or (ii) \(F\) chooses the right branch
(with probability \(1-p\)),
 the outer call of \(F\) uses the argument \(F\,x\) (with probability \(F_1\)),
and the inner call of \(F\) uses the argument \(x\).
By simplifying the equations, we obtain:
\begin{align*}
S_0 &= F_0+F_1\\
F_0 &= (1-p)(F_0+F_1F_0) \\
F_1 &= p+(1-p)F_1^2.
\end{align*}
The least solution is the following:
\[
F_0=0\qquad
S_0=F_1 = \left\{\begin{array}{ll} \frac{p}{1-p} & \mbox{if $0\leq p<\frac{1}{2}$}\\
1 & \mbox{if $\frac{1}{2}\leq p\leq 1$.}
\end{array}\right.
\]
\qed
\end{example}

The translation for general orders is more involved.
For technical convenience in formalizing the translation,
we assume below that the rules of \pHORS{} do not contain \(\Te\);
instead, the start symbol \(S\) (which is now a non-terminal
of type \(\T\to\T\)) takes \(\Te\) from the environment.
Thus, the termination probability we consider is \(\Prob(\GRAM,S\,\Te)\),
where \(\Te\) does not occur in \(\RULES\). This is without any loss of generality,
since \(\Te\) can be passed around as an argument without increasing the order
of the underlying \pHORS{}, if it is higher than $0$.

To see how we can generalize the idea above to deal with higher-order functions,
let us now consider the following example of an order-2 \pHORS{}:
\begin{align*}
S\, x &= F\,(H\,x)\,x \\
F\,f\,y &= f(f\,y)\\
H\,x\,y &= x\C{\frac{1}{2}}(y\C{\frac{1}{2}}\Omega).
\end{align*}
Suppose we wish to characterize the termination probability of \(S\;\Te\), i.e.,
the probability that \(S\) uses the first argument. (In this particular case,
one can easily compute the termination probability by unfolding all the functions,
but we wish to find a compositional translation which works well in presence of
recursion.) We need to compute the probability that \(F\;(H\;x)\;x\) reaches (i.e.,
reduces to) \(x\), which is the probability \(p_1\) that \(F\;(H\;x^{(1)})\;x^{(2)}\) reaches \(x^{(1)}\), 
plus the probability \(p_2\) that \(F\;(H\;x^{(1)})\;x^{(2)}\) reaches \(x^{(2)}\);
we have added annotations to distinguish between the two occurrences of \(x\). 
 What information on \(F\) is required for computing it? To compute \(p_2\), we need to obtain the
probability that \(F\) uses the formal argument \(y\). Since it depends on \(f\), we represent
it as a function \(F_1\) defined 
\iftwocol
by
\( F_1\; f_1 = f_1\cdot f_1\).
\else
by:
\[ F_1\; f_1 = f_1\cdot f_1.\]
\fi
Here, \(f_1\) represents the probability that the original argument \(f\) uses its first argument.
We can thus represent \(p_2\) as \(F_1(H_2)\), where
\(H_2\) is \(\frac{1}{4}\), the probability that
\(f=H\,x\) uses the first argument,
i.e., the probability that
\(H\) uses the second argument.
Now let us consider how to represent \(p_1\), the probability that \(F\,(H\,x^{(1)})\,x^{(2)}\) 
reaches \(x^{(1)}\). We construct another function \(F_0\) from the definition of \(F\) for this purpose.
A challenge is that the variable \(x\) is not visible in (the definition of) \(F\); only
the caller of \(F\) knows the reachability target \(x\).
Thus, we pass to \(F_0\), in addition to \(f_1\) above, 
another argument \(f_0\), which represents the probability that the argument \(f\) 
reaches the current target (which is \(x\) in this case). 
Therefore, \(p_1\) is represented as \(F_0\,(H_1,H_2)\), where
\iftwocol
\( F_0\,(f_0,f_1) = f_0+f_1\cdot f_0\) and \(H_1=\frac{1}{2}\).
\else
\[ F_0\,(f_0,f_1) = f_0+f_1\cdot f_0 \qquad H_1=\frac{1}{2}.\]
\fi
In \(f_0+f_1\cdot f_0\),
the occurrence of \(f_0\) on the lefthand side 
represents the probability that the outer call
of \(f\) in \(f(f\,x)\) reaches the target
(without using \((f\,x)\)),
and \(f_1\cdot f_0\) represents the probability
that the outer call of \(f\) uses the argument
\(f\,x\), and then the inner call of \(f\) reaches
the target.
Now, the whole probability that \(S\) uses its argument is represented as \(S_1\), where
\[ S_1 = F_0(H_1,H_2) + F_1(H_2),\]
with the functions \(F_0, F_1, H_1\) and \(H_2\) being as defined above.
Note that the order of the resulting equations is one.
In summary, as information about an order-1 argument \(f\) of arity \(k\), we pass around
a tuple of real numbers \((f_0,f_1,\ldots,f_k)\) where \(f_i\; (i>0)\) represents the probability that
the \(i\)-th argument is reached, and \(f_0\) represents the probability that the ``current target''
 (which is chosen by a caller) is reached.

A further twist is required in the case of order-3 or higher. 
Consider an order-3 function \(G\) defined by:
\[ G\;h\;z = h\,(H\,z)\,z\]
where \(G\COL ((\T\to\T)\to\T\to\T)\to\T\to\T\), and \(H\) is as defined above.
Following the definition of \(F_1\)
above, one may be tempted to define \(G_1\) (for computing the reachability probability
to \(z\)) as
\( G_1\;h_1 = \cdots\),
where
\(h_1\) is a function to be used for
computing
the probability that \(h\) uses its order-0 argument.
However, \(h_1\)
is not sufficient for computing
the reachability probability to \(z\);
passing the reachability
probability to the
\emph{current} target (like \(f_0\) above)
does not help either, since a caller of \(G\) does not know the current target \(z\).
We thus need to add an additional argument
\(h_2\) for computing the probability to
a target that is yet to be set
by a caller of \(h\).
Thus, the definition of \(G_1\) is:\footnote{
For the sake of simplicity,
the following translation slightly deviates from
the general translation defined later.}
\[ G_1\;(h_1, h_2) = h_2(H_1,H_2)+h_1(1).\]
Here,
 \(h_2(H_1,H_2)\) and \(h_1(1)\)
 respectively represent
the probabilities that \(G\,(H\,z^{(1)})\,z^{(2)}\)
reaches \(z^{(1)}\), and \(z^{(2)}\).
The first argument of \(h_2\) (i.e., \(H_1\))
represents the probability that \(H\,z\) reaches \(z\),
and the second argument of \(h_2\) (i.e., \(H_2\)) represents
the probability that \(H\,z\) reaches its argument (the second argument of \(H\)).

We can now formalize the general translation based on the intuitions above.
We often write \(\sty_1\to\cdots\to\sty_k\To \T^\ell\to\T\) for
\(\sty_1\to\cdots\to\sty_k\to \T^\ell\to\T\) when either \(\order(\sty_k)>0\) 
or \(k=0\). We define \(\arity(\sty)\) as the number of the last order-0 arguments,
i.e., \(\arity(\sty_1\to\cdots\to\sty_k\To \T^\ell\to\T)=\ell\).

Given a rule \(F\,z_1\,\ldots\, z_m=t_L\C{p}t_R\) of \pHORS{},
we uniquely decompose \(z_1,\ldots,z_m\) into two (possibly empty) subsequences
\(z_1,\ldots,z_\ell\) and \(z_{\ell+1},\ldots,z_{m}\) so that the order of \(z_\ell\) is
greater than \(0\) if \(\ell>0\) (note, however, that the orders of \(z_1,\ldots,z_{\ell-1}\) may be \(0\)),
 and \(z_{\ell+1},\ldots,z_{m}\) are order-0 variables
(in other words, \(z_{\ell+1},\ldots,z_{m}\)  is the maximal postfix of
\(z_1,\ldots,z_m\) consisting of only order-0 variables).
Since (the last consecutive occurrences of) order-0 arguments will be treated in a special manner, 
as a notational convenience, when we write
\(F\;\seq{y}\;\seq{x}=t_L\C{p}t_R\) for a rule of \pHORS{},
we implicitly assume that \(\seq{x}\) is the maximal postfix of
the sequence \(\seq{y}\;\seq{x}\) consisting of only order-0 variables.
Similarly, when we write \(F\,\seq{s}\,\seq{t}\) for a fully-applied term (of order 0),
we implicitly assume that \(\seq{t}\) is the maximal postfix of
the sequence of arguments, consisting of only order-0 terms.

Consider a function definition of the form:
\[ F\;y_1\;\cdots\;y_m\; x_1\;\cdots\;x_k=t_L\C{p}t_R\]
where (following the notational convention above) the sequence \(x_1,\ldots,x_k\) 
is the maximal postfix of \(y_1,\ldots,y_m,x_1,\ldots,x_k\) consisting of
only order-0 variables.
We transform each subterm \(t\) of the righthand side \(t_L\C{p}t_R\) by using
the translation relation of the form:
\[
\STE;x_1,\ldots,x_k \pN t: \sty \tr (e_0,e_1,\ldots,e_{\ell+k+1})
\]
where \(\NONTERMS\) and \(\STE\) are type environments for the underlying non-terminals and
\(y_1,\ldots,y_m\) respectively, and \(\sty\) is the type of \(t\) with \(\arity(\sty)=\ell\).
We often omit the subscript \(\NONTERMS\).
The output of the translation, 
\((e_0,e_1,\ldots,e_{\ell+k+1})\), can be interpreted as capturing the following information.
\begin{itemize}
\item \(e_0\):
the reachability probability (or a function that
returns the probability, given appropriate
arguments; similarly for the other
\(e_i\)'s below)
to the current target (set by a caller
of \(F\)).
\item \(e_{i}\; (i\in\set{1,\ldots,\ell})\):
the reachability
probability to \(t\)'s \(i\)-th order-0 argument.
\item \(e_{\ell+i}\; (i\in\set{1,\ldots,k})\):
the reachability probability to
\(x_i\).
\item \(e_{\ell+k+1}\): the reachability
probability to
a ``fresh'' target
(that can be set by a caller of \(t\));
this is the component that should be
passed as \(h_2\) in the discussion above.
In a sense,
this component represents 
the reachability probability to
a variable \(x_{k+1}\) that is
``fresh'' for \(t\)
(in that it does not occur in
\(t\)).
\end{itemize}
In the translation, each variable \(y\)
(including non-terminals) of type
\(\seq{\sty}\To \T^m\to\T\) is replaced by
\((y_0,y_1,\ldots,y_m,y_{m+1})\),
which represents information
 analogous to
 \((e_0,e_1,\ldots,e_\ell,e_{\ell+k+1})\):
\(y_0\) represents (a function for computing) the reachability
 probability to the current target, \(y_i\; (i\in\set{1,\ldots,m})\)
 represents the reachability probability to the \(i\)-th order-0
 argument (among the last \(m\) argument), and \(y_{m+1}\)
(which corresponds to \(h_2\) in the explanation above)
 represents the reachability probability to a fresh target (to be set later).
In contrast, the variables \(x_1,\ldots,x_k\) will be removed
 by the translation.

The translation rules are given in Figure~\ref{fig:tr-n-1}.
In the rules, to clarify the correspondence between source terms
and target expressions, we use metavariables \(s,t,\ldots\) (with subscripts)
also for target expressions (instead of \(e\)). 
We write \(e^k\) for the \(k\) repetitions of \(e\).

\iftwocol
\begin{figure*}[tbp]
\else
\begin{figure}[tbp]
\fi
\footnotesize
\fbox{
\begin{minipage}{.97\textwidth}
\typicallabel{Tr-AppG}
\infrule[Tr-Omega]{}{\STE; x_1,\ldots,x_k\pN \Omega:
  \T\tr (0^{k+2})}
\vspace{0.25cm}
\infrule[Tr-GVar]{}{\STE; x_1,\ldots,x_k\pN x_i:\T\tr
  (0^{i},1,0^{k-i+1})}

\vspace{0.25cm}

\infrule[Tr-Var]{\STE(y)=\seq{\sty}\To\T^\ell\to\T}
        {\STE; x_1,\ldots,x_k\pN y:\seq{\sty}\To\T^\ell\to\T \tr
          (y_0,y_1,\ldots,y_\ell,(y_{\ell+1})^{k+1})}

\vspace{0.25cm}
\infrule[Tr-NT]{\NONTERMS(F)=\seq{\sty}\To\T^\ell\to\T}
        {\STE; x_1,\ldots,x_k\pN F:\seq{\sty}\To\T^\ell\to\T \tr
          (F_0,F_1,\ldots,F_\ell,(F_{0})^{k+1})}
\vspace{0.25cm}

\infrule[Tr-App]{\STE;x_1,\ldots,x_k\pN s:\sty_1\to\seq{\sty}\To\T^\ell\to\T\tr (s_0,\ldots,s_{\ell+k+1})\\
 \STE;x_1,\ldots,x_k\pN t:\sty_1\tr (t_0,\ldots,t_{\ell'+k+1})\andalso
\arity(\sty_1) = \ell'}
{\STE;\seq{x}\pN st:\seq{\sty}\To\T^\ell\to\T \tr (s_0(t_0,\ldots,t_{\ell'},t_{\ell'+k+1}),
  s_1(t_1,\ldots,t_{\ell'},t_{\ell'+k+1}), \ldots, s_\ell(t_1,\ldots,t_{\ell'},t_{\ell'+k+1}), \\
\hfill s_{\ell+1}(t_{\ell'+1},t_1,\ldots,t_{\ell'},t_{\ell'+k+1})
\ldots, s_{\ell+k+1}(t_{\ell'+k+1}, t_1,\ldots,t_{\ell'},t_{\ell'+k+1}))}

\vspace{0.25cm}

\infrule[Tr-AppG]{\STE;x_1,\ldots,x_k\pN s:\T^{\ell+1}\to\T\tr (s_0,\ldots,s_{k+\ell+2})\andalso
\STE;x_1,\ldots,x_k\pN t:\T\tr (t_0,\ldots,t_{k+1})}
{\STE; x_1,\ldots,x_k\pN st:\T^\ell\to\T\tr
(s_0+s_1\cdot t_0,s_2,\ldots,s_{\ell+1},s_{\ell+2}+s_1\cdot t_1,\ldots,s_{\ell+k+2}+s_1\cdot t_{k+1})
}

\vspace{0.25cm}

\infrule[Tr-Rule]{
y_1\COL\sty_1,\ldots,y_\ell\COL\sty_\ell;x_1,\ldots,x_k \pN t_d:\T\tr (t_{d,0},\ldots,t_{d,k+1})\mbox{ for each $d\in\set{L,R}$}\\
\seq{y_i} = (y_{i,0},\ldots,y_{i,\arity(\sty_i)+1})\andalso
  \seq{y_i}' = (y_{i,1},\ldots,y_{i,\arity(\sty_i)+1})
}
        {\NONTERMS\p (F\,\seq{y}\;x_1\,\cdots\,x_k = t_L\C{p}t_R)\tr\\\qquad
          \set{F_i\,\seq{y_1}'\,\cdots\,\seq{y_\ell}'=p t_{L,i} + (1-p)t_{R,i} \mid i\in\set{1,\ldots,k}}
 \cup \set{F_0\,\seq{y_1}\,\cdots\,\seq{y_\ell}=p t_{L,0} + (1-p)t_{R,0}}}
  
\vspace{0.25cm}  
  
\infrule[Tr-Gram]{
  \E = \bigcup \set{\E_i \mid \NONTERMS\p
    (F\,\seq{y}\;x_1\,\cdots\,x_k = t_L\C{p}t_R)\tr \E_i,
    (F\,\seq{y}\;x_1\,\cdots\,x_k = t_L\C{p}t_R)\in \RULES}
}{(\NONTERMS,\RULES,S)\tr (\E,S_1)}
\end{minipage}}
\caption{Translation rules for the order-\((n-1)\) fixpoint characterization}
\vspace{-0.2cm}
\label{fig:tr-n-1}
\iftwocol
\end{figure*}
\else
\end{figure}
\fi

We now explain the translation rules.
In rule \rn{Tr-Omega} for the constant \(\Omega\),
all the components are \(0\) because \(\Omega\) represents divergence.
There is no rule for \(\Te\);
this is due to the assumption that \(\Te\) never occurs in the rules.
Rule \rn{Tr-GVar} is for order-0 variables, for which only one component is 1 and all
the others are 0. The (\(i+1\))-th component is \(1\),
because it represents the probability that \(x_i\) is reached.
In rule \rn{Tr-Var} for variables, the first \(\ell+1\) components are provided by
the environment. Since \(y\) (that is provided by the environment) does not ``know''
the local variables \(x_1,\ldots,x_k\)
(in other words, \(y\) cannot be instantiated to a term that contains \(x_i\)),
 the default parameter \(y_{\ell+1}\)
(for computing the reachability probability to a ``fresh'' target) is used for
all of those components. The rule \rn{Tr-NT} for non-terminals is almost the same as \rn{Tr-Var}, except that
\(F_0\) is used instead of \(F_{\ell+1}\).
This is because \(F\) does not contain any free variables;
the reachability target for \(F\) is not set yet, hence \(F_0\) can be used for computing the reachability
probability to a fresh target.
Rule \rn{Tr-App} is for applications. Basically,
the output of the translation of \(t\) is passed to \(s_i\); note however that
\(t_0\) is passed only to \(s_0\);  since \(s_1,\ldots,s_{\ell+k}\) should provide the reachability
probability to order-0 arguments of \(s\) or local variables,
the reachability probability to the
current target (that is represented by \(t_0\)) is irrelevant for them.
For \(s_{\ell+1},\ldots,s_{\ell+k}\), the reachability targets are \(x_1,\ldots,x_k\);
 thus, information about how \(t\) reaches those variables
is passed as the first argument of \(s_{\ell+1},\ldots,s_{\ell+k}\). For the last component,
\(s_{\ell+k+1}\) and \(t_{\ell'+k+1}\) are used so that 
the reachability target can be set later.
In rule \rn{Tr-AppG}, the reachability probability to the current target
(expressed by the first component) is computed by \(s_0+s_1\cdot t_0\), because
the current target is reached without using \(t\) (as represented by \(s_0\)),
or \(t\) is used (as represented by \(s_1\)) and \(t\) reaches the current target (as
represented by \(t_0\)); similarly for the reachability probability to local variables.
\rn{Tr-Rule} is the rule for translating
a function definition. From the definition for
\(F\), we generate definitions for
functions \(F_0,\ldots,F_{k}\).
For \(i\in \set{1,\ldots,k}\), \(t_{d,i}\) is chosen as the body of \(F_i\),
since it represents the
reachability probability to \(x_i\).
Rule \rn{Tr-Gram} is the translation for
the whole \pHORS{}; we just collect
the output of the translation for each rule.

For a \pHORS{} \(\GRAM=(\NONTERMS,\RULES,S)\) (where
\(\NONTERMS(S)=\T\to\T\)), we write \(\EQref{\GRAM}\) for
\(\E\) such that 
\((\NONTERMS,\RULES,S)\tr (\E,S_1)\).
Such an \(\E\) 
is actually unique (up to \(\alpha\)-equivalence), given \(\NONTERMS\) and \(\RULES\).
Note also that by definition of the translation relation, the output of
the translation always exists.

\begin{example}
\label{ex:order2-phors-eq}
Recall the order-2 \pHORS{} \(\GRAM_2\) in Example~\ref{ex:order2-phors}:
\begin{align*}
  S\ &=\ F\,H \\H\,x &= x \C{\frac{1}{2}} \Omega \\ F\,g &= (g\,\Te) \C{\frac{1}{2}} (F(D\,g))\\
  D\,g\,x &= g\,(g\,x). \end{align*}
It 
can be modified to the following rules so that \(\Te\) does not occur.
\begin{align*}
S\,z &= F\,H\,z\\
H\,x &= x \C{\frac{1}{2}} \Omega\\
F\,g\,z &= (g\,z) \C{\frac{1}{2}} (F(D\,g)\,z)\\
D\,g\,x &= g\,(g\,x). 
\end{align*}
Here, \(\Te\) can be passed around through the variable \(z\).
Consider the body \(F\,H\,z\) of \(S\).
\(F\) and \(H\) are translated as follows.
\[
\begin{array}{l}
\emptyset;z\pN F:(\T\to\T)\to\T\to\T\tr
(F_0,F_1,F_0,F_0)\\
\emptyset;z\pN H:\T\to\T\tr
(H_0,H_1,H_0,H_0)\\
\end{array}
\]
By applying \rn{Tr-App}, we obtain:
\[
\begin{array}{l}
\emptyset;z\pN F\,H:\T\to\T\tr (F_0(H_0,H_1,H_0), F_1(H_1,H_0), F_0(H_0,H_1,H_0), \iftwocol\\\qquad \fi F_0(H_0,H_1,H_0)).
\end{array}
\]
Using \rn{Tr-GVar}, \(z\) can be translated as follows.
\[\emptyset;z\pN z:\T\tr (0,1,0).\]
Thus, by applying \rn{Tr-AppG}, we obtain:
\[
\begin{array}{l}
\emptyset;z\pN F\,H\,z:\T\tr\\\qquad
(F_0(H_0,H_1,H_0)+ F_1(H_1,H_0)\cdot 0,\\\qquad \;
F_0(H_0,H_1,H_0)+F_1(H_1,H_0)\cdot 1,\iftwocol\\\qquad \fi
F_0(H_0,H_1,H_0)+F_1(H_1,H_0)\cdot 0).
\end{array}
\]
By simplifying the output, we obtain:
\[
\begin{array}{l}
\emptyset;z\pN F\,H\,z:\T\tr\iftwocol\\\qquad \fi
(F_0(H_0,H_1,H_0),
F_0(H_0,H_1,H_0)+F_1(H_1,H_0),\iftwocol\\\qquad \fi
F_0(H_0,H_1,H_0)).
\end{array}
\]
 Thus, we have the following equations for \(S_0\) and \(S_1\). \[
\begin{array}{l}
 S_0 = F_0(H_0,H_1,H_0)\qquad S_1 =  F_0(H_0,H_1,H_0)+ F_1(H_1,H_0).
\end{array}
\]
The following equations are obtained for the other non-terminals.
\begin{align*}
   H_0 &= 0\qquad\qquad H_1=\frac{1}{2}\\
   F_0\,(g_0,g_1,g_2) &= \frac{1}{2}g_0+F_0(D_0(g_0,g_1,g_2),D_1(g_1,g_2),D_0(g_2,g_1,g_2))\\
   F_1\,(g_1,g_2)  &= \frac{1}{2}(g_1+g_2)
   +\frac{1}{2}(F_0(D_0(g_2,g_1,g_2),D_1(g_1,g_2),D_0(g_2,g_1,g_2))\\
   &\qquad+F_1(D_1(g_1,g_2),D_0(g_2,g_1,g_2))\\
   D_0(g_0,g_1,g_2)&=g_0+g_1g_0\\
   D_1(g_1,g_2)&=g_2+g_1(g_1+g_2).
\end{align*}
We can observe that the values of the variables \(g_0\) and \(g_2\) are always \(0\).
 Thus, by removing redundant arguments, we obtain:
\begin{align*}
  S_0&=F_0(\frac{1}{2})\\
  S_1&=F_0(\frac{1}{2})+F_1(\frac{1}{2})\\
  F_0(g_1)&=F_0(D_1(g_1))\\
  F_1(g_1) &= \frac{1}{2}g_1+\frac{1}{2}(F_0(D_1(g_1))+F_1(D_1(g_1)))\\
  D_0(g_1)&=0\\
  D_1(g_1)&=g_1^2.
\end{align*}
By further simplification (noting that the least solution for \(F_0\) is
\(\lambda g_1.0\)), we obtain:
\[
\begin{array}{c}
S_1 = F_1(\frac{1}{2})\qquad\qquad
F_1(g_1) = \frac{1}{2}g_1+\frac{1}{2}F_1(g_1^2).
\end{array}
\]
The least solution of \(S_1\) is \(\Sigma_{i\geq 0} \frac{1}{2^{2^i+i+1}} = 0.3205\cdots\).
\hfill\qed
\end{example}

\begin{example}
  Consider the following order-3 \pHORS{}:
  \[
  S\, x=F(C\,x)\quad F\,g = g\,H\quad C\,x\,f=f\,x\quad H\,x=x\C{\frac{1}{2}}\Omega,  \]
  where
  \[
\begin{array}{l}
  S\COL\T\to\T, F\COL ((\T\to\T)\to\T)\to\T,
  C\COL \T\to (\T\to\T)\to\T,
  H\COL\T\to\T.
\end{array}
  \]
  This is a tricky example, where in the body of \(S\), \(x\) is embedded into
  the closure \(C\,x\) and passed to another function \(F\); so, in order to compute
  how \(S\) uses \(x\), we have to take into account how \(F\) uses the closure
  passed as the argument. The \pHORS{} is translated to:
  \begin{align*}
    S_0&= F_0(C_0(0,0),C_0(0,0))\\
    S_1&= F_0(C_0(1,0),C_0(0,0))\\
    F_0\,(g_0,g_1)&=g_0(H_0,H_1,H_0)\\
    C_0\,(x_0,x_1)\,(f_0,f_1,f_2) &= f_0+f_1\cdot x_0\\
    H_0 &=0\qquad H_1=\frac{1}{2},
  \end{align*}
  where
  \begin{align*}
  S_0&\COL\realt, S_1\COL\realt,\\ F_0&\COL (\realt\times\realt\times\realt\to\realt)\times (\realt\times\realt\times\realt\to\realt)\to\realt,\\
C_0&\COL (\realt\times\realt)\to (\realt\times\realt\times\realt)\to\realt,\\
H_0&\COL\realt,H_1\COL\realt. \end{align*}
  The order of the equations is \(2\) (where the largest order is that of
  the type of  \(F_0\)).
  We have:
  \[
  S_1=F_0(C_0(1,0),C_0(0,0))=C_0(1,0)(H_0,H_1,H_0) = H_0+H_1\cdot 1 = \frac{1}{2}.
  \]
  In fact, the probability that \(S\,\Te\) reaches \(\Te\) is \(\frac{1}{2}\). \qed
\end{example}

\begin{example}
\label{ex:eq-treegenp}
Recall \pHORS{} \(\GRAM_5\) from Example~\ref{ex:listgen-variant}:
\begin{align*}
S\;x &= \treegen\;H\;\boolgen\;x\\
\boolgen\;k &= k\\ H\;x\;y &= x\C{\frac{1}{2}}y\\
G\;p\;x\;y &= x\C{\frac{1}{2}}(p\;x\;y)\\
\treegen\;p\;f\;k &= p\;k\;
(f(\treegen\,(G\;p)\,f\,(\treegen\,(G\;p)\,f\,(\treegen\,(G\;p)\,f\,k)))).
\end{align*}
(Here, we have slightly modified the original \pHORS{} so that \(S\) is parameterized with \(\Te\).)
As the output of the translation as defined above is too complex, we show below
a hand-optimized version of the fixpoint equations.
\begin{align*}
S_1 &= \treegen_1\;(H_1,H_2)\\
H_1 &= H_2=\frac{1}{2}\\
G_1\;(p_1,p_2) &= \frac{1}{2}+\frac{1}{2}p_1\\
G_2\;(p_1,p_2) &= \frac{1}{2}p_2\\
\treegen_1\;(p_1,p_2) &= 
p_1 + p_2\cdot (\treegen_1\;(G_1(p_1,p_2), G_2(p_1,p_2)))^3.
\end{align*}
Here, \(\treegen_1\) is the function that returns the probability that
\(\treegen\;p\;\boolgen\;x\) reaches \(x\), where the parameters \(p_1\) and \(p_2\)
represent the probabilities that \(p\) chooses the first and second branches respectively.
Let \(\rho\) be the least solution of the fixpoint equations above.
We can find \(\rho(S_1)=1\) based on the following reasoning (which is also confirmed by
the experiment reported in Section~\ref{sec:exp}).
Let us define an \(m\)-th approximation \(\treegen_1^{(m)}\) of \(\rho(\treegen_1)\) by
\begin{align*}
  \treegen_1^{(0)}\;(p_1,p_2) &= 0 \\
  \treegen_1^{(m+1)}\;(p_1,p_2) &= p_1 + p_2\cdot (\treegen_1^{(m)}\;(G_1(p_1,p_2), G_2(p_1,p_2)))^3.
\end{align*}
Then \(\rho(\treegen_1)\;(p_1,p_2)\geq \treegen_1^{(m)}\;(p_1,p_2)\) for every \(m\geq 0\).
We show \(\treegen_1^{(m)}\;(1-\frac{1}{2^{n}},\frac{1}{2^n})\geq 1-\frac{1}{2^{n+m-1}}\)
for every \(n\geq 2\), \(m\geq 1\) by induction on \(n\).
When \(m=1\), we have:
\[
\treegen_1^{(1)}\;\left(1-\frac{1}{2^n},\frac{1}{2^n}\right) = 
\left(1-\frac{1}{2^n}\right) + \frac{1}{2^n}\cdot 0 = 1-\frac{1}{2^n}=1-\frac{1}{2^{n+m-1}}.
\]
About the inductive step, we have
\begin{align*}
\treegen_1^{(m+1)}\;\left(1-\frac{1}{2^n},\frac{1}{2^n}\right) 
&= 
1-\frac{1}{2^n} + \frac{1}{2^n}\cdot \left(\treegen_1^{(m)}\;\left(1-\frac{1}{2^{n+1}},\frac{1}{2^{n+1}}\right)\right)^3\\
&\geq 
1-\frac{1}{2^n} + \frac{1}{2^n}\cdot \left(1-\frac{1}{2^{n+m}}\right)^3\\
&\geq 
1-\frac{1}{2^n} + \frac{1}{2^n}\cdot \left(1-3\frac{1}{2^{n+m}}\right)\\
&\geq
1-\frac{1}{2^{2n+m-2}}
\geq
1-\frac{1}{2^{n+(m+1)-1}}
\end{align*}
as required.
Thus, 
\[
\treegen_1^{(m)}\left(\frac{1}{2},\frac{1}{2}\right)
= \frac{1}{2}+\frac{1}{2}\left(\treegen_1^{(m-1)}\left(1-\frac{1}{2^2},\frac{1}{2^2}\right)\right)^3
\geq \frac{1}{2}+\frac{1}{2}\left(1-\frac{1}{2^{(2+(m-1)-1)}}\right)^3
\]
for every \(m\geq 2\). Thus, \(\rho(S_1)=\rho(\treegen_1)(\frac{1}{2},\frac{1}{2})\) 
(which should be no less than 
\(\treegen_1^{(m)}\;(\frac{1}{2},\frac{1}{2})\) for every \(m\))
must be \(1\).
\qed
\end{example}
 
\subsubsection*{Correctness of the Translation}
To state the well-formedness of the output of the translation,
we define
the translation of types as follows.
\[
\begin{array}{l}
\trT{(\sty_1\to\cdots\to\sty_k\To \T^\ell\to\T)}\\
=(\trT{\sty_1}\to \cdots \to \trT{\sty_k}\to \realt)
\iftwocol\\\quad\fi\times (\trTp{\sty_1}\to \cdots \to \trTp{\sty_k}\to \realt)^\ell
\times (\trT{\sty_1}\to \cdots \to \trT{\sty_k}\to \realt)\\
\trTp{(\sty_1\to\cdots\to\sty_k\To \T^\ell\to\T)}\\
= (\trTp{\sty_1}\to \cdots \to \trTp{\sty_k}\to \realt)^\ell
\times (\trT{\sty_1}\to \cdots \to \trT{\sty_k}\to \realt).
\end{array}
\]
We also write \(\trTn{(\sty_1\to\cdots\to\sty_k\To \T^\ell\to\T)}{m}\) for
\[
\begin{array}{l}
(\trT{\sty_1}\to \cdots \to \trT{\sty_k}\to \realt)\iftwocol \\\fi
\times (\trTp{\sty_1}\to \cdots \to \trTp{\sty_k}\to \realt)^\ell
\times (\trT{\sty_1}\to \cdots \to \trT{\sty_k}\to \realt)^{m+1}.
\end{array}
\]
It represents the type
of the tuple \((e_0,\ldots,e_{\ell+m+1})\) obtained by translating
a term of type \(\sty_1\to\cdots\to\sty_k\To \T^\ell\to\T\) with
order-0 variables \(x_1,\ldots,x_m\). The distinction between
\(\trT{\sty_i}\) and \(\trTp{\sty_i}\) reflects the fact that
in the output \((e_0,e_1,\ldots,e_\ell,e_{\ell+1},\ldots,e_{\ell+m+1})\)
of the translation, \(e_1,\ldots,e_\ell\) take one less argument
(recall \rn{Tr-App}).
The translation of the type environment \(\NONTERMS\) for non-terminals is defined by:
\[
\begin{array}{l}
(F_1\COL\sty_1,\ldots,F_k\COL\sty_k)^\dagger
=
(F_{1,0},\ldots,F_{1,\arity(\sty_1)})\COL {\sty_1}^{\dagger-1},
\ldots, 
(F_{k,0},\ldots,F_{k,\arity(\sty_k)})\COL \sty_k^{\dagger-1}.
\end{array}
\]
The following lemma states that the output of the translation is 
 well-typed. 
\begin{lemma}[Well-typedness of the output of transformation]
  \label{lem:tr-well-typedness}
  Let \(\GRAM=(\NONTERMS,\RULES,S)\) be a \pHORS{}.
If \(\GRAM\tr (\E,S_1)\), then
   \(\trT{\NONTERMS} \p \E\) and \(\trT{\NONTERMS}(S_1)=\realt\).  
\end{lemma}
 \noindent
As a corollary, it follows that for any order-\(n\) \pHORS\;  \(\GRAM\) (where \(n>0\)),
the order of \(\EQref{\GRAM}\) is \(n-1\).

The following result is the main theorem of this section, which states the correctness
of the translation.
\begin{theorem}
  \label{prop:order-k-1-fixpoint}
  \label{PROP:ORDER-K-1-FIXPOINT}
  Let \(\GRAM=(\NONTERMS,\RULES,S)\) be an order-\(n\) \pHORS{},
  Then, \(\Prob(\GRAM,S\,\Te) = \rho_{\EQref{\GRAM}}(S_1)\).
\end{theorem}

\iftwocol
A proof of the theorem and more examples of
the translation are found in
\iffull Appendix~\ref{sec:proof-sec42}. \else
the longer version~\cite{KDG19LICSfull}.
\fi

\else
A proof of the theorem is found in Appendix~\ref{sec:proof-sec42}.
Here we only sketch the proof. 
We first prove the theorem for \emph{recursion-free} \pHORSs{}
(so that any term is strongly normalizing; see
\iffull Appendix~\ref{sec:proof-tr-n-1} \else \cite{KDG19LICSfull} \fi
for the precise definition), and
extend it to general \pHORSs{} \(\GRAM\) by
using \emph{finite approximations} of \(\GRAM\), obtained by unfolding each non-terminal a finite number of times.
To show the theorem for recursion-free \pHORSs{}, we prove that
the translation relation is preserved by reductions in a certain sense; this is,
however, much more involved than the corresponding proof for
Section~\ref{sec:order-n-equation}: we introduce an alternative
operational semantics for \pHORS{} that uses \emph{explicit} substitutions.
See Appendix~\ref{sec:proof-sec42} for details.
\fi
  
\section{Computing Upper-Bounds of Termination Probability}
\label{sec:upperbound}
\label{SEC:UPPERBOUND}
Theorems~\ref{prop:order-k-fixpoint} and
\ref{prop:order-k-1-fixpoint} immediately provide procedures
for computing \emph{lower}-bounds of the termination probability as
precisely as we need
\footnote{Theorem~\ref{th:lower-bound} also provides a procedure
for computing lower-bounds, but the fixpoint characterizations by
Theorems~\ref{prop:order-k-fixpoint} and \ref{prop:order-k-1-fixpoint}
provide a more efficient procedure.}
The termination
probability, in other words, is a \emph{recursively enumerable} real number
(see, e.g.~\cite{CaludeInformationRandomness}), but it is still open
whether it is a \emph{recursive} one.
Indeed, computing good upper-bounds is non-trivial.  For example, an
upper-bound for the \emph{greatest} solution of \(\EQref{\GRAM}\) can be easily computed,
but it does not provide a good upper-bound for the \emph{least}
solution of \(\EQref{\GRAM}\), unless the solution is unique. Take, as an example, the
trivial \pHORS{}  consisting of a single equation $S=S$: the
greatest solution is $1$, while the least is $0$.

In this section, we will describe how \emph{upper approximations} to
the termination probability can be computed in practice.
We focus our attention mainly on order-2 \pHORSs{}, which yield equations over
first-order functions on real numbers. Order-\(n\) case is only briefly discussed
in Section~\ref{sec:ub:higherorder}.
\subsection{Properties of the Fixpoint Equations Obtained from \pHORS{}}
\label{sec:properties}
Before discussing how to compute an upper-bound of the termination
probability, we first summarize several important properties of the
(order-1) fixpoint equations obtained from an order-2 \pHORS{}
(by the translation in Section~\ref{sec:order-n-1-equation}), which are exploited in
computing upper-bounds.
\begin{asparaenum}
\item
The fixpoint equations can be written in the form:
\begin{equation}\label{equ:fpediscr}    
  f_1(x_1,\ldots,x_{\ell_1}) = e_1,\qquad\cdots\qquad f_k(x_1,\ldots,x_{\ell_k})=e_k,
  \end{equation}
  where each \(e_i\) consists of (i) non-negative constants,
  (ii) additions, (iii) multiplications, and (iv) function applications.
  Each variable \(x_i\) ranges over \([0,1]\).
\item
The formal arguments \(x_1,\ldots,x_{\ell_i}\) of each
  function \(f_i\) can be partitioned into several groups of
  variables \((x_1,\ldots,x_{d_{i,1}}),(x_{d_{i,1}+1},\ldots,x_{d_{i,2}}),\ldots,
  (x_{d_{g_{i}-1}+1},\ldots,x_{\ell_i})\), so that the relevant 
  input values are those such that \emph{the sum} of the
  values of the variables in each group ranges over \([0,1]\). This is
  because each group of
  variables \((x_{d_{i,j-1}+1},\ldots,x_{d_{i,h}})\) either corresponds
  to an order-$0$ argument (and has thus length $1$) or to an
  order-$1$ argument of an order-2 function \(F_i\) of the original
  \pHORS{}, where one of the variables represents the probability
  that \(F_i\) terminates without using any arguments, and each of the
  other variables represents the probability that \(F_i\) uses each
  argument of \(F_i\). Since these events are mutually exclusive, the sum of
 those values ranges over \([0,1]\).
\item The functions \(f_1,\ldots,f_k\) can also be partitioned into several groups
of functions \((f_1,\ldots,f_{j_1}),(f_{j_1+1},\ldots,f_{j_2}),\ldots,(f_{j_{\ell-1}+1},\ldots,f_{j_\ell})\),
so that the sum \(f_{j_{m-1}+1}(\seq{x})+\cdots + f_{j_m}(\seq{x})\) of the return values
of the functions in each group ranges over \([0,1]\) (assuming that the arguments
\(\seq{x}\) are in the valid domain, i.e., the sum of \(\seq{x}\) ranges over \([0,1]\)). 
This is because an order-2 function \(F_i\)
is translated to a tuple of order-1 functions \((F_{i,0},\ldots,F_{i,j})\), and the components of
the tuple return the probabilities to reach mutually different targets.\footnote{According 
to the translation in Section~\ref{sec:order-n-1-equation}, the first element 
\(F_{i,0}\) takes one more argument than the other elements; for the sake of
simplicity, we assume in this section that all the functions in each partition take the same number
of arguments, by adding dummy arguments as necessary.} We write \(\fgroup(f)\) for the partition
that \(f\) belongs to, i.e., \(\fgroup(f_i) = \set{f_{j_{m-1}+1},\ldots,f_{j_m}}\) if \(j_{m-1}+1\leq
i\leq j_m\).
\item
Suppose that \((x_1,\ldots,x_{\ell_i})\) ranges over the
  valid domain of \(f_i\). Then, the value of each subexpression
  of \(e_i\) ranges over \([0,1]\); this is because each subexpression
  represents some probability.  This invariant is not
  necessarily preserved by simplifications like
  \(\frac{1}{2}x+\frac{1}{2}y = \frac{1}{2}(x+y)\); the value
  of \(x+y\) may not belong to \([0,1]\).  We apply
  simplifications only so that the invariant is maintained.
\end{asparaenum}
The properties above can be easily verified by inspecting
the translations from Section~\ref{sec:order-n-1-equation}.

Finally, another important property is pointwise convexity. The
least solution $f$ of the fixpoint equations, as well as any finite
approximations obtained from \(\bot\) by Kleene iterations, are
\emph{pointwise convex}, i.e., convex on each variable,
i.e., \(f(x_1,\ldots,(1-p)x+py,\ldots,x_n) \leq
(1-p)f(x_1,\ldots,x,\ldots,x_n)+p f(x_1,\ldots,y,\ldots,x_n)\)
whenever \(0\leq p\leq 1\) and \(0\leq x,y\).  Note, however,
that \(f\) is not necessarily convex in the usual
sense: \(f((1-p)\vec{x}+p\vec{y}) \leq (1-p)f(\vec{x})+pf(\vec{y})\)
may not hold for some \(\vec{0}\leq \vec{x},\vec{y}\) and \(0\leq
p\leq 1\). For example, let \(f(x_1,x_2)\) be \(x_1\cdot x_2\). Then,
\(\frac{1}{4}=f\left(\frac{1}{2}, \frac{1}{2}\right)
=f\left(\frac{1}{2}(0,1)+\frac{1}{2}(1,0)\right)
> \frac{1}{2}f(0,1)+\frac{1}{2}f(1,0)=0\).
Recall that \(\F_{\EQref{\GRAM}}\) is the functional associated with
the fixpoint equations \(\EQref{\GRAM}\); we simply write \(\F\)
for \(\F_{\EQref{\GRAM}}\) below.
\begin{lemma}
$\F^m(\bot)$ and $\mathbf{lfp}(\F)$ are both pointwise convex.
They are also monotonic.
\end{lemma}
\begin{proof}
The pointwise convexity and monotonicity of $\F^m(\bot)$ follow from the fact
that, following our first observation, it is (a tuple of) 
multi-variate polynomials with non-negative integer coefficients.
The pointwise convexity of $\mathbf{lfp}(\F)$ follows from the fact that, for every $m$,
when $\vec{x}$ and $\vec{y}$ differ by at most one coordinate,
\[
\begin{array}{lcl}
(1-p)\mathbf{lfp}(\F)(\vec{x})+p\mathbf{lfp}(\F)(\vec{y})&\geq&
(1-p)\F^m(\bot)(\vec{x})+p\F^m(\bot)(\vec{y})\\&\geq&
\F^m(\bot)((1-p)\vec{x}+p\vec{y})
\end{array}\] 
and we can then take the supremum to conclude.
The monotonicity of $\mathbf{lfp}(\F)$ also follows from a similar argument.
\end{proof}

\subsection{Computing an Upper-Bound by Discretization}
\label{sec:ub-by-disc}
Given fixpoint equations as in (\ref{equ:fpediscr}), we can compute an
upper-bound of the least solution of the equations, by
overapproximating the values of \(f_1,\ldots,f_k\) at a finite number of
discrete points, \emph{\`a la} ``Finite Element Method''. To clarify the idea,
we first describe a method for the simplest case of a single
equation \(f(x)=e\) on a unary function in Section~\ref{sec:ub:unary}.
We then extend it to deal with a \emph{binary} function in
Section~\ref{sec:ub:binary}, and discuss the general case (where we
need to deal with \emph{multiple equations on multi-variate functions}) in
Section~\ref{sec:ub:general}.
\subsubsection{Computing an Upper-Bound for a Unary Function}
\label{sec:ub:unary}
Suppose that we are given a \pHORS{} \(\GRAM\) and that 
\(\EQref{\GRAM}\) consists of a single equation 
\(f(x)=e\), where $f$ is a function \(f\) from \([0,1]\) to \([0,1]\),
and where \(e\) consists of non-negative real constants, the variable \(x\),
additions, multiplications, and applications of \(f\).  We
abstract \(f\) to a sequence of real numbers \((r_0,\ldots,r_n)\in[0,1]^{n+1}\),
where
\(r_i\) represents the value of \(f(\frac{i}{n})\).
Thus, the abstraction function \(\alpha\) mapping any function
$f:[0,1]\rightarrow[0,1]$ to its abstract form
$[0,1]^{n+1}$ is defined by
\[\alpha(f) = \left(f\left(\frac{0}{n}\right), f\left(\frac{1}{n}\right),\ldots,f\left(\frac{n}{n}\right)\right).\]
We write \(\gamma\) for any concretization
function, mapping any element of $[0,1]^{n+1}$
back to a function in $[0,1]\rightarrow[0,1]$. The idea here
is that \emph{if \(\gamma\) satisfies certain assumptions}, to be 
given later in Lemma~\ref{lem:ub-soundness},
then we can obtain an upper-bound of the least 
solution of \(f=\F(f)\) by solving the following system of inequalities on 
the real numbers $\vec{r}=(r_0,\ldots,r_n)$:
\begin{equation}\label{equ:inapprox}
\vec{r}\geq \alpha(\F(\gamma(\vec{r}))).
\end{equation}
Let \(\widehat{\F}\) be the functional
$\metalambda\vec{s}.\alpha(\F(\gamma(\vec{s})))$.  Notice that
solutions to (\ref{equ:inapprox}) are precisely the pre-fixpoints
of \(\widehat{\F}\), and we will thus call them \emph{abstract pre-fixpoints}
of $\F$.

There are at least two degrees of freedom
  here:
\begin{asparaenum}\item \textbf{How could we define the concretization function?}
  Here we have at least two choices (see
  Figure~\ref{fig:overapproximation}): 
  \begin{varitemize} 
  \item[(a)]
        $\gamma(\vec{r})$ could be
        the \emph{step function} \(\hat{f}\) such that \(\hat{f}(0)=r_0\)
        and \(\hat{f}(x)=r_i\) if \(x\in
        (\frac{i-1}{n}, \frac{i}{n}]\).  
  \item[(b)] 
        $\gamma(\vec{r})$ could 
        be the \emph{piecewise linear function} \(\hat{f}\) such that \(\hat{f}(x)=r_i
        + \frac{x-\frac{i}{n}}{\frac{1}{n}}(r_{i+1}-r_i) (= (i+1-nx)r_i +
        (nx-i)r_{i+1})\)
if \(x\in [\frac{i}{n}, \frac{i+1}{n}]\).
  \end{varitemize}
The discussion above on abstract pre-fixpoints suggests that it is natural
   to require that \((\alpha,\gamma)\) satisfies a Galois connection-like property.
   The first choice indeed turns \((\alpha,\gamma)\) into a Galois connection
   between the set of monotonic functions and that of non-decreasing sequences of
    real numbers. The second choice of \((\alpha,\gamma)\) is not exactly
     a Galois connection (because \(\alpha(f)\le \vec{r}\) does not imply
     $f \le \gamma(\vec{r})$ if $f$ is not convex), but is quite close:
    if an abstraction $\vec{r}$ majorizes
  $\alpha(f)$ for some pointwise convex function $f$, we
  immediately have $\gamma(\vec{r})\geq f$. This way,
  $\gamma$ satisfies the assumption of Lemma~\ref{lem:ub-soundness}
  below.
\item \textbf{How could we solve inequalities?} Again, we have at least two choices.
  \begin{varitemize}
  \item[(c)] Use the decidability of \emph{theories of real arithmetic} 
    (e.g., minimize \(\sum_i r_i\) so that all the inequalities are satisfied).
  \item[(d)] \emph{Abstract} also the values of \(\vec{r}\) so that they can take only
    finitely many discrete values, say, \(0,\frac{1}{m},\ldots, \frac{m-1}{m},1\).
      The inequality (\ref{equ:inapprox}) is then replaced by:
\[
      \vec{s}\geq\alpha_h(\alpha(\F(\gamma(\vec{s})))),
      \]
      where every \(s_i\) is the ``discretized version'' of \(r_i\), and the abstraction
      function \(\alpha_h\), given a tuple of reals as an input,  replaces
      each element \(r\in[0,1]\) 
      with \(\frac{\ceil{rm}}{m}\).
Since they are now inequalities over a \emph{finite} domain, we can obtain the least solution
      by a \emph{finite} number of Kleene iterations, starting from \(\vec{s}=\vec{0}\).
  \end{varitemize}
\end{asparaenum}
\begin{figure}
\begin{center}
\includegraphics[scale=0.3]{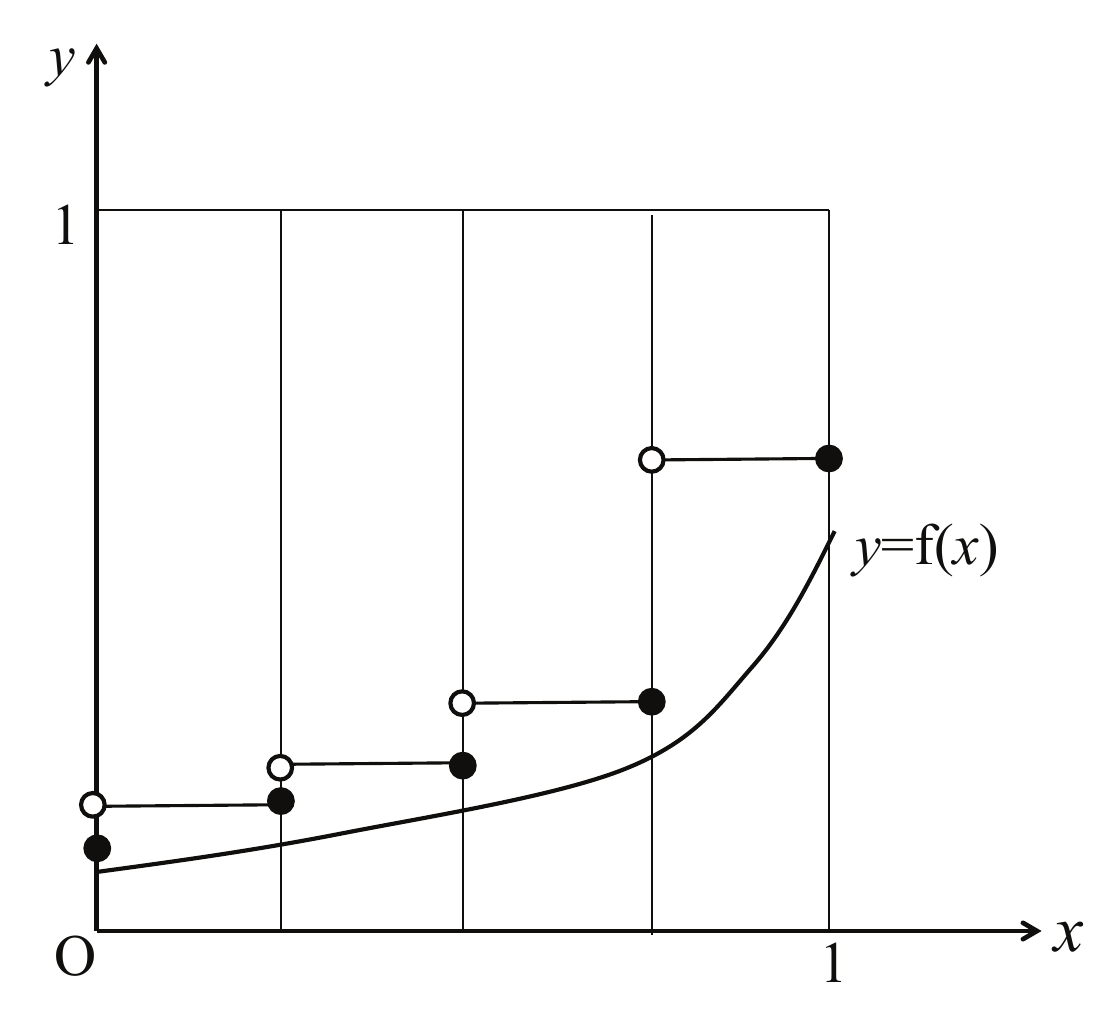}\qquad
\includegraphics[scale=0.3]{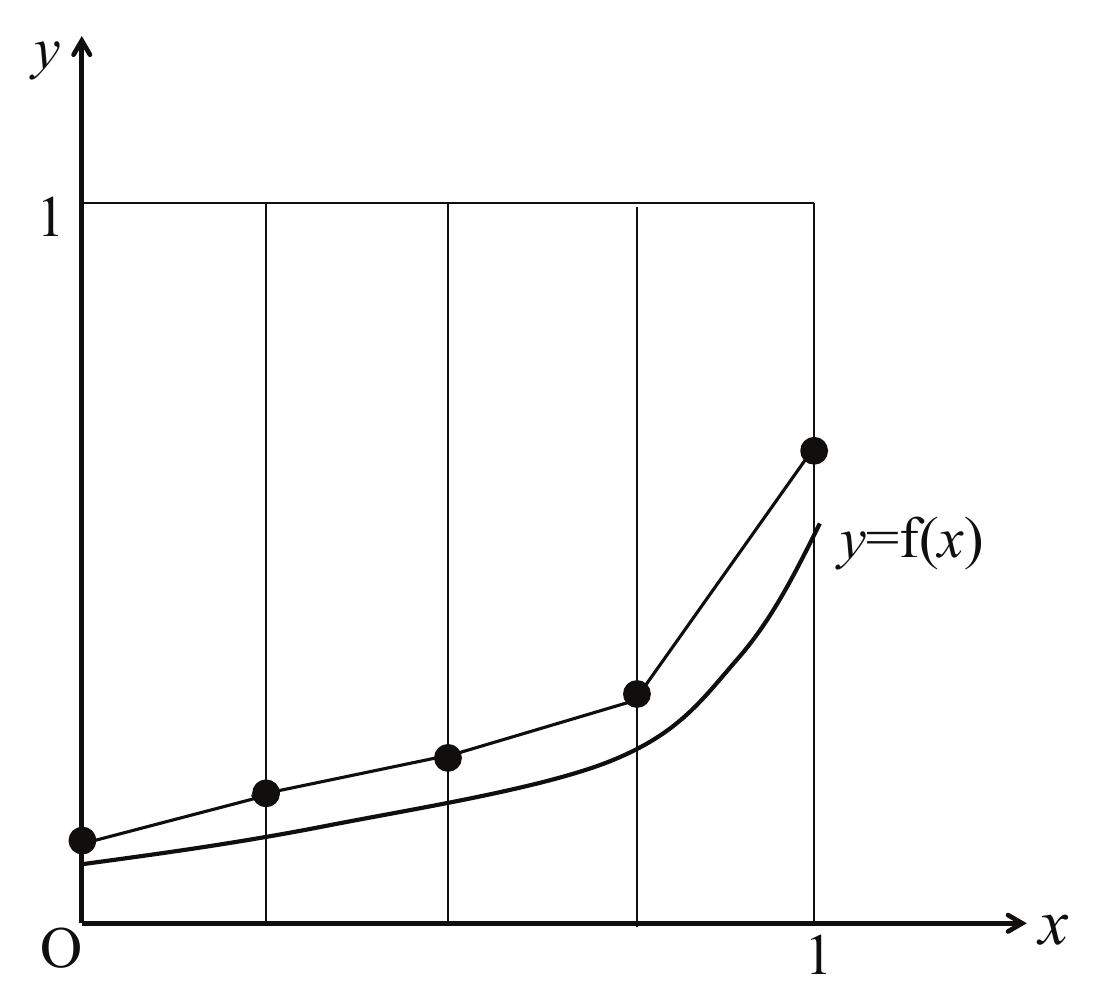}
\end{center}
\caption{Overapproximation by a step-function (left) and a stepwise linear function (right)}
\label{fig:overapproximation}
\end{figure}
The following lemma ensures that the inequality (\ref{equ:inapprox})
is indeed a sufficient condition for \(\gamma(\vec{r})\) to be an upper-bound
on \(\mathbf{lfp}(\F)\). Note that both step functions and stepwise linear
functions satisfy the assumption of the lemma below.
\begin{lemma}
\label{lem:ub-soundness}
  Suppose that the concretization function \(\gamma\) is monotonic,
  and that $\vec{r}\geq\alpha(f)$ implies $\gamma(\vec{r})\geq f$ for
  every pointwise convex $f$. Then, any abstract pre-fixpoint of $\F$
  is an upper bound of $\mathbf{lfp}(\F)$. 
\end{lemma}
\begin{proof}First, we show
  that \(\gamma(\widehat{\F}^m(\bot))\geq \F^m(\bot)\)
holds for every $m$, by induction on $m$.
The base case $m=0$ is trivial, since $\F^0(\bot)=\bot$.
If $m>0$, then we have
     \[
\begin{array}{l}
     \widehat{\F}^m(\bot)=
     \alpha(\F(\gamma(\widehat{\F}^{m-1}(\bot)))) \geq
     \alpha(\F(\F^{m-1}(\bot)))=
     \alpha(\F^{m}(\bot)).
\end{array}
     \]
     Since, by hypothesis, $\vec{r}\geq\alpha(\F^m(\bot))$ implies that
     $\gamma(\vec{r})\geq \F^m(\bot)$ (since $\F^m(\bot)$ is pointwise convex),
     we can conclude that $\gamma(\widehat{\F}^m(\bot))\geq \F^m(\bot)$.
Now, suppose $\vec{r}$ is an abstract pre-fixpoint of \(\F\). Then \(
  \gamma(\vec{r})\geq\gamma(\widehat{\F}^m(\vec{r}))\geq\gamma(\widehat{\F}^m(\bot))\geq \F^m(\bot)
\),
and as $\F$ is $\omega$-continuous, we have \(
  \gamma(\vec{r})\geq\sup_{m\in\omega}\F^m(\bot)=\mathbf{lpf}(\F)
\)
as required. \end{proof}

Below we consider the combination of (b) and (d).
Figure~\ref{fig:algo-unary} shows a pseudo code for computing an upper-bound of
\(f(c)\) for the least solution \(f\) of \(f(x)=e\) and \(c\in[0,1]\).
In the figure, \(\alpha_h(x)=\frac{\ceil{mx}}{m}\).
The algorithm terminates under the assumptions that (i) \(e\) consists of
non-negative constants, \(x\), \(+\), \(\cdot\), and applications of \(f\), and 
(ii) every subexpression of \(e\) 
evaluates to a value in \([0,1]\) (if \(x\in [0,1]\) and \(f\in [0,1]\to[0,1]\)),
which are satisfied by the fixpoint equations obtained from a \pHORS{}
(recall Section~\ref{sec:properties}).
\iftwocol
\begin{figure*}[tb]
\else
\begin{figure}[tb]
\fi
\begin{alltt}
main(\(e\), \(c\))\{
  \(r\) := \([0,\ldots,0]\);  \(r'\) := \([1,\ldots,1]\) (* dummy *);
  while not(\(r\)=\(r'\)) do \{
    \(r'\) := \(r\); (* copy the contents of array \(r\) to \(r'\) *)
    for each \(i\in\set{0,\ldots,n}\) do \(r'[i]\) := \(\alpha_h\)(eval(\(e\), \(\set{f\mapsto r, x\mapsto \frac{i}{n}}\)))\};
  return apply(\(r\), \(c\)); \}
apply(\(r\), \(c\)) \{ (* apply the function represented by array \(r\) to \(c\) *)
  \(i\) := \(\floor{nc}\); (* \(\frac{i}{n}\leq c <\frac{i+1}{n}\) *)
  return \((i+1-nc)r[i]+(nc-i)r[i+1]\); \}
eval(\(e\), \(\rho\))\{
  match \(e\) with
     \(x\) \(\to\) return \(\rho(x)\) |  \(c\) \(\to\) return \(c\)
   | \(f(e')\) \(\to\) return apply(\(\rho(f)\), eval(\(e'\), \(\rho\)))
   | \(e_1+e_2\) \(\to\) return eval(\(e_1\),\(\rho\))+eval(\(e_2\),\(\rho\))
   | \(e_1\cdot{}e_2\) \(\to\) return eval(\(e_1\),\(\rho\))\(\cdot\)eval(\(e_2\),\(\rho\)) \}
\end{alltt}
\caption{Pseudo code for computing an upper-bound of \(f(c)\) where \(f(x)=e\) (unary function case)}
\label{fig:algo-unary}
\iftwocol
\end{figure*}
\else
\end{figure}
\fi
\iftwocol\else
\begin{example}
Consider \(f(x)=\frac{1}{4}x+\frac{3}{4}f(f(x))\) and let \(n=2\) and \(m=4\). 
The value \(r^{(j)}\) of \(r\) after the \(j\)-th iteration is given by:
\[
r^{(0)} = [0,0,0]; \qquad
r^{(1)} = [0,0.25,0.25]; \qquad
r^{(2)} = [0,0.25,0.5]; \qquad
r^{(3)} = [0,0.25,0.5].
\]
Thus, the upper-bound obtained for \(f(1)\) is \(0.5\). The exact value of \(f(1)=\frac{1}{3}\).
A more precise upper-bound is obtained by increasing the values of \(n\) and \(m\).
For example, if \(n=16\) and \(m=256\),
the upper-bound (obtained by running the tool reported in a later section)
is \(0.3398\cdots\). \qed
\end{example}
\fi
\subsubsection{Computing an upper-bound for a binary function}
\label{sec:ub:binary}
We now consider a fixpoint equation of the form \(f(x_1,x_2)=e\), where
\(x_1\) and \(x_2\) are such that \(0\leq x_1,x_2\), and \(x_1+x_2\leq 1\).
Such an equation is obtained from an order-2 \pHORS{} by using the fixpoint characterization in 
the previous section.
A new difficulty compared with the unary case is that 
\(f(x_1,x_2)\) may take a value outside \([0,1]\), or may even be undefined
for \((x_1,x_2)\in [0,1]\times[0,1]\) such that
\(x_1+x_2>1\). Figure~\ref{fig:multivariate} shows how we discretize the domain of \(f\).
The grey-colored and red-colored areas show the valid domain of \(f\), for which we wish to approximate
$f(x_1,x_2)$ using the values at discrete points.
An upper-bound of the value of \(f\) at a point \((x_1,x_2)\) in the grey area can be
obtained by (pointwise) linear interpolations from
(upper-bounds of) the values at the surrounding four points,
i.e.,
\((\frac{i_1}{n},\frac{i_2}{n}),
(\frac{i_1}{n},\frac{i_2+1}{n}),
(\frac{i_1+1}{n},\frac{i_2}{n}), (\frac{i_1+1}{n},\frac{i_2+1}{n})\)
where \(x_1\in [\frac{i_1}{n},\frac{i_1+1}{n}]\)
and \(x_2\in [\frac{i_2}{n},\frac{i_2+1}{n}]\),
as follows.
\[
\begin{array}{l}
\hat{f}(x_1,x_2) \\=
(i_2+1-nx_2)\hat{f}(x_1,\frac{i_2}{n}) + (nx_2-i_2)\hat{f}(x_1,\frac{i_2+1}{n})\\
= 
(i_2+1-nx_2)(i_1+1-nx_1)\hat{f}(\frac{i_1}{n},\frac{i_2}{n}) \\\qquad+
(i_2+1-nx_2)(nx_1-i_1)\hat{f}(\frac{i_1+1}{n},\frac{i_2}{n})\\\qquad +
(nx_2-i_2)(i_1+1-nx_1)\hat{f}(\frac{i_1}{n},\frac{i_2+1}{n}) \\\qquad+
(nx_2-i_2)(nx_1-i_1)\hat{f}(\frac{i_1+1}{n},\frac{i_2+1}{n}).
  \end{array}
\]
Note that \(\hat{f}(x,y)\geq f(x,y)\) at the four points imply that
\(\hat{f}(x_1,x_2)\geq f(x_1,x_2)\), because
\(f(x,y)\) is convex on each of \(x\) and \(y\)
(recall Section~\ref{sec:properties}).

\begin{figure}[tb]
\begin{center}
\includegraphics[scale=0.25]{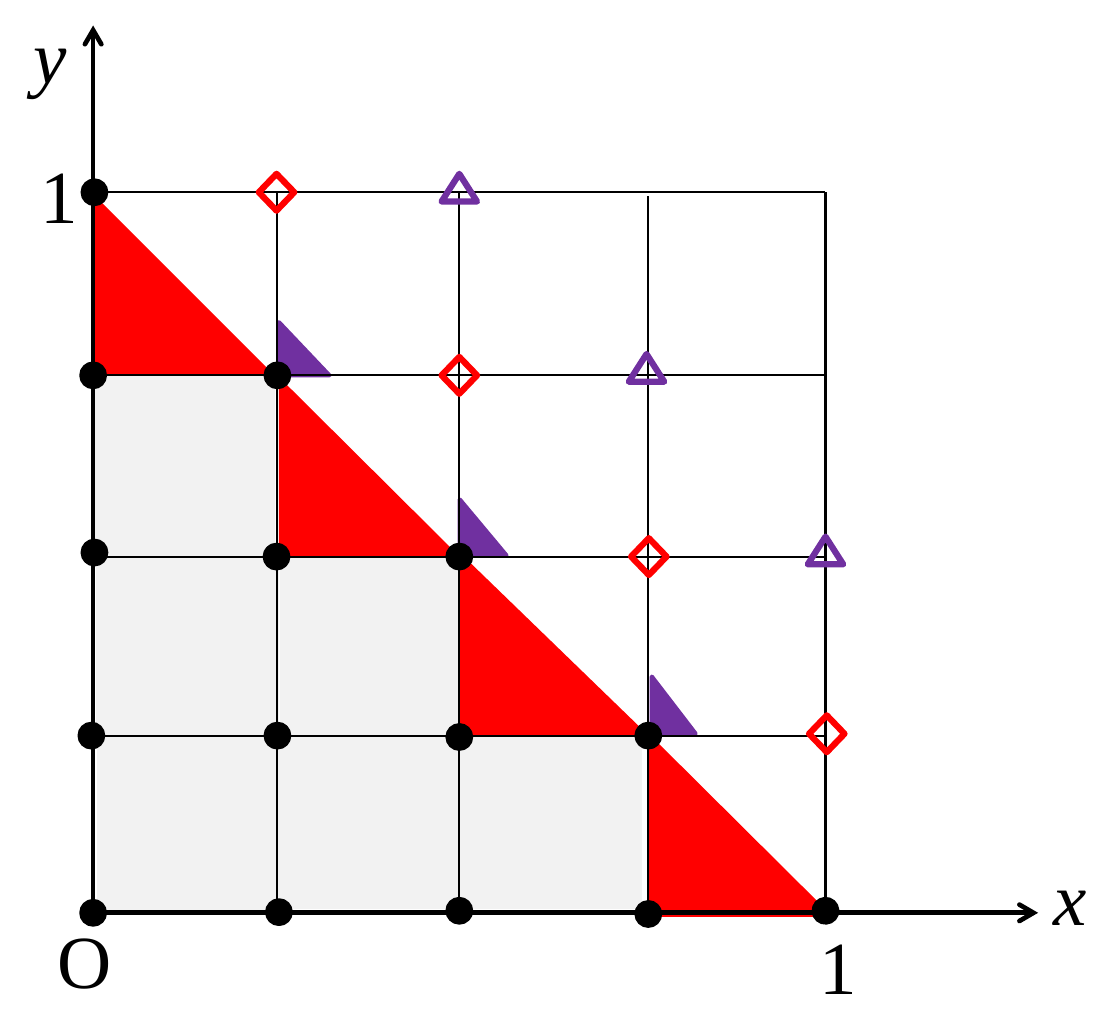}
\end{center}
\caption{Discretization in the case of a binary function \(f(x,y)\), whose domain is
\(\set{(x,y) \mid 0\leq x, 0\leq y, x+y\leq 1}\) (i.e., the grey and red areas). 
Outside the domain, the value of \(f(x,y)\) 
may not belong to \([0,1]\), or may be even undefined. The value at a point in the grey area
can be estimated by using those at discrete points \(\bullet\). To estimate the value
at a point in the red area, we also need the value of \(f(x,y)\) at a point marked by \(\textcolor{red}{\Diamond}\).}
\label{fig:multivariate}
\end{figure}

A difficulty is that to estimate the value of $f$ at a point in the
red area, we need the value at a red
point \textcolor{red}{\(\Diamond\)}, but the value of \(f\) at the red
point may be greater than 1 or even \(\infty\), being outside $f$'s
domain. To this end, we discretize the codomain of \(f\)
to \(\set{0,\frac{1}{m},\ldots,\frac{mh-1}{m},h,\infty}\) (instead
of \(\set{0,\frac{1}{m},\ldots,\frac{m-1}{m},1}\)) for some \(h\geq
1\). Any value greater than \(h\) is approximated to \(\infty\). The
value at a point in the red area is then approximated in the same way
as for the case of a point in the grey area, except that
if \(\hat{f}(\frac{i_1+1}{m},\frac{i_2+1}{m})=
\infty\), then \(\hat{f}(x_1,x_2)=1\).

A further complication arises when the equation contains
function compositions, as in \(f(x_1,x_2)=E[f(f(x_1,x_2),x_2)]\)
where \(E\) denotes some context. In this case, the point \((\hat{f}(x_1,x_2),x_2)\) may even be outside
the area surrounded by \textcolor{red}{\(\Diamond\)} and \(\bullet\)-points
either if \((x_1,x_2)\) is a \textcolor{red}{\(\Diamond\)}-point or
if \((x_1,x_2)\) is a \(\bullet\)-point but \(\hat{f}(x_1,x_2)+x_2\) is too large
due to an overapproximation. In such a case, \((\hat{f}(x_1,x_2),x_2)\) belongs to the purple area 
(lower left triangles) in the figure.
To this end, we also compute (upper-bounds of) the values at points marked by
\textcolor{purple}{\(\triangle\)} and use them to estimate the value at
a point in the purple area. 
If the point \((\hat{f}(x_1,x_2),x_2)\) is even outside the area surrounded by
\(\bullet\), \textcolor{red}{\(\Diamond\)}, or \textcolor{purple}{\(\triangle\)}, then 
 we use \(\infty\) as an upper-bound of \(f(f(x_1,x_2),x_2)\)
if \((x_1,x_2)\) is a \textcolor{red}{\(\Diamond\)}-, or \textcolor{purple}{\(\triangle\)}-point, 
and \(1\) if  \((x_1,x_2)\) is a \(\bullet\)-point.

Except the above differences, the overall algorithm is similar to the unary case
in Figure~\ref{fig:algo-unary}, and essentially the same soundness argument
as Lemma~\ref{lem:ub-soundness} applies.

\begin{example}
\label{ex:ub-binary}
Consider \(f(x_1,x_2)=x_1+x_2(f(x_1,x_2))^2\), and let \(n=m=2\).
The value of \(r=\left(\begin{array}{ccc}r_{0,2}&r_{1,2}&r_{2,2}\\
r_{0,1}&r_{1,1}&r_{2,1}\\
r_{0,0}&r_{1,0}&r_{2,0}
\end{array}\right)\),
where \(r_{i,j}\) is an upper-bound of the value of \(f(\frac{i}{2},\frac{j}{2})\),
changes as follows.
\[
\begin{array}{l}
\left(\begin{array}{ccc}0&0&0\\
0&0&0\\
0&0&0
\end{array}\right)
\longrightarrow
\left(\begin{array}{ccc}0&0.5 & 1\\
0& 0.5& 1\\
0&0.5& 1
\end{array}\right)
\iftwocol \\\else\fi
\longrightarrow
\left(\begin{array}{ccc}0& 1 & \infty\\
0& 1& \infty\\
0&0.5& 1
\end{array}\right)
\longrightarrow
\left(\begin{array}{ccc}0& \infty & \infty\\
0& 1& \infty\\
0&0.5& 1
\end{array}\right)
\end{array}
\]
Thus, for example, \(f(0, 0.5)\) 
and \(f(0.3,0.3)\) are overapproximated respectively by \(0\) and 
\(0.36 \hat{f}(0.5,0.5)+0.24 \hat{f}(0.5,0)+0.24\hat{f}(0,0.5)+0.16\hat{f}(0,0) = 0.48\). 
The exact values for \(f(0,0)\) and \(f(0.3,0.3)\) are \(0\) and \(\frac{1}{3}\); so
the upper-bound for \(f(0.3,0.3)\) is sound but imprecise.
By choosing \(n=16\) and \(m=256\), we obtain \(0.3359\cdots\) as an upper-bound of 
\(f(0.3,0.3)\). 

This is an example where the values of \(f\) at
red points in Figure~\ref{fig:multivariate} are \(\infty\).
The exact value of \(f(x_1,x_2)\) for general \(x_1\) and \(x_2\) is given by:
\[
f(x_1,x_2) = \left\{ \begin{array}{ll}x_1 & \mbox{if \(x_2=0\)}\\
  \frac{1-\sqrt{1-4x_1x_2}}{2x_2} & \mbox{if $x_1>0$}
\end{array}\right.
\]
Thus, 
\(f(0.5,0.5)=1\), but \(f(x_1,x_2)\) is undefined whenever \(x_1x_2>0.25\); in particular
 \(f(0.5,1)\) and \(f(1,0.5)\) (the values of red points in Figure~\ref{fig:multivariate}
for the case \(n=2\)) are undefined.
\qed
\end{example}

\subsubsection{Computing an Upper-Bound: General Case}
\label{sec:ub:general}

The binary case discussed above can be easily extended to handle the general case,
where the goal is to estimate the value of \(f_1(c_1,\ldots,c_{\ell_1})\)
for the least solution of the fixpoint equations:
\[
f_1(x_1,\ldots,x_{\ell_1}) = e_1\qquad\cdots\qquad f_k(x_1,\ldots,x_{\ell_k})=e_k.
\]
Here, the formal arguments 
\(x_1,\ldots,x_{\ell_i}\) of each function \(f_i\) are partitioned into
several groups \((x_1,\ldots,x_{d_{i,1}}),(x_{d_{i,1}+1},\ldots,x_{d_{i,2}}),\ldots,
(x_{d_{g_{i}-1}+1},\ldots,x_{\ell_i})\), so that the sum of the values of the variables
in each group ranges over \([0,1]\).
Following the binary case, we discretize the domain so that each variable ranges over
\(\set{0,\frac{1}{n},\ldots,\frac{n-1}{n},1}\), where the variables
in each group \((x_{d_{i,j-1}+1},\ldots,x_{d_{i,j}})\) are constrained by
\(x_{d_{i,j-1}+1}+\cdots+x_{d_{i,j}}\leq \frac{n+2}{n}\).
Note that we choose \(\frac{n+2}{n}\) instead of \(1\) as the upper-bound of the sum, to include
the points 
\textcolor{red}{\(\Diamond\)} and \textcolor{purple}{\(\triangle\)}
in Figure~\ref{fig:multivariate}. We write \(D_i\) for the discretized domain of
function \(f_i\), and \(D_i'\) for the subset of \(D_i\) where 
the variables in each group are constrained by \(x_{d_{i,j-1}+1}+\cdots+x_{d_{i,j}}\leq 1\);
note that \(f_i(x_1,\ldots,x_{\ell_i})\in [0,1]\) for \((x_1,\ldots,x_{\ell_i})\in D_i'\),
but \(f_i(x_1,\ldots,x_{\ell_i})\) may be greater than \(1\) or undefined 
for \((x_1,\ldots,x_{\ell_i})\in D_i\setminus D_i'\).
We also write \(\overline{D_i}\) for the set 
\(\set{(x_1,\ldots,x_{\ell_i})\mid (\ceil{nx_1}/n,\ldots,\ceil{nx_{\ell_i}}/n)\in D_i}\)
(i.e., the set of points for which the value of \(f_i\) can be approximated by using
values at points in \(D_i\)).

The pseudo code for computing an upper-bound of \(f_1(c_1,\ldots,c_{\ell_1})\)
is given in Figure~\ref{fig:algo-general}.
On the 9th line (``\(\texttt{if }\vec{v}\in D_i'\texttt{ then ...}\)''), 
we also make use of the constraint that \(\Sigma_{f'\in\fgroup(f_i)} f'(\seq{v})\) ranges over \([0,1]\)
if \(\seq{v}\) belongs to the valid domain \(D_i'\)
(recall the 3rd property in Section~\ref{sec:properties}). We assume that the procedure
\(\texttt{lb}(f',\seq{v})\) returns a sound lower-bound of \(f'(\seq{v})\), e.g., by using Kleene iteration.
See Remark~\ref{rem:ub-by-lb} to understand the need for this additional twist.
\iftwocol
\begin{figure*}[tb]
\else
\begin{figure}[tb]
\fi
\begin{alltt}
main(\(e_1,\ldots,e_k\), \(\vec{c}\))\{
  \(\rho\) := \([f_1\mapsto [\vec{0}],\ldots,f_k\mapsto [\vec{0}]]\);
      (* \(\rho(f_i)\) is an array indexed by each element of \(D_i\) *)
  \(\rho'\) := \([f_1\mapsto [\vec{1}],\ldots,f_k\mapsto [\vec{1}]]\); (* dummy *)
  while not(\(\rho\)=\(\rho'\)) do \{
    \(\rho'\) := \(\rho\); (* copy the contents *)
    for each \(i\in\set{0,\ldots,k}\) do
       for each \(\vec{v}\in{}D_i\) do 
          let r = eval(\(e_i\), \(\rho\set{\vec{x}\mapsto \vec{v}}\), \(\vec{v}\stackrel{?}{\in}D'_i\)) in
            if \(\vec{v}\in D'_i\) then \(\rho'(f_i)[\vec{v}]\) := \(\alpha_h\)(min(r, \(1-\Sigma_{f'\in\fgroup(f_i)\setminus\{f\}}\) lb(\(f'\), \(\vec{v}\)))) 
            else \(\rho'(f_i)[\vec{v}]\) := \(\alpha_h\)(r);
  return apply(\(\rho(f_1)\), \(\vec{c}\)); \}

eval(\(e\), \(\rho\), \(b\))\{ 
 (* \(b\) represents whether we are computing the value of \(f_i\) in the valid
    domain; in that case, the value of \(e\) should range over \([0,1]\). *)
  let \(r\) = 
    match \(e\) with
       \(x\) \(\to\) \(\rho(x)\) | \(c\) \(\to\) \(c\)
     | \(f_i(\vec{e'})\) \(\to\) let \(\vec{v}\) = eval(\(\vec{e'}\),\(\rho\),\(b\)) in
                 if \(\vec{v}\not\in \overline{D_i}\) then \(\infty\) else apply(\(\rho(f_i)\), \(\vec{v}\))
     | \(e_1+e_2\) \(\to\) eval(\(e_1\), \(\rho\), \(b\))+eval(\(e_2\), \(\rho\), \(b\))
     | \(e_1\cdot{}e_2\) \(\to\) eval(\(e_1\), \(\rho\), \(b\))\(\cdot\)eval(\(e_2\), \(\rho\), \(b\)) 
  in if \(b\) then return min(\(r\),1) else return \(r\) \}
\end{alltt}
\caption{Pseudo code for computing an upper-bound for the general case}
\label{fig:algo-general}
\iftwocol
\end{figure*}
\else
\end{figure}
\fi

The function \(\alpha_h\) takes a real value
(or \(\infty\)) \(x\), and returns the least element in 
\(\set{0,\frac{1}{n},\ldots,\frac{mh-1}{m},h,\infty}\) that is no less than \(x\).
The function \texttt{apply} in the figure takes the current approximations of
values of \(f_i\) at the points \(D_i\) and the arguments \(\vec{v}\in \overline{D_i}\), and 
returns an approximation of \(f_i(\vec{v})\). 
It is given by \(\hat{f}(\vec{v})\), where:
\[
\begin{array}{l}
\hat{f_i}(x_1,\ldots,x_{\ell_i}) =\iftwocol\\\fi \sum_{b_1,\ldots,b_{\ell_i}\in\set{0,1}} p_1^{b_1}(1-p_1)^{1-b_1}\cdots
p_{\ell_i}^{b_{\ell_i}}(1-p_{\ell_i})^{1-b_{\ell_i}}\iftwocol\\\hfill\fi\hat{f_i}\left(\frac{i_1+b_1}{n},\ldots,\frac{i_{\ell_i}+b_{\ell_i}}{n}\right)
\end{array}
\]
Here, \(i_j = \floor{nx_j}\),
\(p_j=nx-i_j\), and
\(\hat{f_i}(\frac{i_1+b_1}{n},\ldots,\frac{i_m+b_m}{n})\) is the current
approximation of the value of \(f_i\) at 
\((\frac{i_1+b_1}{n},\ldots,\frac{i_m+b_m}{n})\in D_i\).
The function \(\hat{f}\) above is obtained by applying linear interpolations coordinate-wise.

\begin{remark}\rm
\label{rem:ub-by-lb}
To see the motivation for the 9th line in Figure~\ref{fig:algo-general}, consider the following
fixpoint equations:
\begin{align*}
S &= f_1()\\
f_1() &= 0.5 \cdot (f_1()\cdot f_1() + f_2()\cdot f_2())\\
f_2() &= 0.5 + f_1()\cdot f_2().
\end{align*}
They are obtained from the following order-1 PHORS \(\GRAM_{\mathtt{treeeven}}\):
\begin{align*}
S\,z &= F\,z\,\Omega\\
F\;x_1\;x_2 &= x_2 \C{p} F\;(F\;x_1\;x_2)\;(F\;x_2\;x_1),
\end{align*}
where \(p=0.5\), and
\(f_1()\) (\(f_2()\), resp.) represents the probability that
\(x_1\) (\(x_2\), resp.) is used by \(F\).
This PHORS is actually a variation
of \(\GRAM_6\) from Example~\ref{ex:listgen-even} (with manual optimization),
whose termination probability represents the probability that a program that randomly generates 
binary trees (instead of lists, unlike in the case of Example~\ref{ex:listgen-even}) 
contains an even number of leaves.
Since the events that \(F\) uses the first and second arguments are mutually exclusive, we have
the constraint \(f_1()+f_2()\leq 1\). The exact solution for the equations above is
\(f_1()=1-\frac{1}{\sqrt{2}}\) and \(f_2()=\frac{1}{\sqrt{2}}\). Since their lower-bounds can
be computed with arbitrary precision, thanks to the part 
\texttt{\(1-\Sigma_{f'\in\fgroup(f_i)\setminus\{f\}}\) lb(\(f'\), \(\vec{v}\))} of
the 9th line of Figure~\ref{fig:algo-general}, we can also compute upper-bounds with
arbitrary precision (as upper-bounds of \(f_1()\) and \(f_2()\) are respectively provided
by \(1-\texttt{lb}(f_2,())\) and \(1-\texttt{lb}(f_1,())\)).

If the then-clause were the same as the else-clause on the 10th line, then we would not get 
a precise upper-bound for the following reason. 
When the main loop in Figure~\ref{fig:algo-general} stops, 
upper-bounds \(\overline{f}_1()\) and \(\overline{f}_2()\) 
must either have reached the maximal value \(1\), or satisfy:
\[
\begin{array}{l}
 \overline{f}_1() \ge 0.5 \cdot (\overline{f}_1()\cdot \overline{f}_1() + \overline{f}_2()\cdot \overline{f}_2())\\
\overline{f}_2() \ge 0.5 + \overline{f}_1()\cdot \overline{f}_2().
\end{array}
\]
These conditions imply that:
\[\overline{f}_1()+\overline{f}_2() \ge 
0.5 \cdot (\overline{f}_1()\cdot \overline{f}_1() + \overline{f}_2()\cdot \overline{f}_2())
+0.5 + \overline{f}_1()\cdot \overline{f}_2(),\]
i.e., 
\[0\ge (\overline{f}_1()+ \overline{f}_2()-1)^2,\]
which is equivalent to \(\overline{f}_1()+ \overline{f}_2()=1\).
Thus, unless the co-domain of \(\alpha_h\) contains the exact values 
\(1-\frac{1}{\sqrt{2}}\) and \(\frac{1}{\sqrt{2}}\), 
the main loop would only return the imprecise upper-bound \(\overline{f}_1()=\overline{f}_2()=1\). \qed
\end{remark}

\subsection{Order-$n$ Case}\label{sec:ub:higherorder}
We now briefly discuss how to extend the method discussed above
to obtain a sound (but incomplete) method for overapproximating the
termination probability of \pHORS{} of order greater than 2.
Recall that by the fixpoint characterization given in Section~\ref{sec:order-n-1-equation},
it suffices to overapproximate the least solution of equations of the form
\( \vec{f} = {\F}(\vec{f})\)
where \(\vec{f}\) is a tuple of order-(\(n-1\)) functions on reals.

The abstract interpretation framework~\cite{DBLP:conf/popl/Cousot97}
provides
a sound but incomplete methodology: the reason why we decided to
slightly divert from it in Section~\ref{sec:ub-by-disc} is that this
allows us to use piecewise linear functions, which are more precise.
We first recall a basic principle of abstract interpretation.
Let \((C, \Leq_C,\bot_C)\) and \((A, \Leq_A,\bot_A)\) be \(\omega\)-cpos.
Suppose that \(\alpha: (C,\Leq_C)\to(A, \Leq_A)\) and 
\(\gamma:(A, \Leq_A)\to (C,\Leq_C)\) are continuous (hence also monotonic) such that
\(\alpha(\gamma(a))=a\) for every \(a\in A\), and 
\( c\Leq_C \gamma(\alpha(c))\) for every \(c\in C\).
Suppose also that \(\F\) is a continuous 
function from \((C,\Leq_C,\bot_C)\) to \((C,\Leq_C,\bot_C)\).
Let \(\abs{\F}:(A,\Leq_A,\bot_A)\to(A,{\Leq_A},\bot_A)\) be 
\(\lambda x\in A.\alpha(\F(\gamma(x)))\), which is an ``abstract version'' of $\F$. 
Note that \(\abs{\F}\) is also continuous. Then, we have:
\begin{proposition}
  \label{prop:abs-fixpoint}
\(\LFP(\F) \Leq_C \gamma(\LFP(\abs{\F}))\). \end{proposition}
\noindent
This result is standard (see, e.g., \cite{DBLP:conf/popl/Cousot97}, Proposition~18) but we provide a proof 
for the convenience of the reader.
\commentout{
\begin{proof}[Proof of Lemma~\ref{lemma:pointwise-convex}]
 The pointwise convexity of $\F^m(\bot)$ follows from the fact
that, following our first observation, it is (a tuple of) 
multi-variate polynomials with non-negative integer coefficients.
Let us treat the case of a multi-variate polynomial of the shape
$$
\sum_{(i_1,\ldots,i_n) \in \{0,\ldots,d\}^n} \ a_{(i_1,\ldots,i_n)}\, x_1^{i_1} \ldots x_j^{i_j} 
\ldots x_n^{i_n}
$$
Fix $(x_1,\ldots,x_{j-1},x_{j+1},\ldots,x_n) \in [0,1]^{n-1}$.
Remark now that the pointwise convexity on the $j$-th component of this multivariate polynomial is equivalent to the convexity of the function
$$
x_j \mapsto \sum_{(i_1,\ldots,i_n) \in \{0,\ldots,d\}^n} \ a_{(i_1,\ldots,i_n)}\, x_1^{i_1} \ldots x_j^{i_j} 
\ldots x_n^{i_n}
$$
This function is polynomial, and therefore smooth. It follows that it is convex over $[0,1]$ iff its second derivative is positive over $[0,1]$.
But this is obvious as the polynomials we consider are built from non-negative coefficients and that $x_l^m \geq 0$ for every $l \in \{1,\ldots,n\}$
and every $m \in \{0,\ldots,d\}$. It follows that the polynomial we consider is pointwise convex. The extension to tuples is straightforward.

The pointwise convexity of $\mathbf{lfp}(\F)$ follows from the fact that, for every $m$,
\[(1-p)\mathbf{lfp}(\F)(\vec{x})+p\mathbf{lfp}(\F)(\vec{y})\geq
(1-p)\F^m(\bot)(\vec{x})+p\F^m(\bot)(\vec{y})\geq
\F^m(\bot)((1-p)\vec{x}+p\vec{y})\] 
and we can then take the supremum to conclude.
\end{proof}
}

\ifacm
\begin{proof}[Proof of Proposition~\ref{prop:abs-fixpoint}]
\else
\begin{proofn}{Proof of Proposition~\ref{prop:abs-fixpoint}}
\fi
By the monotonicity of \(\F\) and \(\abs{\F}\), we have:
\(\bot_C\Leq_C \F(\bot_C) \Leq_C \F^2(\bot_C) \Leq_C \cdots
\) and 
\(\bot_A\Leq_A \abs{\F}(\bot_A) \Leq_A {\abs{\F}}^2(\bot_A) \Leq_A \cdots\);
hence both \({\Lub_C} \set{ \F^i(\bot_C)\mid {i\in \omega}}\) and 
\({\Lub_A}\set{ {\abs{\F}}^i(\bot_A)\mid i\in \omega}\) exist, and by
the \(\omega\)-continuity of
\(\F\) and \(\abs{\F}\), they are the least fixpoints of \(\F\) and \(\abs{\F}\) respectively.
Therefore, it suffices to show that 
\(\F^i(\bot_C) \Leq_C \gamma({\abs{\F}}^i(\bot_C))\). The proof proceeds by induction on \(i\).
The base case \(i=0\) is trivial. If \(i>0\), we have:
\begin{align*}
\gamma({\abs{\F}}^i(\bot_C))
&= \gamma(\alpha(\F(\gamma({\abs{\F}}^{i-1}(\bot_C))))) \iftwocol \\ \fi
&&\mbox{(by the definition of \(\abs{\F}\))}\\
&\sqsupseteq_C \F(\gamma({\abs{\F}}^{i-1}(\bot_C))) &&\mbox{(by \(\gamma(\alpha(x))\sqsupseteq_C x\))}\\
&\sqsupseteq_C \F(\F^{i-1}(\bot_C)) &&\mbox{(by induction hypothesis)}\\
&= \F^i(\bot_C).
\end{align*}
\ifacm
\end{proof}
\else
\end{proofn}
\fi
 
By the proposition above, to overapproximate the least fixpoint of \(\F\), it suffices
to find an appropriate abstract domain \((A,\Leq_A,\bot_A)\) and \(\alpha,\gamma\)
that satisfy the conditions above, so that the least fixpoint of \(\abs{\F}\) is easily computable.
In the case of overapproximation of the termination probability of order-\(n\) \pHORSs{}, we need to set
up an abstract domain \((A,\Leq_A,\bot_A)\) for 
a tuple of order-(\(n-1\)) functions on reals.
A simple solution (that is probably too naive in practice) is to use
the abstract domain consisting of higher-order step functions, inductively defined by:
\begin{align*}
A^{\realt} &= \set{\frac{0}{m},\frac{1}{m},\ldots,\frac{m-1}{m},\frac{m}{m}, \infty}\\
\Leq_{A^{\realt}} &= \set{(\frac{i}{m},\frac{j}{m}) \mid 0\leq i\leq j\leq m} \cup \set{(\frac{i}{m},\infty) \mid 0\leq i\leq m}\\
\alpha^{\realt}(x) &= \left\{\begin{array}{ll}
   \frac{i}{m} &\mbox{if \(\frac{i-1}{m}\leq x\leq \frac{i}{m}\)}\\
   \infty &\mbox{if \(x>1\)}
\end{array}
\right.\\
\gamma^{\realt}(x) &= x\\
A^{\tau_1\to\tau_2} &=
  \set{f\in A^{\tau_1}\to A^{\tau_2}\mid \mbox{$f$ is monotonic}}\\
\Leq_{A^{\tau_1\to\tau_2}} &= 
  \set{(f_1,f_2)\in A^{\tau_1\to\tau_2}\times A^{\tau_1\to\tau_2} \iftwocol\\\hfill\fi\mid 
 \forall x\in A^{\tau_1\to\tau_2}. f_1\,x \Leq_{A^{\tau_2}}f_2\,x }\\
\alpha^{\tau_1\to\tau_2}(f) &=
  \set{(y, \alpha^{\tau_2}(f(\gamma^{\tau_1}(y)))) \mid y\in A^{\tau_1}}\\
\gamma^{\tau_1\to\tau_2}(f') &=
  \set{(x, \gamma^{\tau_2}(f'(\alpha^{\tau_1}(x)))) \mid x\in C^{\tau_1}}.
\end{align*}
Here, the concrete domain \(C^\tau\) denotes \(\sem{\tau}\) in
Section~\ref{sec:ho-fixpoint}.
Then, \(\alpha^{\tau}\) and \(\beta^{\tau}\) satisfy the required conditions
(\(\alpha(\gamma(a))=a\) and \( c\Leq_C \gamma(\alpha(c))\)). Since \(A^{\tau}\) is finite,
we can effectively compute \(\LFP(\abs{\F})\).

We note, however, that the above approach has the following shortcomings. First,
although \(A^{\tau}\) is finite, its size is too large: \(k\)-fold exponential for order-\(k\)
type \(\tau\). As in the case of non-probabilistic HORS model checking~\cite{Kobayashi09PPDP,Kobayashi13horsat,Ramsay14POPL}, therefore, we need a practical algorithm that avoids eager enumeration of
abstract elements.
Second, due to the use of step functions, the obtained upper-bound will be
too imprecise. To see why step functions suffer from the incompleteness, consider the equations:
\(s = f(\frac{1}{2})\) and \(f(x)=\frac{1}{2}x + f(\frac{1}{2}x)\).
The exact least solution is \(s=\frac{1}{2}\) and \(f(x)=x\).
With step functions (where \(\frac{1}{n}\) is the size of each interval), however,
the abstract value \(\hat{f}(\frac{1}{n})\) must be no less than
\(\frac{1}{2}\frac{1}{n}+f(\frac{1}{2}\cdot\frac{1}{n})\),
but \(f(\frac{1}{2}\cdot\frac{1}{n})\)  is overapproximated as
\(\hat{f}(\frac{1}{n})\) (because \(\frac{1}{2n}\) belongs to the interval
\((0,\frac{1}{n}]\)). Therefore, 
  \(\hat{f}(\frac{1}{n})\) should be no less than
\(\frac{1}{2n}+\hat{f}(\frac{1}{n})\), which is impossible.
Thus, the computation diverges and \(1\) is obtained as an obvious upper bound.

The step functions only use monotonicity of the least solution of
fixpoint equations.
As in the use of stepwise multilinear functions in
the case of order-1 equations (for order-2 \pHORSs{}), exploiting 
an additional property like convexity would be important for obtaining a more precise method;
this is left for future work.

\section{Experiments}
\label{sec:exp}

We have implemented a prototype tool to compute lower/upper bounds of
the least solution of order-1 fixpoint equations (that are supposed to
have been obtained from order-2 or order-1 \pHORS{} by using the translations in 
Section~\ref{sec:fixpoint} modulo some simplifications;
we have not yet implemented the translators from \pHORS{} to fixpoint equations, which is easy but tedious).
The computation of a lower bound is based on naive Kleene iterations, and that of an upper-bound
is based on the method discussed in Section~\ref{sec:ub-by-disc}.
The tool uses floating point arithmetic, and ignores rounding errors.

We have tested the tool on several small but tricky examples.
The experimental results are summarized in Table~\ref{tab:exp}.
The column ``equations'' lists the names of systems of equations.
The column
``\#iter'' shows the number of Kleene iterations used for computing a
lower-bound. The columns  ``\#dom'' and ``\#codom'' show
the numbers of partitions of the interval \([0,1]\) for the domain
and codomain of a function respectively. The default values
for them were set to 12, 16, and 512, respectively in the experiment;
they were, however, adjusted for some of the equations.
The columns ``l.b.'' and ``u.b.'' are lower/upper bounds computed by
the tool. The lower (upper, resp.) bounds shown in the table have been obtained by 
rounding down (up, resp.) the outputs of the tool to 3-decimal places.
The column ``step'' shows the upper-bounds obtained by using step functions
instead of piecewise linear functions; this column has been prepared to confirm the advantage
of piecewise linear functions over step functions.
The column ``exact'' shows the exact value of the least solution
when we know it. The column ``time'' shows the total time for
computing both lower and upper bounds (excluding the time for ``step'').

\begin{table}[tbp]
\caption{Experimental results (times are in seconds).}
\label{tab:exp}
\begin{center}
\begin{tabular}{|l|r|r|r|r|r|r|r|r|}
\hline 
equations & \#iter & \#dom & \#codom & l.b. & u.b. & step & exact & time \\
\hline \hline
Ex2.3-1 & 12 & 16 & 512 & 0.333 & 0.336 & 1.0 & \(\frac{1}{3}\) & 0.010 \\\hline
Ex2.3-0 & 12 &  16 & 512 &0.333 & 0.334 & 0.334 & \(\frac{1}{3}\)& 0.008 \\\hline
Ex2.3-v1 & 12 &  16 & 512 & 0.312 & 0.315 & 0.365 & - &0.005 \\\hline
Ex2.3-v2 & 12 &  16 & 512 & 0.262 & 0.266 & 0.321 & - &0.022 \\\hline
Ex2.3-v3 & 12 &  16 & 512 & 0.263 & 0.266 & 0.309 & - &0.01\\\hline
Ex2.4 & 12 &  16 & 512 & 0.320 & 0.323 & 0.329 & - &0.011 \\\hline
Double & 12 &  16 & 512 & 0.649 & 0.653 &1.0 & - & 0.010\\\hline
Listgen & 15 &  16 & 512 & 0.999 & 1.0 &1.0 & 1.0 & 0.009\\\hline
Treegen & 15 &  64 & 4096 & 0.618 & 0.619 &1.0 & $\frac{\sqrt{5}-1}{2}$ & 0.471\\\hline
Treegenp & 12 &  16 & 512 & 1.0 & 1.0 &1.0 & 1.0 & 0.011\\\hline
ListEven & 12 &  32 & 1024 & 0.666 & 0.667 & 0.667 & $\frac{2}{3}$ & 0.009\\\hline
ListEven2 & 12 &  16 & 512 & 0.749 & 0.75 & 0.75 & $\frac{3}{4}$ & 0.013\\\hline
Determinize & 12 &  16 & 512 & 0.993 & 1.0 &1.0 & 1.0 & 9.64\\\hline
TreeEven(0.5) & 15 & 64  & 4096 & 0.286 & 0.299 & 0.300 & \(1-\frac{1}{\sqrt{2}}\) & 0.050\\\hline
TreeEven(0.49) & 15 & 64 & 4096 & 0.276 & 0.280 & 0.280 & 0.2774\(\cdots\)& 0.052\\\hline
TreeEven(0.51) & 15 & 64 & 4096 & 0.287 & 0.290 & 0.290 & 0.2887\(\cdots\) & 0.055\\\hline
Ex5.4(0,0) & 12 &  16 & 512 & 0.0 & 0.0 & 0.0 & 0 & 0.008\\\hline
Ex5.4(0.3,0.3) & 12 &  16 & 512 & 0.333 & 0.336 & 0.35 & \(\frac{1}{3}\) & 0.007\\\hline
Ex5.4(0.5,0.5) & 10000 &  16 & 512 & 0.999 & 1.0 & 1.0 & 1 & 0.010\\\hline
Discont(0,1) &  12 & 16 & 512 & 0.0 & 0.0 & 0.0 & 0 & 0.006\\\hline
Discont(0.01,0.99) & 1000 & 16 & 512 & 0.999 & 1.0 & 1.0 & 1 & 0.006\\\hline
Incomp & 10000 & 16 & 512 & 0.299 & 1.0 & 1.0 & 0.3 & 0.003\\\hline
Incomp & 10000 &10 & 100 & 0.299 & 0.3 & 0.3 & 0.3 & 0.003\\\hline
Incomp2 & 12 & 16 & 512 & 0.249 & 1.0 & 1.0 & 0.25 & 0.003\\\hline
Incomp2 & 12 & 256 & 65536 & 0.249 & 1.0 & 1.0 & 0.25 & 2.87\\\hline
\end{tabular}
\end{center}
\end{table}

The equations ``Ex2.3-1'' and ``Ex2.3-0'' are 
order-1 and order-0 equations obtained from the \pHORS{} in
Example~\ref{ex:random-walk} (see also Examples~\ref{ex:order-n-equation} and \ref{ex:random-walk-eq})
by using the translations
in Sections~\ref{sec:order-n-equation} and
\ref{sec:order-n-1-equation} respectively; specifically,
``Ex2.3-1'' consists of \(s=f(1)\) and \(f(x)=0.25x+0.75f(f(x))\).
The equations ``Ex2.3-v1'', ``Ex2.3-v2'', and ``Ex2.3-v3'' are variations
of them, where the equation on \(f\) is replaced by
\(f(x)=0.25x+0.75f(f(x^2))\), 
\(f(x)=0.25x+0.75f(f(f(x^2)))\),  and 
\(f(x)=0.25x+0.75(f(x))^2\), respectively.
``Ex2.4'' is the equations obtained from the order-2 \pHORS{}
in Example~\ref{ex:order2-phors} (see also Example~\ref{ex:order2-phors-eq}).
The equations ``Double'' are those obtained from the following order-2 \pHORS{}:
\begin{align*}
S\ &=\ F\,H \\ 
H\,x\,y &= x \C{\frac{1}{2}} y\\
F\,g &= g\,\Te\, (F(D\,g))\\
D\,g\,x\,y &= g\,(g\,x\,y)\,y,
\end{align*}
with manual simplifications.
The equations ``Listgen'', ``Treegen'', and  ``Treegenp'' are from Example~\ref{ex:listgen},
corresponding to \(\GRAM_3\), \(\GRAM_4\) and \(\GRAM_5\), respectively.
The equations ``ListEven'' and ``ListEven2'' are from Example~\ref{ex:listgen-even}, and
``Determinize'' is from Example~\ref{ex:palgo}.
``TreeEven(\(p\))'' (for \(p\in\set{0.5,0.49,0.51}\)) is from Remark~\ref{rem:ub-by-lb}.
If we disable the trick (the one on line 9 in Figure 7) we discussed in the remark, 
the tool returns an imprecise upper-bound of \(1.0\) for \(p=0.5\) (for \(p=0.49\) and \(p=0.51\),
however, the tool can compute a precise upper-bound even without the trick).
The equations ``Ex5.4($x$,$y$)'' (for \((x,y)\in\set{(0,0), (0.3,0.3),(0.5,0.5)}\))
are from Example~\ref{ex:ub-binary}. 
The equations ``Discont($p$,$1-p$)'' consist of:
\(s=f(p,1-p)\) and \(f(x_0,x_1)=x_0+x_1f(x_0,x_1)\),
which is obtained from \pHORS{}:
\[ S = F\,G\qquad F\,g = g\,\Te\,(F\,g)\qquad G\,x_0\,x_1 = x_0\C{p}x_1.\]
Interestingly, \(f\) is discontinuous at \((0,1)\) (in the usual sense
of analysis in mathematics; it is still \(\omega\)-continuous
as functions on \(\omega\)-cpo's): the exact value of \(f\) is
given by:
\[f(x_0,x_1) = \left\{\begin{array}{ll}
  0 & \mbox{if $x_0=0$}\\
  \frac{x_0}{1-x_1}&\mbox{if $x_0>0$}.
  \end{array}
  \right.
  \]
The equations ``Incomp'' and ``Incomp2'' consist of:
\[ s = f(s)\qquad\qquad f\,x = x^2+0.4x+0.09,\]
and 
\[ s = f(s)\qquad\qquad f\,x = 0.5x^2+2f(0.5x)\]
respectively.
They do not correspond to any \pHORS{}
--- in fact, the value of \(f(1)\) for Incomp is \(1.49\), 
which does not make sense as a probability.
We have included them since they show a source of the possible incompleteness
of our method. Indeed, the tool fails to find precise upper bounds.
To see why the tool does not work for Incomp1 (with the the default values of \#dom and \#codom), note that since \(s = f(s) = s^2+0.4s+0.09\), 
 \(\hat{s} \geq \hat{s}^2+0.4\hat{s}+0.09\) must be satisfied for any valid upper-bound \(\hat{s}\).
However, \(\hat{s} \geq \hat{s}^2+0.4\hat{s}+0.09\)
 is equivalent to \(0\geq (\hat{s}-0.3)^2\), which is satisfied only by \(\hat{s}=0.3\). 
So, the only valid upper-bound for \(s\) is actually the exact one \(0.3\). But then
an upper-bound \(\hat{f}\) of \(f\) must satisfy \(\hat{f}(0.3)=0.3\), which can be found only if
the set of discrete values (used for abstracting the domain and codomain) contains
\(0.3\).
That is why the tool returns
 \(1\) (which is the largest value, assuming that \(s\) represents a probability) for the default
values of \#dom and \#codom. When we adjust them to \(10\) and \(100\) (so that \(0.3\) belongs
to the sets of abstract values of domains and codomains), 
the precise upper-bound (i.e., \(0.3\)) is obtained; 
this is, however, impossible in general, without knowing the exact solution a priori.

The reason for ``Incomp2'' is more subtle.
Notice that the least solution for
\[ s = f(s)\qquad f\,x = 0.5x^2+2f(0.5x)\]
is \(f(x)=x^2\). Let \(\frac{1}{n}\) be
 the size of each interval used for abstracting the domain.
Suppose that, at some point,
an upper-bound of \(f(\frac{1}{n})\) becomes \(\frac{c}{n^2}\).
Due to the linear interpolation (and since the value at \(x=0\) converges to
\(0\)), the value of \(f\) at \(0.5r\) (which belongs to the domain \((0,\frac{1}{n})\))
is overapproximated by
\(0.5 \cdot \frac{c}{n^2}\).
Thus, at the next iteration,
the upper-bound at \(\frac{1}{n}\) is further updated to a value greater than
\(0.5\frac{1}{n^2}+2\cdot 0.5\cdot \frac{c}{n^2} = \frac{c+0.5}{n^2}\). Thus,
the computation of an upper-bound for the value at \(x=\frac{1}{n}\) never converges.
In this case, changing the parameters \#dom and \#codom does not help.
We do not know, however, whether such situations occur in the fixpoint equations that arise
from actual order-2 \pHORSs{}; it is left for future work to see whether our method (or a minor 
modification of it) is actually complete (in the sense that upper-bounds can always 
be computed with arbitrary precision by increasing the parameters \#dom and \#codom).

To summarize, for all the \emph{valid} inputs (i.e., except ``Incomp''
and ``Incomp2'', which are invalid in the sense that
they do not correspond to \pHORS{}), our
tool (with piecewise linear functions) could properly compute lower/upper bounds. 
In contrast, from the column ``step'',
we can observe that the replacement of 
piecewise (multi)linear functions with step functions
not only worsens the precision (as in ``Ex2.3-v1'', ``Ex2.3-v2'', and ``Ex2.3-v3'') significantly,
but also makes the procedure obviously incomplete\footnote{As already mentioned, 
our method with piecewise linear functions may also be incomplete, but that does not
show up in the current benchmark set.}, as in ``Ex2.3-1'' and ``Double''
(recall the discussion on the incompleteness of step functions
in Section~\ref{sec:ub:higherorder}).

\section{Related Work}
\label{sec:related}

As already mentioned in Section~\ref{sec:intro}, this work is intimately
related to both probabilistic model checking, and higher-order model checking.
Let us give some hints on \emph{how} our work is related to the
two aforementioned research areas, without any hope to be exhaustive.

\smallbreak

\noindent
\textbf{Model checking of probabilistic recursive systems.}
Model checking of probabilistic systems with \emph{recursion} (but not higher-order functions),
such as recursive Markov chains and probabilistic pushdown systems,
has been actively studied~\cite{Etessami09,DBLP:journals/jacm/EtessamiY15,DBLP:journals/jcss/BrazdilBFK14}. Our \pHORS{} are strictly more expressive than those models, 
as witnessed by the undecidability
result from Section~\ref{sec:undecidability}, and the encoding of recursive Markov chains
into order-1 \pHORS{s} \iffull given in Appendix~\ref{sec:encoding-rmc}.
\else \cite{KDG19LICSfull}. \fi
Our fixpoint characterization of the termination probability of \pHORS{}
is a generalization of the fixpoint characterization of the
termination probability for recursive Markov
models~\cite{Etessami09} to arbitrary
orders.
Various methods have been studied for solving
the order-0 fixpoint equations
(or, polynomial equations) obtained from recursive
Markov chains~\cite{Etessami09,DBLP:conf/stoc/KieferLE07,DBLP:conf/stacs/EsparzaGK10}.
Interestingly, also in those methods,
computing an upper-bound of the least solution
is more involved than computing a lower-bound.
It is left for future work to investigate whether some of the ideas in those methods
can be used also for solving order-1 fixpoint equations.

\smallbreak

\noindent
\textbf{Termination of probabilistic infinite-data programs.}
Methods for computing the termination probabilities of infinite-data programs 
(with real-valued variables,
but without higher-order recursion) have also been actively studied, mainly
in the realm of imperative programs (see, as an
example,~\cite{bournezgarnier,esparza12,DBLP:conf/popl/FioritiH15,DBLP:conf/popl/ChatterjeeNZ17,Kaminski18,AvanziniDalLagoYamada18,chakarov13}); 
to the best of our knowledge, none of those methods deal with higher-order
programs, at least directly.
All these pieces of work present sound but \emph{incomplete} methodologies
for checking almost sure termination of programs. Incompleteness is of course inevitable
due to the Turing completeness of the underlying language considered.
In fact, 
\ifacm \citeN{Kaminski15} have 
\else
Kaminsiki and Katoen~\cite{Kaminski15} have 
\fi
shown that almost sure termination of probabilistic imperative
programs is $\Pi_2^0$-complete. Since their proof relies on Turing
completeness of the underlying language, it does not apply to the
setting of our model \pHORS{}, which is a probabilistic extension of a
Turing-\emph{incomplete} language, namely that of HORS.

\smallbreak

\noindent
\textbf{Model checking of higher-order programs.}
Model checking of (non-probabilistic) higher-order
programs has been an active topic of research
in the last fifteen years, with many positive results~\cite{Knapik02FOSSACS,Ong06LICS,Hague08LICS,Kobayashi13JACM,KO09LICS,KSU11PLDI,DBLP:conf/csl/GrelloisM15,DBLP:conf/mfcs/GrelloisM15,DBLP:conf/csl/TsukadaO14,Salvati11ICALP,DBLP:conf/stacs/Parys18}.
Strikingly, not only termination, but also a much larger class of
properties (those expressible in the modal $\mu$-calculus)
are known to be decidable
for ordinary (i.e. non-probabilistic) HORS. This is in stark
contrast with our undecidability result from Section~\ref{sec:undecidability}:
already at order-2 and for a very simple property like termination,
verification cannot be effectively solved.

\smallbreak

\noindent
\textbf{Probabilistic functional programs.}
Probabilistic functional programs have recently attracted the
attention of the programming language community, although
probabilistic $\lambda$-calculi have been known for forty years
now~\cite{SahebDjahromi,JonesPlotkin1989}.  Most of the work in this
field is concerned with operational semantics~\cite{dallagozorzi2012},
denotational semantics (see,
e.g.,~\cite{JUNG199870,danosharmer,ehrhardtassonpagani2014,DBLP:conf/lics/StatonYWHK16,DBLP:conf/lics/BacciFKMPS18}),
or program equivalence (see,
e.g.,~\cite{DLSA14,DBLP:conf/esop/CrubilleL14,DBLP:conf/popl/SangiorgiV16}),
which sometimes becomes decidable (e.g.~\cite{MurawskiO05}), but only when higher-order
recursion is forbidden.  The interest in probabilistic higher-order
functional languages stems from their use as a way of writing
probabilistic graphical models, as in languages like
\textsf{Church}~\cite{church} or \textsf{Anglican}~\cite{anglican}.
There are some studies to analyze the termination behavior of probabilistic
higher-order programs (with infinite data) by using types.
\ifacm \citeN{DBLP:conf/esop/LagoG17} 
\else
Dal Lago and Grellois~\cite{DBLP:conf/esop/LagoG17} 
\fi generalized
sized
types~\cite{hughes-pareto-sabry:sized-types,barthe-et-al:type-based-termination}
to obtain a sound but highly incomplete technique.
\ifacm \citeN{BreuvartDalLago} 
\else
Breuvart and Dal Lago~\cite{BreuvartDalLago} 
\fi
developed systems of intersection
types from which the termination probability of higher-order programs
can be inferred from (infinitely many) type derivations. This however does
\emph{not} lead to any practical verification methodology.

\smallbreak

\noindent
\textbf{Relevant proof techniques.}
Our  technique (of using the undecidability of Hilbert's 10th problem) for
proving the undecidability of almost sure termination
of order-2 \pHORS{} has been inspired by Kobayashi's proof of undecidability
of the inclusion between order-2 (non-probabilistic) word languages
and the Dyck language~\cite{KobayashiDyck}. Other undecidability results
on probabilistic systems include the undecidability of the emptiness of
probabilistic automata~\cite{DBLP:conf/icalp/GimbertO10}.
Their proof is based on the reduction from Post correspondence problem.
The technique does not seem applicable to our context.

\section{Conclusion}
\label{sec:conc}
We have introduced \pHORS{}, a probabilistic extension of higher-order recursion schemes, and studied the problem of computing their termination probability.
We have shown that almost sure termination is undecidable. As positive results,
we have also shown that the termination probability of order-\(n\) \pHORS{} can be
characterized by order-(\(n-1\)) fixpoint equations, which immediately yields a method
for computing a precise lower-bound of the termination probability. Based on the
fixpoint characterization, we have proposed a sound procedure for computing
an upper-bound of the termination probability,
which worked well on preliminary experiments.

It is left for future work to settle the question of whether it is possible to
compute the termination probability with arbitrary precision, which seems to be
a  difficult problem.\footnote{As described in the footnote
of Remark~\ref{rem:approximation},
this has recently been settled by ChatGPT; see also Appendix~\ref{sec:chatgpt}}
Another direction of future work is to develop a
(sound but incomplete) model checking procedure for \pHORS{}, using
the procedure for computing the termination probability as a backend.
 
\subsection*{Acknowledgments}
We would like to thank Kazuyuki Asada and Takeshi Tsukada for discussions on the topic,
and anonymous referees for useful comments. This work was supported by JSPS \textsc{KAKENHI} Grant Number JP15H05706, JP20H00577, and
JP20H05703, and by ANR \textsc{PPS} Grant Number 19-CE48-0014, and by ERC
CoG \textsc{DIAPASoN} Grant Agreement 818616.

We would also like to ChatGPT (hence also OpenAI) for solving the problem left
open in our original publication in LMCS. We also thank Hiroyuki Katsura for
asking ChatGPT to solve the problem.

%\bibliographystyle{alpha}
%\bibliography{full,koba}
\newcommand{\etalchar}[1]{$^{#1}$}

\newpage
\appendix
\section*{Appendix}
\section{Relationship between \pHORS{} and Recursive Markov Chains}
In this section, we provide mutual translations between order-1 \pHORSs{} and
recursive Markov chains.
\subsection{Encoding Recursive Markov Chains into Order-1 \pHORSs{}}
\label{sec:encoding-rmc}
In this section, we will give a sketch of a proof that any recursive
Markov chain (RMC in the following) can be faithfully encoded as an
order-1 \pHORS{}. In doing that, we will closely follow the notational
conventions and definitions from \cite{Etessami09}, Section 2.

Let us first of all fix an RMC $A=(A_1,\ldots,A_k)$, where each
component graph is
$A_i=(N_i,B_i,Y_i,\mathit{En}_i,\mathit{Ex}_i,\delta_i)$. We fix a
reachability problem, given in the form of a triple $(i_I,s_I,q_I)$
where $i_I\in\{1,\ldots,k\}$, $s_I$ is a vertex of $A_{i_I}$, and
$q_I\in\mathit{Ex}_{i_I}$, where a vertex of each $A_i$ is defined as an
element of
$$
N_i\cup\bigcup_{b\in B_i}\mathit{Call}_b\cup\bigcup_{b\in B_i}\mathit{Return}_b.
$$
Here, \(\mathit{Call}_b = \set{(b,\mathit{en})\mid \mathit{en}\in \mathit{En}_{Y_i(b)}}\) and \(\mathit{Return}_b = \set{(b,\mathit{ex})\mid \mathit{ex}\in \mathit{Ex}_{Y_i(b)}}\)
The reachability problem \((i_I,s_I,q_I)\) specifies \(\langle \epsilon,s_I\rangle\) as
the initial state, where \(s_I\) is a vertex of the component graph \(i_I\),
and \(\langle \epsilon,q_I\rangle\) as the reachability target (cf.
Section 2.2 of \cite{Etessami09}).
The \pHORS{} $\GRAM_A=(\NONTERMS_A,\RULES_A,S_A)$ is defined as follows:
\newcommand\trprob[3]{#1_{#2,#3}}
\begin{varitemize}
\item
  Nonterminals are defined as symbols of the form
  $F_{i,s}$ where $i\in\{1,\ldots,k\}$, and 
  $s$ is a vertex $A_{i}$.
  The type $\NONTERMS_A(F_{i,s})$ is $\T^{|\mathit{Ex}_i|}\to\T$.
  There is also a nonterminal $S_A$ of type \(\T\), which is taken to be
  $(i_I,s_I,q_I)$. The start symbol is \(S_A\).
\item
  Rules in $\RULES_A$ are of four kinds:
  \begin{varitemize}
  \item
    There is a rule
    $$
    (i_I,s_I,q_I)=S_A=F_{i_I,s_I}(\underbrace{\Omega,\ldots,\Omega}_{\mbox{$j-1$ times}},\Te,\Omega,\ldots,\Omega)
    $$
    where $\mathit{Ex_i}=\{s_1,\ldots,s_{|\mathit{Ex_i}|}\}$ and \(q_I=s_j\).
  \item
    For every $i$ and for every exit node
    $s_j\in\mathit{Ex_i}=\{s_1,\ldots,s_{|\mathit{Ex_i}}|\}$, there is a
    rule
    $$
    F_{i,s_j}(x_1,\ldots,x_{|\mathit{Ex}_i|})=x_j
    $$
  \item
    For each $i$ and for each non-exit node
    or return port $s$ of $A_i$, there is a rule
    $$
    F_{i,s}(x_1,\ldots,x)=\bigoplus_{j} \trprob{p}{s}{q} F_{i,q}(x_1,\ldots,x)
    $$
    where $\trprob{p}{s}{q}$ is the probability to go from $s$ to $q$,
    as given by the transition function $\delta_i$.
    \nk{I have changed the notation from \(p_s^q\) to \(\trprob{p}{s}{q}\),
    since the former can be confused with \((p_s)^q\). Also,
    the latter is consistent with the notation used in \cite{Etessami09}.}
  \item
    For every $i$ and for every call port $s=(b,\mathit{en})$ of $A_i$ which is
    in $\mathit{Call}_b$, there is a rule
    $$
    F_{i,s}(\vec{x})=F_{Y_i(b),\mathit{en}}(F_{i,(b,{\mathit{ex}_1})}(\vec{x}),\ldots,F_{i,(b,{\mathit{ex}_v})}(\vec{x}))
    $$
    and $\mathit{Ex}_{Y_i(b)}=\{\mathit{ex}_1,\ldots,\mathit{ex}_v\}$.
  \end{varitemize}
\end{varitemize}
The next step is to put any global state in $M_A$ in correspondence to
a term of $\GRAM_A$. This is actually quite easy, once one realizes
that:
\begin{varitemize}
\item
  $\GRAM_A$ is designed so that every term to which $S_A$ reduces
   can be seen as a complete ordered tree.
\item
  The rules in $\RULES_A$ have been designed so as to closely mimick
  the four inductive clauses by which the transition relation
  \(\Delta\)
    of the Markov chain $M_A$ is defined. In particular, any such pair $\langle\beta,u\rangle$ is such that
  the length of $\beta$ corresponds to the height of the corresponding
  term to which $S_A$ reduces. The only caveat is that the first
  such inductive clause needs to be restricted, because in \pHORS{},
  contrarily to Markov chains, one needs to fix \emph{one}
  initial state.
\item
  $(\langle \beta,u\rangle,p,\langle\beta',u'\rangle)\in\Delta$ if and only if the term corresponding to
   $\langle \beta,u\rangle$
  rewrites to the term corresponding to $\langle\beta',u'\rangle$ with probability $p$ 
  in one step.
\end{varitemize}
As a consequence, one easily derive that $\Prob(\GRAM_A)$ is precisely
the probability, in $M_A$, to reach $\langle{\epsilon,q_I}\rangle$ starting from
$\langle{\epsilon,s_I}\rangle$. 

\subsection{Encoding Order-1 \pHORSs{} into Recursive Markov Chains}
\label{sec:encoding-phors}

In this section, we show 
that any order-1 \pHORS{} can be encoded into
 a recursive Markov chain that has the same termination probability.

First, we can normalize any order-1 \pHORS{} to the one consisting of
the rewriting rules of the form:
\[
\begin{array}{l}
  S = F_1\,\Te\,\cdots\,\Te\\
  F_1\,x_1\,\cdots\,x_k = t_{1,L}\C{p_1}t_{1,R}\\
  \cdots\\
  F_m\,x_1\,\cdots\,x_k = t_{m,L}\C{p_1}t_{m,R},
\end{array}
\]
where each \(t_{i,d}\) (\(i\in \set{1,\ldots,m}, d\in \set{L,R}\)) is a
variable \(x_j\ (j\in\set{1,\ldots,k})\), or is of the form:
\[F_i\,(F_{j_1}\,x_1\,\cdots\,x_k)\,\cdots\,(F_{j_k}\,x_1\,\cdots\,x_k).\]
Note that \(\Omega\) on the righthand side can be replaced by
\(F\,x_1\,\cdots\,x_k\) where \(F\) is defined by
\[ F\,x_1,\cdots\,x_k = F(F\,x_1,\cdots\,x_k)\cdots(F\,x_1,\cdots\,x_k).\]

\newcommand\En{\mathit{En}}
\newcommand\Ex{\mathit{Ex}}
Given the normalized order-1 HORS above,
let \(M\) be a recursive Markov chain consisting of a single component
\(A_1=(N_1,B_1,Y_1,\En_1,\Ex_1,\delta_1)\) where:
\begin{itemize}
\item \(B_1\) is the set of terms of the form
  \(  F_i\,(F_{j_1}\,x_1\,\cdots\,x_k)\,\cdots\,(F_{j_k}\,x_1\,\cdots\,x_k)\) on
  the righthand side.
\item \(Y_1(b)=1\) for every \(b\in B_1\).
\item \(\En_1 = \set{F_1,\ldots,F_m}\).
\item \(\Ex_1 = \set{x_1,\ldots,x_k}\).
\item \(N_1 = \En_1\cup \Ex_1\).
\item \(\delta_1\) is the least set of the transitions that satisfies:
  \begin{itemize}
  \item \((F_i, p, x_j)\in \delta_1\)
    for each transition rule \(F_i\,x_1,\cdots\,x_k\redp{d,p}{} x_j\)\\
    (recall that we write \(F\,x_1,\cdots\,x_k\redp{L,p}{} t_L\)
    and \(F\,x_1,\cdots\,x_k\redp{R,1-p}{} t_R\)
    if there is a rule \(F\,x_1,\cdots\,x_k=t_L\C{p}t_R\)).
  \item \((F_i, p, (t, F_j))\in \delta_1\)\\
    if \(F_i\,x_1,\cdots\,x_k\redp{d,p}{} t\)
    and \(t\) is of the form
    \(  F_j\,(F_{j_1}\,x_1\,\cdots\,x_k)\,\cdots\,(F_{j_k}\,x_1\,\cdots\,x_k)\).
  \item \(((t,x_i), p, x_\ell)\in\delta_1\)\\ if
    \(t=F_j\,(F_{j_1}\,x_1\,\cdots\,x_k)\,\cdots\,(F_{j_k}\,x_1\,\cdots\,x_k)\),
    and \(F_{j_i}\,x_1\,\cdots\,x_k \redp{d,p}{} x_\ell\).
  \item \(((t,x_i), p, (t',F_{j'}))\in\delta_1\)\\ if
    \(t=F_j\,(F_{j_1}\,x_1\,\cdots\,x_k)\,\cdots\,(F_{j_k}\,x_1\,\cdots\,x_k)\)
    and \(F_{j_i}\,x_1\,\cdots\,x_k \redp{d,p}{} t'\),
     where 
    \(t'=F_{j'}\,(F_{j'_1}\,x_1\,\cdots\,x_k)\,\cdots\,(F_{j'_k}\,x_1\,\cdots\,x_k)\),
  \end{itemize}
\end{itemize}

Intuitively,
a \pHORS{} term of the form
\(  F_j\,(F_{j_1}\,x_1\,\cdots\,x_k)\,\cdots\,(F_{j_k}\,x_1\,\cdots\,x_k)\)
 is modeled as a call of \(F_j\), where
\(F_{j_i}\) is executed when the call exits from the exit port \(x_i\).
That is why, in the third and fourth kinds of transition rules above,
the next node is determined by the rule for \(F_{j_i}\).
From this intuition, it should be trivial that the termination probabilities
of the RMC and the original \pHORS{} coincide.

 \section{Proofs for Section~\ref{sec:fixpoint}}
\subsection{Proofs for Section~\ref{sec:order-n-equation}}
\label{sec:proof-tr-n-1}

\iftwocol
\subsubsection{Well-Typeness of the Equation}
\else
\subsubsection{Proof of Lemma~\ref{lem:tr-n-wf}}
\fi
We first prove the following lemma:
\begin{lemma}
  \label{lem:tr-n-wft}
If \(\STE \p t:\sty\),
    then \(\STE^\# \p t^\#:\sty^\#\).
\end{lemma}
\begin{proof}
  This follows by straightforward induction on the derivation of
  \(\STE \p t:\sty\).
\end{proof}

\iftwocol
The following lemma states that the output of the translation is
well-typed. 
\begin{lemma}
  \label{lem:tr-n-wf}
  Let \(\GRAM=(\NONTERMS,\RULES,S)\) be an order-\(n\) \pHORS.
  Then \(\NONTERMS^\# \p \E_{\GRAM}\) and \(\NONTERMS^\#\p S:\realt\).
\end{lemma}
\begin{proof}
\else
\begin{proofn}{Proof of Lemma~\ref{lem:tr-n-wf}}
\fi
  \(\NONTERMS^\#\p S:\realt\) follows immediately from
  \(\NONTERMS(S)=\T\), Lemma~\ref{lem:tr-n-wft}, and \(\T^\#=\realt\).
  Let \(\RULES\) be
  \(\set{F_i\,x_1\,\cdots,x_{\ell_i}=t_i\mid i\in \set{1,\ldots,m}}\).
  Then, by the definition of \pHORS{}, we have
  \(\NONTERMS, x_1\COL\sty_{i,1},\ldots,x_{\ell_i}\COL\sty_{i,\ell_i}
  \p t_i:\T\),
  with \(\NONTERMS(F)=\sty_{i,1}\to\cdots\to\sty_{i,\ell_i}\to\T\).
  We need to show that
  \[ \NONTERMS^\#, x_1\COL\sty_{i,1}^\#,\ldots,x_i\COL\sty_{i,\ell_i}^\#
  \p t_i^\#:\realt\]
  for each \(i\), but this follows immediately from the typing of \(t_i\) above
  and Lemma~\ref{lem:tr-n-wft}. 
\iftwocol
\end{proof}
\else
\end{proofn}
\fi
\subsubsection{Proof of Theorem~\ref{prop:order-k-fixpoint}}
We call  a \pHORS{} \(\GRAM\) \emph{recursion-free} if there is no cyclic dependency
on its non-terminals. More precisely,
given a \pHORS{}
\(\GRAM\), we define the relation \(\succ_{\GRAM}\) on its non-terminals by:
 \(F_i\succ_{\GRAM} F_j\) iff \(F_j\) occurs on the righthand side of the rule for \(F_i\).
A \pHORS{} \(\GRAM\) is defined to be recursion-free if the transitive closure of
\(\succ_{\GRAM}\) is irreflexive.

Below we write \(\qv{t}{\rho}\) for \(\seme{t^\#}{\rho}\).
\begin{lemma}
\label{lem:order-k-fixpoint}
Let \(\GRAM=(\NONTERMS,\RULES,S)\) be a recursion-free \pHORS{}, and
 \(\rho\) be the least solution of \(\E_{\GRAM}\).
If \(\NONTERMS\p t:\T\), then
\(\Prob(\GRAM,t) = \qv{t}{\rho}\).
\end{lemma}
\begin{proof}
Since \(\GRAM\) is recursion-free, it follows from the strong normalization of
the simply-typed \(\lambda\)-calculus that \(t\) does not have any infinite
reduction sequence.
Because the reduction relation is finitely branching, by K\"{o}nig's lemma,
there are only finitely many reduction sequences from \(t\);
thus a longest reduction sequence from \(t\) exists, and we write
  \(\sharp(t)\) for its length.
The proof proceeds  by induction on \(\sharp(t)\).
 If \(\sharp(t)=0\), then \(t\) is either \(\Te\) (in which case,
both sides of the equation are \(1\)) or \(\Omega\)
(in which case,
both sides of the equation are \(0\)); thus, the result follows immediately.
Otherwise, \(t\) must be of the form \(F\;s_1\cdots s_k\)
where \(F\,x_1\cdots x_k = t_1\C{p}t_2\).
Then
\[
\begin{array}{l}
\Prob(\GRAM, t) =
p\Prob(\GRAM, [s_1/x_1,\ldots,s_k/x_k]t_1)\\\qquad\qquad +
(1-p)\Prob(\GRAM, [s_1/x_1,\ldots,s_k/x_k]t_2).
\end{array}\]
Since \(\sharp(t)>\sharp([s_1/x_1,\ldots,s_k/x_k]t_i)\),
by the induction hypothesis, the righthand side equals:
\[
\begin{array}{l}
  p ([s_1/x_1,\ldots,s_k/x_k]t_1)^{\#}_{\rho}
  +(1-p)( [s_1/x_1,\ldots,s_k/x_k]t_2)^{\#}_{\rho}\\
  =
  p (t_1)^{\#}_{\rho\set{\seq{x}\mapsto
    \seq{s}^{\#}_{\rho}}}
  +(1-p)(t_2)^{\#}_{\rho\set{\seq{x}\mapsto
        \seq{s}^{\#}_{\rho}}}\\
  = (F\;s_1\cdots s_k)^\#_{\rho}
  = t^\#_{\rho},
  \end{array}
  \]
  as required.
\end{proof}

For a \pHORS{} \(\GRAM=(\NONTERMS,\RULES,S)\)
with \(\dom(\NONTERMS)=\set{F_1,\ldots,F_m}\), we define
its \emph{\(k\)-th approximation} \(\GRAM^{(k)}=(\NONTERMS^{(k)},\RULES^{(k)},S^{(k)})\) by:
\[
\begin{array}{l}
\NONTERMS^{(k)} = \set{F_j^{(i)}\mapsto \NONTERMS(F_j) \mid j\in\set{1,\ldots,m}, 0\leq i\leq k}\\
\RULES^{(k)}(F_j^{(i)}) = [F_1^{(i-1)}/F_1,\ldots,F_m^{(i-1)}/F_m]\RULES(F_j)\\\hfill
\mbox{ for each $i\in\set{1,\ldots,k}$}\\
\RULES^{(k)}(F_j^{(0)}) = \lambda \seq{x}.\Omega\C{1}\Omega.
\end{array}
\]
The following properties follow immediately from the construction of \(\GRAM^{(k)}\).
(Recall that \(\F_{\E}\) denotes the function associated with the fixpoint equations
\(\E\), as defined in Section~\ref{sec:ho-fixpoint}.)
\begin{lemma}
\label{lem:approximation}
\begin{enumerate}
\item \(\GRAM^{(k)}\) is recursion-free.
\item
\(\Prob(\GRAM) = \bigsqcup_{k\in\omega} \Prob(\GRAM^{(k)})\).
\item
  \(\F_{\E_{\GRAM}}^k(\bot_{\sem{\NONTERMS}})(F) = \F_{\E_{\GRAM^{(k)}}}^k(\bot_{\sem{\NONTERMS^{(k)}}})(F^{(k)}) =
  \LFP(\F_{\E_{\GRAM^{(k)}}})(F^{(k)})\) for each non-terminal \(F\) of \(\GRAM\).
\end{enumerate}
\end{lemma}
\begin{proof}
  \begin{enumerate}
  \item This follows immediately from the fact that
  \(F_\ell^{(i)}\succ_{\GRAM^{(k)}} F_{\ell'}^{(j)}\) only if \(j=i-1\).
\item Let \(P\) be the set \(\set{(\pi,p)\mid S\redsp{\pi,p}{\GRAM}\Te}\)
  and \(P^{(k)}\) be \(\set{(\pi,p)\mid S^{(k)}\redsp{\pi,p}{\GRAM^{(k)}}\Te}\).
  Then \(\Prob(\GRAM)=\sum_{(\pi,p)\in P} p\) and
  \(\Prob(\GRAM^{(k)})=\sum_{(\pi,p)\in P^{(k)}} p\).
  Note that
  for any reduction \(s\redp{d,p}{\GRAM^{(k)}} t\) with \(t\neq \Omega\),
  there exists a corresponding reduction \(s^!\redp{d,p}{\GRAM} t^!\) where
  \(s^!\) and \(t^!\) are the terms of \(\GRAM\) obtained from \(s\) and \(t\)
  respectively,  by removing indices, i.e. by replacing each \(F^{(i)}\) with \(F\).
  Thus, \(P^{(k)}\subseteq P\) for any \(k\).
  Conversely, if \(S\redp{\pi,p}{\GRAM}\Te\), then \(S^{(|\pi|)}\redp{\pi,p}{\GRAM^{(|\pi|)}}\Te\), because non-terminals are unfolded at most
  \(|\pi|\) times in \(S\redp{\pi,p}{\GRAM}\Te\). 
  Therefore, \(P = \bigcup_k P^{(k)}\), from which the result follows.
\item We show that
  \(\F_{\E_{\GRAM}}^k(\bot_{\sem{\NONTERMS}})(F)= \F_{\E_{\GRAM^{(\ell)}}}^k(\bot_{\sem{\NONTERMS^{(\ell)}}})(F^{(k)})\) holds for any \(\ell\geq k\), by induction on \(k\),
  from which the first equality follows.
  The base case \(k=0\) is trivial. For \(k>0\),
By the definition of the rule for \(F^{(k)}\) and the induction hypothesis, we have:
\[
\begin{array}{l}
\F_{\E_{\GRAM^{(\ell)}}}^k(\bot_{\sem{\NONTERMS^{(\ell)}}})(F^{(k)})\\
  = \seme{([F_1^{(k-1)}/F_1,\ldots,F_m^{(k-1)}/F_m]\RULES(F))^\#}
  {\F_{\E_{\GRAM^{(\ell)}}}^{k-1}(\bot_{\sem{\NONTERMS^{(\ell)}}})}\\
  =\seme{\RULES(F)^\#}
  {
\set{F_i\mapsto \F_{\E_{\GRAM^{(\ell)}}}^{k-1}(\bot_{\sem{\NONTERMS^{(\ell)}}})(F_i^{(k-1)})\mid i\in\set{1,\ldots,m}}
}
\\
  = \seme{\RULES(F)^\#}{\F_{\E_{\GRAM}}^{k-1}(\bot_{\sem{\NONTERMS}})} \mbox{ (by the induction hypothesis)}\\
  = \F_{\E_{\GRAM}}^{k}(\bot_{\sem{\NONTERMS}})(F),
  \end{array}
    \]
    as required. (Here, we have extended \((\cdot)^\#\) and \(\seme{t}{\rho}\)
    for \(\lambda\)-terms in the obvious manner.)

    For the second equality, we can show that 
\(\F_{\E_{\GRAM^{(\ell)}}}^k(\bot_{\sem{\NONTERMS^{(\ell)}}})(F^{(k)}) =
\LFP(\F_{\E_{\GRAM^{(\ell)}}})(F^{(k)})\) holds for any \(\ell\geq k\), by
straightforward induction on \(k\).
  \end{enumerate}
\end{proof}

Theorem~\ref{prop:order-k-fixpoint} follows as a corollary of
the above lemmas.
\ifacm
\begin{proof}[Proof of Theorem~\ref{prop:order-k-fixpoint}]
\else
\begin{proofn}{Proof of Theorem~\ref{prop:order-k-fixpoint}}
\fi
By Lemmas~\ref{lem:order-k-fixpoint} and \ref{lem:approximation},
we have
\[
\begin{array}{l}
\Prob(\GRAM) = \bigsqcup_k \Prob(\GRAM^{(k)}) =
\bigsqcup_k \LFP(\F_{\E_{\GRAM^{(k)}}})(S^{(k)}) \\\qquad = 
\bigsqcup_k \F_{\E_\GRAM}^k(\bot)(S) = \LFP(\F_{\E_\GRAM})(S)
\end{array}
\]
as required.
\ifacm
\end{proof}
\else
\end{proofn}
\fi

\subsection{Proofs for Section~\ref{sec:order-n-1-equation}}
\label{sec:proof-sec42}

\subsubsection{Proof of Lemma~\ref{lem:tr-well-typedness}}

We define the translation for a type environment on variables (other than non-terminals;
note that the translation is different from the one for \(\NONTERMS\))
by:
\[
\begin{array}{l}
(y_1\COL\sty_1,\ldots,y_k\COL\sty_k)^\dagger
=
(y_{1,0},\ldots,y_{1,\arity(\sty_1)+1})\COL \sty_1^\dagger,
\ldots, 
(y_{k,0},\ldots,y_{k,\arity(\sty_k)+1})\COL \sty_k^\dagger.
\end{array}
\]

\begin{lemma}
  \label{lem:tr-term}
  If \(\NTE\cup\STE, \seq{x}\COL\seq{\T}\p t:\sty\) and
  \(\STE; \seq{x}\pN t:\sty \tr e\), 
  then \(\NONTERMS^\dagger\cup\STE^\dagger \p e: {\sty}^{\dagger+|\seq{x}|}\).
\end{lemma}
\begin{proof}
  This follows by straightforward induction on
  the derivation of \(\NTE\cup\STE, \seq{x}\COL\seq{\T}\p t:\sty\).
\end{proof}

We also prepare the following lemma on the syntactic property
of the translation result, which
is important for Lemma~\ref{lem:tr-well-typedness} and
the substitution lemma (Lemma~\ref{lem:n-1:subj} below) proved later.
\begin{lemma}
  \label{lem:n-1:independence}
Suppose:
\[\STE; \seq{z}\pN t:\sty\tr (t_0,\ldots,t_{\arity(\sty)},t_{\arity(\sty)+1},\ldots,t_{\arity(\sty)+|\seq{z}|+1}).\]
Then,
for each \(y_i\in \dom(\STE)\), \(y_{i,0}\) 
does not occur in \(t_1,\ldots,t_{\arity(\sty)+|\seq{z}|+1}\).
\end{lemma}
\begin{proof}
  This follows by straightforward induction on the structure of \(t\).
  \end{proof}

\ifacm
\begin{proof}[Proof of Lemma~\ref{lem:tr-well-typedness}]
\else
\begin{proofn}{Proof of Lemma~\ref{lem:tr-well-typedness}}
\fi
  \(\trT{\NONTERMS}(S_1)=\realt\) follows immediately
  from \(\NONTERMS(S)=\T\to\T\) and the definition of \(\trT{\NONTERMS}\).
  To prove \(\trT{\NONTERMS} \p \E\),
  let \(F\,y_1\,\cdots\,y_\ell\,x_1\,\cdots\,x_k=t_L\C{p}t_R\in \RULES\),
  with \(\NONTERMS(F)=\sty_1\to\cdots\to\sty_\ell\To\T^k\to\T\).
  Suppose also 
  \[y_1\COL\sty_1,\ldots,y_\ell\COL\sty_\ell; x_1,\ldots,x_k
  \pN t_d:\T \tr (t_{d,0},\ldots,t_{d,k+1})\]
  for \(d\in\set{L,R}\).
  We need to prove 
  \[ \NONTERMS^\dagger,
  (\seq{y_1})\COL\trT{\sty_1},\ldots,(\seq{y_\ell})\COL\trT{\sty_\ell}
  \p t_{d,0}:\realt\]
  and 
  \[ \NONTERMS^\dagger,
  (\seq{y_1}')\COL\trTp{\sty_1},\ldots,(\seq{y_\ell}')\COL\trTp{\sty_\ell}
  \p t_{d,i}:\realt\]
  for \(i\in\set{1,\ldots,k}\), where \(\seq{y_i}\) and \(\seq{y_i}'\) are
  as given in the premises of the rule \rn{Tr-Rule}.
  The former follows immediately from Lemma~\ref{lem:tr-term}.
  For the latter, by Lemma~\ref{lem:tr-term}, we have
  \[ \NONTERMS^\dagger,
  (\seq{y_1})\COL\trT{\sty_1},\ldots,(\seq{y_\ell})\COL\trT{\sty_\ell}
  \p t_{d,i}:\realt.\]
 By Lemma~\ref{lem:n-1:independence},
 \(y_{i,0}\) does not occur in \(t_{d,i}\). Thus we can remove
  the type bindings on them and obtain
  \[ \NONTERMS^\dagger,
  (\seq{y_1}')\COL\trTp{\sty_1},\ldots,(\seq{y_\ell}')\COL\trTp{\sty_\ell}
  \p t_{d,i}:\realt\]
   as required.
\ifacm
  \end{proof}
\else
  \end{proofn}
\fi

\subsubsection{Proof of Theorem~\ref{prop:order-k-1-fixpoint}}

Given two expressions \(e_1,e_2\) and fixpoint equations \(\EQref{\GRAM}\),
we write \(e_1\cong_{\EQref{\GRAM}}e_2\)
  if \(\seme{e_1}{\rho_{\EQref{\GRAM}}} = \seme{e_e}{\rho_{\EQref{\GRAM}}}\).
  We often omit the subscript.

As sketched in Section~\ref{sec:order-n-1-equation}, we first prove
the theorem for recursion-free \pHORSs{}. A key property used for showing it
is that the translation relation is preserved by reductions, roughly in the sense
that if  \(t\redp{L,p}{} t_L\) and \(t\redp{R,1-p}{}t_R\), then
\(t\tr e\) (i.e., \(t\) is translated to \(e\))
implies that there exist \(e_L\) and \(e_R\) such that
\(t_L\tr e_L\), \(t_R\tr e_R\) and \(e\cong p \cdot e_L+(1-p)\cdot e_R\)
(where \(+\) and \(\cdot\) are pointwise extended to operations on tuples).
Thus, the property that \(e\) represents the termination probability of \(t\)
follows from the corresponding properties of \(e_L\) and \(e_R\); by
induction (note that since we are considering recursion-free \pHORS{},
\(\flat(t)>\flat(t_L), \flat(t_R)\), where \(\flat(t)\) denotes the length of the
longest reduction sequence from \(t\)),
 it follows that if the initial term is translated to \(e_0\), then \(e_0\)
 represents the termination probability of the initial term.

 Unfortunately, however, the translation relation is \emph{not} necessarily
 preserved by the standard reduction relation \(\redp{d,p}{\GRAM}\) defined
 in Section~\ref{sec:problem}. We thus introduce another reduction relation
 that uses \emph{explicit substitutions} on order-0 variables.
To this end, we extend the syntax of terms as follows.
\[
\begin{array}{l}
  t \mbox{ (extended terms)}::=
    \Omega\mid x \mid t_1t_2 \mid \Subs{t_1/x_1,\ldots,t_k/x_k}{t_0}
\end{array}
\]
Here, 
\(\Subs{t_1/x_1,\ldots,t_k/x_k}{t_0}\) represents
an \emph{explicit} substitution;
the intended meaning is the same as the ordinary substitution
\([t_1/x_1,\ldots,t_k/x_k]{t_0}\) (which represents the term
obtained from \(t_0\) by simultaneously substituting 
\(t_i\) for \(x_i\)), but the substitution is delayed until one of the variables
in \(x_1,\ldots,x_k\) becomes necessary.
We often abbreviate \(\Subs{t_1/x_1,\ldots,t_k/x_k}\) as \(\Sub{\seq{x}}{\seq{t}}\).
Note that we have omitted \(\Te\); we consider
an open term \(S\,x\) as the initial term instead of \(S\,\Te\).
The type judgment relation for terms is extended by adding the following rule:
\infrule{\STE \p s_i:\T \mbox{ (for each $i\in\set{1,\ldots,k}$)}\andalso \STE,x_1\COL\T,\ldots,x_k\COL\T\p t:\T   }
        {\STE \p \Subs{s_1/x_1,\ldots,s_k/x_k}{t}:\T}
Thus,  explicit substitutions are allowed only for order-0 variables.
        
\noindent        
\textbf{Reductions with explicit substitutions:}\ \\
We now define a reduction relation for extended terms.
The set of evaluation contexts, ranged over by \(E\), is defined by:
\[ E ::= \Hole \mid \Sub{\seq{x}}{\seq{t}}E.\]
The new reduction relation \(t\newredpg{d,p}{\GRAM}t'\) (where \(d\in\set{L,R,\epsilon}\))
is defined as follows.
\infrule{z\notin \set{x_1,\ldots,x_k}}
   {E[\Subs{t_1/x_1,\ldots,t_k/x_k}z]\newredpg{\epsilon,1}{\GRAM} E[z]}
\infrule{}
   {E[\Subs{t_1/x_1,\ldots,t_k/x_k}x_i]\newredpg{\epsilon,1}{\GRAM} E[t_i]}
   \infrule{\RULES(F)=\lambda \seq{y}.\lambda \seq{z}.u_L\C{p}u_R\\
          \NONTERMS(F)=\seq{\sty}\To\T^\ell\to\T\andalso \ell=\seql{z}=\seql{t}
     \andalso \seql{y}=\seql{s}\\
     \mbox{$\seq{z}$ do not occur in $E[F\,\seq{s}\,\seq{t}]$}}
{E[F\,\seq{s}\,\seq{t}] \newredpg{L,p}{\GRAM} E[\Sub{\seq{z}}{\seq{t}}[\seq{s}/\seq{y}]u_L]}
\infrule{\RULES(F)=\lambda \seq{y}.\lambda \seq{z}.u_L\C{p}u_R\\
          \NONTERMS(F)=\seq{\sty}\To\T^\ell\to\T\andalso \ell=\seql{z}=\seql{t}
     \andalso \seql{y}=\seql{s}\\
     \mbox{$\seq{z}$ do not occur in $E[F\,\seq{s}\,\seq{t}]$}}
{E[F\,\seq{s}\,\seq{t}] \newredpg{R,1-p}{\GRAM} E[\Sub{\seq{z}}{\seq{t}}[\seq{s}/\seq{y}]u_R]}

        We call reductions using the first two rules (i.e., reductions labeled by
        \(\newredpg{\epsilon,p}{\GRAM}\)) \emph{administrative reductions}.
        In the last two rules, we assume that \(\alpha\)-conversion is implicitly
        applied so that \(\seq{z}\) do not clash with variables that are
        already used in 
        \(E[F\,\seq{s}\,\seq{t}]\).
        In those rules,
        recall also our notational convention that
        when we write \(F\,\seq{s}\,\seq{t}\), the second sequence \(\seq{t}\)
        is the maximal sequence of order-0 terms (that condition is made explicit
        in the above rules, but below we often omit to state it).
        As before, we often omit the subscript \(\GRAM\).

        For an extended term \(t\), we write \(t^*\) for the term obtained
        by replacing explicit substitutions with ordinary substitutions.
        For example, \((\Sub{x}{t}(F\,x))^* = F\,t\).
        The following lemma states that the new reduction relation is
        essentially equivalent to the original reduction relation:
        \begin{lemma}
          \label{lem:red-vs-esred}
          Let \(s\) be an extended term. 
          \begin{enumerate}
          \item If \(s\newredp{\epsilon,1}{}t\), then \(s^*=t^*\).
          \item If \(s\newredp{d,p}{}t\) with \(d\in\set{L,R}\),
            then \(s^*\redp{d,p}{}t^*\).
          \item If \(s^*\redp{d,p}{}u\), then
            there exists \(t\) such that \(s(\newredp{\epsilon,1})^*\newredp{d,p}t\) and \(t^*=u\).
            \end{enumerate}
        \end{lemma}
        \begin{proof}
          Immediate from the definitions of \(\redp{d,p}{}\) and \(\newredp{d,p}{}\).
          \end{proof}
\newcommand\Rpath{\mathit{P}_{\mathtt{es}}}        
 We define \(\newredspg{\pi,p}{\GRAM}\)
 in an analogous manner to \(\redsp{\pi,p}{\GRAM}\), where
 the label \(\epsilon\) is treated as an empty word.
 For an extended term \(t\) that may contain an order-0 free variable \(x\),
 we write \(\Rpath(\GRAM,t,x)\) for the set
 \(\set{(\pi,p)\mid t\newredspg{\pi,p}{\GRAM}x}\), and
write \(\ProbES(\GRAM,t,x)\) for
 \(\sum_{(\pi,p)\in \Rpath(\GRAM,t,x)}p\),
 based on the new reduction relation.
 The following lemma follows immediately from the above definitions
 and Lemma~\ref{lem:red-vs-esred}.
 \begin{lemma}
   \label{lem:red-vs-redes}
   Let \(t\) be a term of \pHORS{} \(\GRAM=(\NONTERMS,\RULES,S)\)
   such that \(\NONTERMS,x\COL\T\p t:\T\) and \(t\) does not contain
   \(\Te\). Then
  \( \Prob(\GRAM,[\Te/x]t)=\ProbES(\GRAM,t,x)\).
 \end{lemma}
 \begin{proof}
   By Lemma~\ref{lem:red-vs-esred}, \(t\newredsp{\pi,p}x\) if and only if
   \(t^*\redsp{\pi,p}{}x\), if and only if
   \([\Te/x]t^*\redsp{\pi,p}{}\Te\)  (for the second ``if and only if'',
   recall the assumption that \(\Te\) does not occur in \(\RULES\)),
   from which the result follows.
 \end{proof}
We extend the translation relation for terms with the following rule.
\infrule[Tr-Sub]{\STE;x_1,\ldots,x_k\pN s_i:\T\tr (s_{i,0},\ldots,s_{i,k+1}) \mbox{ (for each $i\in\set{1,\ldots,\ell}$)}\\
\STE;z_1,\ldots,z_\ell,x_1,\ldots,x_k\pN t:\T\tr (t_0,\ldots,t_{k+\ell+1})}
        {\STE;x_1,\ldots,x_k\pN \Subs{s_1/z_1,\ldots,s_\ell/z_\ell}{t}:\T\tr\\\quad 
(t_0+\Sigma_{i=1}^\ell t_{i}\cdot s_{i,0},
          t_{\ell+1}+\Sigma_{i=1}^\ell t_{i}\cdot s_{i,1}, \ldots,t_{k+\ell+1}+\Sigma_{i=1}^\ell t_{i}\cdot s_{i,k+1})}

We shall prove that the translation relation is preserved by the new reduction
relation (Lemmas~\ref{lem:tr-subj}, \ref{lem:tr-subj-ad1}, and
\ref{lem:tr-subj-ad2} below).

\begin{lemma}[Weakening]
  \label{lem:tr-weakening}
  \begin{enumerate}
    \item 
      If \(\STE;x_1,\ldots,x_k\pN t:\sty\tr e\), then
      \(\STE,y\COL\sty_y ;x_1,\ldots,x_k\pN t:\sty\tr e\).
    \item
      If \(\STE;x_1,\ldots,x_k\pN t:\sty\tr (t_0,\ldots,t_\ell)\), then
      \(\STE ;x_1,\ldots,x_k, x_{k+1}\pN t:\sty\tr (t_0,\ldots,t_\ell,t_\ell)\).
  \end{enumerate}
\end{lemma}
\begin{proof}
  This follows by straightforward induction on the structure of \(t\).
  \end{proof}

\begin{lemma}[Exchange]
  \label{lem:tr-exchange}\ \\
  If \(\STE;x_1,\ldots,x_{i},x_{i+1},\ldots,x_k\pN t:\sty\tr
  (t_0,\ldots, t_{\arity(\sty)+i},t_{\arity(\sty)+i+1},\ldots,t_{\arity(\sty)+k+1})\), then
      \(\STE;x_1,\ldots,x_{i+1},x_{i},\ldots,x_k\pN t:\sty\tr 
  (t_0,\ldots, t_{\arity(\sty)+i+1},t_{\arity(\sty)+i},\ldots,t_{\arity(\sty)+k+1})\).
\end{lemma}
\begin{proof}
  This follows by straightforward induction on the structure of \(t\).
  \end{proof}

As usual, the substitution lemma, stated below, is a critical lemma for
proving subject reduction. The statement of our substitution lemma
is, however, quite delicate, due to a special treatment of order-0 variables.

\begin{lemma}[Substitution]
  \label{lem:n-1:subj}
  Suppose \(t\) does not contain any explicit substitutions (i.e., any subterms
  of the form \(\Sub{\seq{x}}{\seq{u}}s\)).
If \(\seq{y}\COL \seq{\sty_y}; \seq{z}\pN t:\sty\tr (\seq{t},t_{m+1})\) and
\(\emptyTE; x_1,\ldots,x_k\pN s_i:\sty_{y,i}\tr (\seq{s_i},s_{i,\ell_i+1},\ldots,s_{i,\ell_i+k+1})\)
with \(\set{x_1,\ldots,x_k}\cap \set{\seq{z}}=\emptyset\) and
\(\ell_i=\arity(\sty_{y,i})\),
then: \[
\begin{array}{l}
\emptyTE;\seq{z},x_1,\ldots,x_k\pN [\seq{s}/\seq{y}]t:\sty\tr (\theta_0\seq{t},
\theta_{1}t_{0},\ldots,
\theta_{k}t_{0},\theta_0t_{m+1}),
\end{array}\]
where
 the substitutions
 \(\theta_j (j\in\set{0,\ldots,k})\) are defined by:
 \[
 \begin{array}{l}
   \theta_j = \theta_{1,j}\cdots\theta_{\seql{s},j}
  \mbox{ for $j\in\set{1,\ldots,k}$}\\
  \theta_{i,0}=[s_{i,0}/y_{i,0},\ldots,s_{i,\ell_i}/y_{i,\ell_i},\iftwocol \\\quad\qquad\quad\fi
s_{i,\ell_i+k+1}/y_{i,\ell_i+1}]
  \mbox{ for $i\in\set{1,\ldots,|\seq{s}|}$}\\
  \theta_{i,j} = [s_{i,\ell_i+j}/y_{i,0},s_{i,1}/y_{i,1},\ldots,s_{i,\ell_i}/y_{i,\ell_i},
\iftwocol \\\quad\qquad\quad\fi s_{i,\ell_i+k+1}/y_{i,\ell_i+1}]\\\hfill
  \mbox{ for $i\in\set{1,\ldots,|\seq{s}|}, j\in\set{1,\ldots,k}$}.
  \end{array}
\]
\end{lemma}
Here, the part
\(\theta_{1}t_{0},\ldots,
\theta_{k}t_{0}\) accounts for information about
reachability to the newly introduced variables \(x_1,\ldots,x_k\).
\begin{proof}
  Induction on the derivation of
  \(\seq{y}\COL\seq{\sty_y};\seq{z}\pN t:\sty\tr (\seq{t},t_{m+1})\).
\begin{itemize}
\item Case \rn{Tr-Omega}: In this case, \(t=\Omega\) with
  \(\seq{t}=\seq{0}\) and \(t_{m+1}=0\).
Thus, the result follows immediately from the rule \rn{Tr-Omega}.
\item Case \rn{Tr-GVar}: In this case,
  \(t=z_i\) and \(\sty=\T\), with \(\seq{t} = 0^i, 1, 0^{|z|-i}\) and \(t_{m+1}=0\).
  By using \rn{Tr-GVar}, we obtain
  \[\emptyTE;\seq{z},x_1,\ldots,x_k\pN [\seq{s}/\seq{y}]t (=z_i) :\sty\tr
  (0^{i}, 1, 0^{\seql{z}+k-i+1}).\]
  Since
\[  (\theta_0\seq{t},
\theta_{1}t_{0},\ldots,
\theta_{k}t_{0},\theta_0t_{m+1})= 
(\seq{t},
\underbrace{t_{0},\ldots,t_{0}}_k,t_{m+1})= 
(0^i, 1, 0^{\seql{z}+k-i+1}),\]
we have the required result.

\item Case \rn{Tr-Var}: 
In this case, \(t=y_i\), 
with \(\seq{t}=y_{i,0},\ldots,y_{i,\ell_i},y_{i,\ell_i+1}^{\seql{z}}\), and \(t_{m+1}=y_{i,\ell_i+1}\).
By applying Lemmas~\ref{lem:tr-weakening} and \ref{lem:tr-exchange} to 
\(\emptyTE; x_1,\ldots,x_k\pN s_i:\sty_{y,i}\tr (\seq{s_i},s_{i,\ell_i+1},\ldots,s_{i,\ell_i+k+1})\),
we obtain:
\[
\begin{array}{l}
\emptyTE; \seq{z}, x_1,\ldots,x_k\pN s_i:\sty_{y,i}\tr 
\iftwocol \\\quad\qquad\quad\fi
(\seq{s_i},s_{i,\ell_i+k+1}^{\seql{z}},s_{i,\ell_i+1},\ldots,s_{i,\ell_i+k+1}).
\end{array}
\]
Since \(\seq{t}=y_{i,0},\ldots,y_{i,\ell_i},y_{i,\ell_i+1}^{\seql{z}}\),  and
\(t_{m+1}=y_{i,\ell_i+1}\),
we have
\[
\begin{array}{l}
(\theta_0\seq{t},
\theta_{1}t_{0},\ldots,
\theta_{k}t_{0},\theta_0t_{m+1}) = \iftwocol \\\quad\qquad\quad\fi
(\seq{s_i},s_{i,\ell_i+k+1}^{\seql{z}},s_{i,\ell_i+1},\ldots,s_{i,\ell_i+k+1}).
\end{array}
\]
Thus we have the required result.
\item Case \rn{Tr-App}: In this case, we have \(t=uv\) and
\[
\begin{array}{l}
\seq{y}\COL\seq{\sty_y};\seq{z} \pN u:\sty_v\to \sty\tr (u_0,\seq{u},u_{\ell'+1},\ldots,u_{\ell'+|\seq{z}|+1})\\
\seq{y}\COL\seq{\sty_y};\seq{z} \pN v:\sty_v\tr (v_0,\seq{v},v_{\ell''+1},\ldots,v_{\ell''+|\seq{z}|+1})\\
\seq{t} =
u_0(v_0,\seq{v},v_{\ell''+|\seq{z}|+1}),
\seq{u}(\seq{v},v_{\ell''+|\seq{z}|+1}), \\\qquad\qquad
u_{\ell'+1}(v_{\ell''+1},\seq{v},v_{\ell''+|\seq{z}|+1}),\ldots,\iftwocol \\\quad\qquad\quad\fi
u_{\ell'+|\seq{z}|}(v_{\ell''+|\seq{z}|},\seq{v},v_{\ell''+|\seq{z}|+1})\\
t_{m+1} = u_{\ell'+|\seq{z}|+1}(v_{\ell''+|\seq{z}|+1},\seq{v},v_{\ell''+|\seq{z}|+1})\\
\ell' = \seql{u}=\arity(\sty) \quad m = \ell'+|\seq{z}| \quad
\ell'' = \arity(\sty_v)=\seql{v}.
\end{array}
\]
By the induction hypothesis, we have:
\[
\begin{array}{l}
\emptyTE;\seq{z},x_1,\ldots,x_k \pN [\seq{s}/\seq{y}]u:\sty_v\to \sty\tr \\\qquad
  (\theta_0(u_0,\seq{u},u_{\ell'+1},\ldots,u_{\ell'+|\seq{z}|}), \iftwocol \\\quad\qquad\quad\fi\theta_{1}u_0
\ldots, 
   \theta_{k}u_0, \theta_0u_{\ell'+\seql{z}+1})\\
\emptyTE;\seq{z},x_1,\ldots,x_k \pN [\seq{s}/\seq{y}]v:\sty_v\tr \\\qquad
(\theta_0(v_0,\seq{v},v_{\ell''+1},\ldots,v_{\ell''+|\seq{z}|}), \iftwocol \\\quad\qquad\quad\fi \theta_{1}v_0, \ldots,
 \theta_{k}v_{0}, \theta_0v_{\ell''+\seql{z}+1}
)\\
\end{array}
\]
By applying \rn{Tr-App}, we obtain:
\[
\begin{array}{l}
\emptyTE;\seq{z},x_1,\ldots,x_k \pN ([\seq{s}/\seq{y}]u)([\seq{s}/\seq{y}]v):\sty\tr \\\qquad
((\theta_0u_0)(\theta_0v_0,\theta_0\seq{v},\theta_{0}v_{\ell''+|\seq{z}|+1}),\iftwocol \\\qquad\fi
(\theta_0\seq{u})
(\theta_0\seq{v},\theta_{0}v_{\ell''+|\seq{z}|+1}),\\\qquad
(\theta_{0}u_{\ell'+1})(\theta_0v_{\ell''+1},\theta_0\seq{v},\theta_{0}v_{\ell''+|\seq{z}|+1}),\ldots, \iftwocol \\\qquad\fi (\theta_{0}u_{\ell'+\seql{z}})(\theta_0v_{\ell''+\seql{z}},\theta_0\seq{v},\theta_{0}v_{\ell''+|\seq{z}|+1}),\\\qquad
(\theta_{1}u_{0})(\theta_{1}v_0,\theta_0\seq{v},\theta_{0}v_{\ell''+|\seq{z}|+1}),\ldots, \iftwocol \\\qquad\fi (\theta_{k}u_{0})(\theta_{k}v_0,\theta_0\seq{v},\theta_{0}v_{\ell''+|\seq{z}|+1})\\\qquad
(\theta_{0}u_{\ell'+\seql{z}+1})(\theta_{0}v_{\ell'+\seql{z}+1},\theta_0\seq{v},\theta_{0}v_{\ell''+|\seq{z}|+1}).
\end{array}
\]
Since \(y_{i,0}\) does not occur in \(\seq{v}\) and \(v_{\ell''+\seql{z}+1}\)
(Lemma~\ref{lem:n-1:independence}),
\(\theta_0\seq{v}\) and \(\theta_{0}u_{\ell'+\seql{z}+1}\) are
equivalent to \(\theta_j\seq{v}\) and \(\theta_{j}u_{\ell'+\seql{z}+1}\) respectively
for any \(j\in\set{1,\ldots,k}\).
Therefore, the whole output of transformation is equivalent to:
\[
\begin{array}{l}
(\theta_0(u_0(v_0,\seq{v},v_{\ell''+|\seq{z}|+1})),
\theta_0(\seq{u}(\seq{v},v_{\ell''+|\seq{z}|})),\\\quad
\theta_{0}(u_{\ell'+1}(v_{\ell''+1},\seq{v},v_{\ell''+|\seq{z}|+1})),\ldots, \iftwocol \\\quad\fi \theta_{0}(u_{\ell'+\seql{z}}(v_{\ell''+\seql{z}},\seq{v},v_{\ell''+|\seq{z}|+1})),\\\quad
\theta_{1}(u_0(v_0,\seq{v},v_{\ell''+|\seq{z}|+1})),
\ldots,
\theta_{k}(u_0(v_0,\seq{v},v_{\ell''+|\seq{z}|+1})),\\\quad
\theta_{0}(u_{\ell'+\seql{z}+1}(v_{\ell'+\seql{z}+1},\seq{v},v_{\ell''+|\seq{z}|+1}))).\\
\end{array}
\]
Thus, we have the required result.
\item Case \rn{Tr-NT}: The result follows immediately from \rn{Tr-NT}.
\item Case \rn{Tr-AppG}: In this case, we have \(t=uv\) and:
\[
\begin{array}{l}
\seq{y}\COL \seq{\sty_y}; \seq{z}\pN u:\T^{\ell'+1}\to\T\tr (u_0,\ldots,u_{\ell'+|\seq{z}|+2})\\
\seq{y}\COL \seq{\sty_y}; \seq{z}\pN v:\T\tr (v_0,\ldots,v_{|\seq{z}|+1})\\
(\seq{t},t_{m+1}) = (u_0+u_1\cdot v_0, u_2,\ldots,u_{\ell'+1}, \iftwocol \\\qquad\qquad\fi u_{\ell'+2}+u_1\cdot v_{1},
\ldots, u_{\ell'+|\seq{z}|+2}+u_1\cdot v_{|\seq{z}|+1})\\
\sty = \T^{\ell'}\to\T. \end{array}
\]
By the induction hypothesis, we have:
\[
\begin{array}{l}
\emptyTE; \seq{z},x_1,\ldots,x_k\pN [\seq{s}/\seq{y}]u:\T^{\ell'+1}\to\T\tr (\theta_{0}u_0,\ldots,\theta_{0}u_{\ell'+|\seq{z}|+1}, \theta_{1}u_0,\ldots,
\theta_{k}u_0,
\theta_{0}u_{\ell'+|\seq{z}|+k+2})\\
\emptyTE; \seq{z},x_1,\ldots,x_k\pN [\seq{s}/\seq{y}]v:\T\tr \iftwocol\\\qquad\fi
(\theta_{0}v_0,\ldots,\theta_{0}v_{|\seq{z}|}, \theta_{1}v_{0},\ldots,
\theta_{k}v_0,\theta_{0}v_{|\seq{z}|+1}).
\end{array}
\]
By applying \rn{Tr-AppG}, we obtain:
\[
\begin{array}{l}
\emptyTE; \seq{z},x_1,\ldots,x_k\pN ([\seq{s}/\seq{y}]u)([\seq{s}/\seq{y}]v): \T^{\ell'}\to\T\tr \\\qquad
 (\theta_0u_0+\theta_0u_1\cdot \theta_0v_0,
  \theta_0u_2,\ldots, \theta_0u_{\ell'+1},\\\qquad
\theta_0u_{\ell'+2}+\theta_{0}u_1\cdot \theta_{0}v_1,\iftwocol\\\qquad\fi
\ldots, 
\theta_{0}u_{\ell'+\seql{z}+1}+\theta_0u_1\cdot \theta_0v_{\seql{z}},\\\qquad
\theta_{1}u_0+\theta_{0}u_1\cdot \theta_{1}v_0,
\ldots, \iftwocol\\\qquad\fi
\theta_{k}u_0+\theta_{0}u_1\cdot \theta_{k}v_0,\iftwocol\\\qquad\fi
\theta_{0}u_{\ell'+|\seq{z}|+2}+\theta_{0}u_1\cdot \theta_{0}v_{|\seq{z}|+1}).
\end{array}
\]
Since \(u_1\) does not contain any occurrence of \(y_{i,0}\)
(Lemma~\ref{lem:n-1:independence}),
\(\theta_0 u_1=\theta_j u_1\) for any \(j\in\set{1,\ldots,k}\).
Therefore, the output of the translation is equivalent to:
\[
\begin{array}{l}
 (\theta_0(u_0+u_1\cdot v_0),
  \theta_0u_2,\ldots, \theta_0u_{\ell'+1},\\\qquad
\theta_0(u_{\ell'+2}+u_1\cdot v_1),
\ldots, 
\theta_{0}(u_{\ell'+\seql{z}+1}+u_1\cdot v_{\seql{z}}),\\\qquad
\theta_{1}(u_0+u_1\cdot v_0),
\ldots, 
\theta_{k}(u_0+u_1\cdot v_0),\iftwocol\\\qquad\fi
\theta_{0}(u_{\ell'+|\seq{z}|+2}+u_1\cdot v_{|\seq{z}|+1})).
\end{array}
\]
Thus, we have  the required result.
\end{itemize}
\end{proof}

We are now ready to prove that the translation relation is preserved by
reductions (Lemmas~\ref{lem:tr-subj}, \ref{lem:tr-subj-ad1}, and
\ref{lem:tr-subj-ad2} below).

\begin{lemma}[Subject Reduction]
  \label{lem:tr-subj}
  If \(\emptyTE;\seq{x}\pN F\,\seq{s}\,\seq{t}:\T\tr (v_0,\ldots,v_{|\seq{x}|+1})\)
  and \(\RULES(F)=\lambda \seq{y}.\lambda \seq{z}.u_L\C{p}u_R\)
  (where \(\set{\seq{x}}\cap \set{\seq{z}}=\emptyset\)),
  then there exist \(w_{L,i},w_{R,i}\ (i\in\set{0,\ldots,|\seq{x}|+1})\) such that
  \(\emptyTE;\seq{x}\pN \Sub{\seq{z}}{\seq{t}}[\seq{s}/\seq{y}]u_d
  \tr (w_{d,0},\ldots,w_{d,|\seq{x}|+1})\)
  and \(v_i\cong pw_{L,i}+(1-p)w_{R,i}\) for each
  \(i\in\set{0,\ldots,|\seq{x}|+1}\).
\end{lemma}
\begin{proof}
By the assumptions,
we have:
  \[
  \begin{array}{l}
    \emptyTE;\seq{x}\pN s_i:\sty_i\tr
      (s_{i,0},\ldots,s_{i,\arity(\sty_i)+|\seq{x}|+1})\hfill \mbox{ for each \(i\in\set{1,\ldots,|\seq{s}|}\)}\\
    \emptyTE; \seq{x}\pN t_i:\T \tr (t_{i,0},\ldots,t_{i,|\seq{x}|+1})
    \mbox{ for each \(i\in\set{1,\ldots,\seql{z}}\)}\\
    v'_0 = F_0(s_{1,0},\ldots,s_{1,\arity(\sty_1)},s_{1,\arity(\sty_1)+\seql{x}+1})
    \cdots \iftwocol\\\qquad\fi(s_{|\seq{s}|,0},\ldots,
    s_{|\seq{s}|,\arity(\sty_{|\seq{s}|})},s_{|\seq{s}|,\arity(\sty_{|\seq{s}|})+\seql{x}+1})\\
    v'_j = F_j(s_{1,1},\ldots,s_{1,\arity(\sty_1)},s_{1,\arity(\sty_1)+\seql{x}+1})\cdots \iftwocol\\\qquad\fi(s_{|\seq{s}|,1},\ldots,
    s_{|\seq{s}|,\arity(\sty_{|\seq{s}|})},s_{|\seq{s}|,\arity(\sty_{|\seq{s}|})+\seql{x}+1}) \\
    \hfill \mbox{ for each $j\in\set{1,\ldots,\seql{z}}$}\\
    v'_{\seql{z}+j} = \iftwocol\\\ \fi
F_0(s_{1,\arity(\sty_1)+j},s_{1,1},\ldots,s_{1,\arity(\sty_1)},s_{1,\arity(\sty_1)+\seql{x}+1})\cdots\\
\iftwocol\ \else \qquad\qquad\quad\fi
    (s_{\seql{s},\arity(\sty_1)+j},s_{|\seq{s}|,1},\ldots,
    s_{|\seq{s}|,\arity(\sty_{|\seq{s}|})},s_{|\seq{s}|,\arity(\sty_{|\seq{s}|})+\seql{x}+1}) \\\hfill
    \mbox{ for each $j\in\set{1,\ldots,|\seq{x}|}$}\\
  v'_{\seql{z}+\seql{x}+1} =
F_{0}(s_{1,\arity(\sty_1)+\seql{x}+1},s_{1,1},\ldots,
    s_{1,\arity(\sty_1)}, s_{1,\arity(\sty_1)+\seql{x}+1})\cdots \\\qquad\qquad\qquad
    (s_{|\seq{s}|,\arity(\sty_{|\seq{s}|})+\seql{x}+1},s_{\seql{s},1}\ldots,
    s_{|\seq{s}|,\arity(\sty_{|\seq{s}|})},s_{|\seq{s}|,\arity(\sty_{|\seq{s}|})+\seql{x}+1})\\
    v_0 \cong v'_0+v'_1\cdot t_{1,0}+\cdots +v'_{\seql{z}}\cdot t_{\seql{z},0}\\
    v_{\seql{z}+i}\cong v'_{\seql{z}+i}+v'_1\cdot t_{1,i}+\cdots +v'_{\seql{z}}\cdot t_{\seql{z},i}
\iftwocol\\\hfill\fi
    \mbox{ for each $i\in\set{1,\ldots,|\seq{x}|+1}$}\\
    \seq{y}\COL\seq{\sty}; \seq{z}\pN u_d:\T
      \tr (u_{d,0},\ldots,u_{d,|\seq{z}|+1}) \mbox{ for \(d\in\set{L,R}\)}
  \end{array}
  \]
  By applying 
  the substitution lemma (Lemma~\ref{lem:n-1:subj}) to the last condition,
  we obtain:
  \[
  \begin{array}{l}
    \emptyTE ; \seq{z},\seq{x} \pN [\seq{s}/\seq{y}]u_d:\T\tr
    (w'_{d,0},\ldots,w'_{d,\seql{z}+|\seq{x}|+1})
  \end{array}
  \]
  where
  \[
  \begin{array}{l}
    (w'_{d,0},\ldots,w'_{d,\seql{z}+|\seq{x}|+1})= (\theta_0u_{d,0},
  \ldots, \theta_0u_{d,|\seq{z}|},
  \theta_1u_{d,0},\ldots,
  \theta_{\seql{x}}u_{d,0}, \theta_0u_{d,\seql{z}+1}) \\
\theta_j = \theta_{1,j}\cdots\theta_{\seql{s},j}
  \mbox{ for $j\in\set{0,\ldots,|\seq{x}|}$}\\
  \theta_{i,0}=[s_{i,0}/y_{i,0},\ldots,s_{i,\arity(\sty_i)}/y_{i,\arity(\sty_i)},
\iftwocol\\\qquad \fi s_{i,\arity(\sty_i)+\seql{x}+1}/y_{i,\arity(\sty_i)+1}]
  \mbox{ for $i\in\set{1,\ldots,|\seq{s}|}$}\\
  \theta_{i,j} = [s_{i,\arity(\sty_i)+j}/y_{i,0},\ldots,s_{i,\arity(\sty_i)}/y_{i,\arity(\sty_i)},
\iftwocol\\\qquad \fi s_{i,\arity(\sty_i)+\seql{x}+1}/y_{i,\arity(\sty_i)+1}]\\\hfill
  \mbox{ for $i\in\set{1,\ldots,|\seq{s}|}, j\in\set{1,\ldots,|\seq{x}|}$}.\\
  \end{array}
\]
Then, we have \(v_j' \cong pw'_{L,j'}+(1-p)w'_{R,j'}\).
(Here, the only non-trivial case is for \(j=\seql{z}+\seql{x}+1\),
where we need to show that \(v_j' \cong p\theta_{\seql{x}+1}u_{L,0}+(1-p)\theta_{\seql{x}+1}u_{R,0}\)
is equivalent to \(pw'_{L,j'}+(1-p)w'_{R,j'}\); in this case,
by induction on the structure of \(u_d\),
it follows that \(\theta_{\seql{x}+1}u_{d,0}
= \theta_0u_{d,\seql{z}+1}\), which implies the required property)
Let \(\seq{w}_d\) be \((w_{d,0},\ldots,w_{d,|\seq{x}|+1})\), where:
\[
\begin{array}{l}
    w_{d,0} = w'_{d,0}+w'_{d,1}\cdot t_{1,0}+\cdots +w'_{d,{\seql{z}}}\cdot t_{\seql{z},0}\\
    w_{d,i}= w'_{d,{\seql{z}+i}}+w'_{d,1}\cdot t_{1,i}+\cdots +w'_{d,{\seql{z}}}\cdot t_{\seql{z},i} \\\hfill\mbox{ for \(i\in\set{1,\ldots,\seql{x}+1}\)}.
\end{array}
\]
Then, the required result is obtained by applying \rn{Tr-Sub} to
  \[
  \begin{array}{l}
    \emptyTE ; \seq{z},\seq{x} \pN [\seq{s}/\seq{y}]u_d:\T\tr
    (w'_{d,0},\ldots,w'_{d,\seql{z}+|\seq{x}|+1}).
  \end{array}
  \]
\end{proof}

\begin{lemma}[Subject Reduction (for administrative steps)]
  \label{lem:tr-subj-ad1}
  If \(\emptyTE;\seq{x}\pN \Sub{\seq{z}}{\seq{s}}x_i:\T\tr (t_0,\ldots,t_{|\seq{x}|+1})\)
  with \(x_i\notin \set{\seq{z}}\), then 
  \(\emptyTE;\seq{x}\pN x_i:\T\tr (u_0,\ldots,u_{\seql{x}+1})\)
  for some \(u_j\ (j\in\set{0,\ldots,\seql{x}+1})\) such that \(t_j\cong u_j\) for each \(j\in\set{0,\ldots,|\seq{x}|+1}\).
\end{lemma}
\begin{proof}
  By the assumption  \(\emptyTE;\seq{x}\pN \Sub{\seq{z}}{\seq{s}}x_i:\T\tr (t_0,\ldots,t_{|\seq{x}|+1})\),  we have:
  \[
\begin{array}{l}
  \emptyTE;\seq{z},\seq{x}\pN x_i:\T\tr ({0}^{i+\seql{z}+1},1,{0}^{|\seq{x}|-i+1})\\
  \emptyTE;\seq{x}\pN s_j:\T\tr (s_{j,0},\ldots,s_{j,|\seq{x}|+1}) \mbox{ (for each $j\in\set{1,\ldots,\seql{s}}$)}\\
  t_i=1+\Sigma_{j=1}^{\seql{s}} 0\cdot s_{j,i} \cong 1\qquad t_{i'}\cong 0\mbox{ for \(i'\neq i\)}
\end{array}
\]
Since 
\(\emptyTE;\seq{x}\pN x_i:\T\tr (0^i,1,0^{|\seq{x}|-i+1})\),
we have the required condition for: \(u_i=1\) and \(u_j=0\) for \(j\neq i\).
\end{proof}

\begin{lemma}[Subject Reduction (for administrative steps (ii))]
  \label{lem:tr-subj-ad2}
  If \(\emptyTE;\seq{x}\pN \Sub{\seq{z}}{\seq{s}}z_i:\T\tr (t_0,\ldots,t_{|\seq{x}|+1})\), then
  \(\emptyTE;\seq{x}\pN s_i:\T\tr (u_0,\ldots,u_{\seql{x}+1})\)
  for some \(u_j\ (j\in\set{0,\ldots,\seql{x}+1})\) such that \(t_j\cong u_j\) for each \(j\in\set{0,\ldots,|\seq{x}|+1}\).
\end{lemma}
\begin{proof}
  By the assumption  \(\emptyTE;\seq{x}\pN \Sub{\seq{z}}{\seq{s}}z_i:\T\tr (t_0,\ldots,t_{|\seq{x}|+1})\),  we have:
  \[
\begin{array}{l}
  \emptyTE;\seq{z},\seq{x}\pN z_i:\T\tr (0^{i},1,{0}^{\seql{z}-i+|\seq{x}|+1})\\
  \emptyTE;\seq{x}\pN s:\T\tr (s_0,\ldots,s_{|\seq{x}|+1})\\
  t_j\cong0+1\cdot s_j\cong s_j \mbox{ for each \(j\in\set{0,\ldots,|\seq{x}|+1}\)}.
\end{array}
\]
Thus, the result holds for \(u_j = s_j\).
\end{proof}

We are now ready to prove Theorem~\ref{prop:order-k-1-fixpoint}, restricted to
recursion-free \pHORSs{} (the definition of recursion-free \pHORSs{} is found
in Section~\ref{sec:proof-tr-n-1}).
\begin{lemma}
\label{lem:tr-for-recfree}
  Let \(\GRAM=(\NONTERMS,\RULES,S)\) be a recursion-free \pHORS{},
  and \(\rho\) be the least solution of \(\EQref{\GRAM}\).
  Then, \(\Prob(\GRAM,S\,\Te) = \rho(S_1)\).
\end{lemma}
\begin{proof}
  By Lemma~\ref{lem:red-vs-redes},
  it suffices to show that
  \(\emptyTE; x \pN t\tr (t_0,t_1,t_2)\) implies 
  \(\ProbES(\GRAM,t,x)\cong t_1\).
  This follows by induction on \(\flat(t)\), the length of the longest reduction
  sequence from \(t\) by \(\newredp{d,p}{}\);
  note that \(\flat(t)\) is well defined because
  \(t\) is finitely branching and
  strongly normalizing with respect to \(\newredp{d,p}{}\);
  the strong normalization follows from Lemma~\ref{lem:red-vs-esred} and
  the fact
  that there can be no infinite sequence of consecutive administrative reductions
  (in fact, each administrative reduction consumes one explicit substitution).

  Suppose \(\flat(t)=0\). Then, \(t\) is either \(x\) (in which case, \(t_1=1\))
  or \(\Omega\) (in which case, \(t_1=0\)). Thus, the result follows immediately.

  If \(\flat(t)>0\), we perform a case analysis on a reduction of \(t\).
  If \(t\newredp{\epsilon,1}t'\), then by Lemmas~\ref{lem:tr-subj-ad1}
  and \ref{lem:tr-subj-ad2}, 
  \(\emptyTE; x \pN t\tr (t'_0,t'_1,t'_2)\) and \(t_1\cong t'_1\) for some \((t'_0,t'_1,t'_2)\).
  By the induction hypothesis, 
  \(\ProbES(\GRAM,t',x)\cong t'_1\). Therefore, we have
  \(\ProbES(\GRAM,t,x)=\ProbES(\GRAM,t',x)\cong t'_1\cong t_1\) as required.

  If \(t\newredp{d,p}t'\) for \(d\in\set{L,R}\), then
  \(t\) must be of the form \(E[F\,\seq{s}\,\seq{t}]\),
  \(\RULES(F)=\lambda \seq{y}.\lambda \seq{z}.u_L\C{p}u_R\), 
  \(t\newredp{L,p} t_L\), and   \(t\newredp{R,1-p} t_R\) with
  \(t_L=E[\Sub{\seq{z}}{\seq{t}}[\seq{s}/\seq{y}]u_L]\)
  and 
  \(t_R= E[\Sub{\seq{z}}{\seq{t}}[\seq{s}/\seq{y}]u_R]\).
  By Lemma~\ref{lem:tr-subj}, there exist
  \((t_{L,0},t_{L,1}, t_{L,2})\) and   \((t_{R,0},t_{R,1}, t_{R,2})\) such that
  \(\emptyTE;x\pN t_L\tr (t_{L,0},t_{L,1}, t_{L,2})\)
  and 
  \(\emptyTE;x\pN t_R\tr (t_{R,0},t_{R,1}, t_{R,2})\)
  with \(t_1 \cong pt_{L,1}+(1-p)t_{R,1}\).
  Thus, we have
  \(\ProbES(\GRAM,t,x)=p\ProbES(\GRAM,t_L,x)+(1-p)\ProbES(\GRAM,t_R,x)
  \cong pt_{L,1}+(1-p)t_{R,1}\cong t_1\) as required.
\end{proof}

We are now ready to prove Theorem~\ref{prop:order-k-1-fixpoint}.

\ifacm
\begin{proof}[Proof of Theorem~\ref{prop:order-k-1-fixpoint}]
\else
\begin{proofn}{Proof of Theorem~\ref{prop:order-k-1-fixpoint}}
\fi
  Consider the finite approximation \(\GRAM^{(k)}\) (defined in 
  Section~\ref{sec:proof-tr-n-1}).
  Then, we have
  \[
\begin{array}{l}
  \Prob(\GRAM, S\,\Te) = \bigsqcup_k \Prob(\GRAM^{(k)}, S^{(k)}\,\Te) =
  \bigsqcup_k \LFP(\F_{\GRAM^{(k)}})(S_1^{(k)}) \\
  = 
\bigsqcup_k \F_{\GRAM}^k(\bot) (S_1) = \LFP(\F_{\GRAM})(S_1)=\rho_{\EQref{\GRAM}}(S_1)
\end{array}
\]
as required.  Here \(\F_\GRAM\) is the functional associated with fixpoint equations
\(\EQref{\GRAM}\), as defined in Section~\ref{sec:ho-fixpoint},
and we can use essentially the same reasoning as in the proofs of
Lemma~\ref{lem:approximation} and Theorem~\ref{prop:order-k-fixpoint}.
\ifacm
\end{proof}
\else
\end{proofn}
\fi

\section{On Remark~\ref{rem:approximation}}
\label{sec:chatgpt}
This section settles the questions in
Remark~\ref{rem:approximation}, which were left open
in our original publication in LMCS, and further proves that
almost sure termination of \pHORS{} is \(\Pi^0_2\)-complete.
The following is an updated version of Table~\ref{tab:complexity}.
\begin{center}
\begin{tabular}{|l|l|l|l|}
\hline
Models  & \(\SETofPHORS_{>0}\) & \(\SETofPHORS_{>r}\) 
(\(r\in (0,1)\))& \(\SETofPHORS_{<r}\) (\(r\in (0,1]\))\\\hline
RMC  & P  & PSPACE & PSPACE\\\hline
\pHORS{} &  \((k-1)\)-EXPTIME-complete &\( \Sigma^0_1\)-complete
& \( \Sigma^0_2\)-complete\\\hline
Turing-complete lang.  & \( \Sigma^0_1\)-complete  & \( \Sigma^0_1\)-complete  & \( \Sigma^0_2\)-complete \\
\hline
\end{tabular}
\end{center}

The proofs provided below are solely due to ChatGPT (5.6 Sol Ultra, normal speed)
in 84 minutes, although
they were confirmed by the authors.\footnote{The proof in Section~\ref{app:further-undecidability} was generated by ChatGPT in response to Hiroyuki Katsura's request to solve the open
  questions in Remark~\ref{rem:approximation}.
  The proof in Section~\ref{app:ast-pi2-hardness} was then generated by ChatGPT
  after a few interactions with the authors of this paper.}
\newcommand{\Prb}{\mathop{\mathrm{Pr}}}
\newcommand{\Mul}{\mathsf{Mul}}
\newcommand{\Div}{\mathsf{Div}}
\newcommand{\Base}{\mathsf{Base}}
\renewcommand{\Zero}{\mathsf{Zero}}
\newcommand{\halt}{\mathsf{halt}}
\newcommand{\INC}{\mathsf{INC}}
\newcommand{\gotoop}{\mathsf{goto}}
\newcommand{\Conj}{\mathsf{Conj}}
\newcommand{\Amp}{\mathsf{Amp}}
\newcommand{\Swap}{\mathsf{Swap}}

\subsection{Further Undecidability Results for Order-2 PHORS}
\label{app:further-undecidability}

In this section, we settle negatively the open questions (iv)--(vi) in
Remark~\ref{rem:approximation}.  The key ingredient is a direct simulation of a deterministic
two-counter Minsky machine by an order-2 PHORS with a fixed gap in termination
probability.

\begin{theorem}
\label{thm:minsky-gap}
Given a deterministic two-counter Minsky machine $M$, one can effectively
construct an order-2 PHORS $\mathcal G_M$ such that
\[
 M\text{ halts}\quad\Longrightarrow\quad
 \mathrm{Pr}(\mathcal G_M)>\frac{14}{15},
\]
whereas
\[
 M\text{ does not halt}\quad\Longrightarrow\quad
 \mathrm{Pr}(\mathcal G_M)<\frac{1}{15}.
\]
\end{theorem}

\begin{proof}
Let
\[
 C=o\to o\to o.
\]
We regard a term $f:C$ that almost surely chooses one of its two arguments as
a probabilistic chooser.  Let $L_f$ and $R_f$ be the probabilities with which
$f$ chooses its left and right arguments, respectively, and suppose that
$L_f,R_f>0$ and $L_f+R_f=1$.  Define the \emph{odds} of $f$ by
\[
 \omega(f)=\frac{R_f}{L_f}.
\]

We first define two order-2 nonterminals that perform multiplication and
division on odds:
\begin{align*}
 \mathsf{Mul}\ f\ g\ x\ y
   & = f\bigl(g\,x\,(\mathsf{Mul}\ f\ g\ x\ y)\bigr)
          \bigl(g\,(\mathsf{Mul}\ f\ g\ x\ y)\,y\bigr),\\
 \mathsf{Div}\ f\ g\ x\ y
   & = f\bigl(g\,(\mathsf{Div}\ f\ g\ x\ y)\,x\bigr)
          \bigl(g\,y\,(\mathsf{Div}\ f\ g\ x\ y)\bigr).
\end{align*}
For $\mathsf{Mul}$, one trial exits to $x$ with probability $L_fL_g$, exits
to $y$ with probability $R_fR_g$, and repeats with probability
$L_fR_g+R_fL_g$.  Since $L_fL_g+R_fR_g>0$, the probability of repeating
forever is $0$.  Therefore,
\[
 \omega(\mathsf{Mul}(f,g))
 =\frac{R_fR_g}{L_fL_g}
 =\omega(f)\omega(g).
\]
Similarly, one trial of $\mathsf{Div}$ exits to $x$ and $y$ with probabilities
$L_fR_g$ and $R_fL_g$, respectively, and hence
\[
 \omega(\mathsf{Div}(f,g))
 =\frac{R_fL_g}{L_fR_g}
 =\frac{\omega(f)}{\omega(g)}.
\]
These identities are valid even when $x$ and $y$ are arbitrary recursively
defined continuations.  Indeed, writing $X=\mathrm{Pr}(x)$ and $Y=\mathrm{Pr}(y)$, the
termination probabilities $m$ and $d$ of the two terms satisfy
\begin{align*}
 m &=L_fL_gX+R_fR_gY+(L_fR_g+R_fL_g)m,\\
 d &=L_fR_gX+R_fL_gY+(L_fL_g+R_fR_g)d,
\end{align*}
and in each equation the coefficient of the recursive term is strictly less
than $1$.

We next encode the two counters.  Define
\begin{align*}
 \mathsf{Base}\ x\ y &= x\oplus_{16/17}y,\\
 \mathsf{Zero}\ x\ y &= x\oplus_{1/2}y.
\end{align*}
Thus
\[
 \omega(\mathsf{Base})=\beta_0=\frac1{16},
 \qquad
 \omega(\mathsf{Zero})=1.
\]
At simulation time $s$, the current base $b$ and the chooser $c_i$
representing the $i$-th counter are maintained so that
\[
 \omega(b)=\beta_s,
 \qquad
 \omega(c_i)=\beta_s^{n_i},
\]
where $n_i\in\mathbb N$ is the current value of the counter.  At the beginning
of every simulated machine step, let
\[
 b^+=\mathsf{Mul}(b,b),
 \qquad
 c_i^+=\mathsf{Mul}(c_i,c_i).
\]
Then
\[
 \beta_{s+1}=\beta_s^2,
 \qquad
 \omega(c_i^+)=\beta_{s+1}^{n_i}.
\]

To test whether counter $i$ is zero, use
\[
 z_i=\mathsf{Div}(c_i^+,b).
\]
We have
\[
 \omega(z_i)=\beta_s^{2n_i-1}.
\]
We interpret the left branch of $z_i$ as ``nonzero'' and the right branch as
``zero''.  If $n_i=0$, then the odds are $\beta_s^{-1}$, and the probability
of incorrectly choosing the left branch is
\[
 \frac{\beta_s}{1+\beta_s}<\beta_s.
\]
If $n_i>0$, the odds are at most $\beta_s$, and the probability of incorrectly
choosing the right branch is
\[
 \frac{\beta_s^{2n_i-1}}{1+\beta_s^{2n_i-1}}
 \leq \frac{\beta_s}{1+\beta_s}
 <\beta_s.
\]
After a correct zero-test decision, the counter encodings are updated exactly:
\[
\begin{array}{c|c|c}
 \text{operation}&\text{new chooser}&\text{odds}\\ \hline
 \text{unchanged}&c_i^+&\beta_{s+1}^{n_i}\\
 \text{increment}&\mathsf{Mul}(c_i^+,b^+)&\beta_{s+1}^{n_i+1}\\
 \text{decrement}&\mathsf{Div}(c_i^+,b^+)&\beta_{s+1}^{n_i-1}.
\end{array}
\]

We now compile the control flow of $M$.  For every control label $\ell$, add a
nonterminal
\[
 Q_\ell:C\to C\to C\to o\to o,
\]
whose arguments are the base $b$, the two counter encodings $c_1,c_2$, and a
continuation $k$.  For the halting instruction, put
\[
 Q_{\mathsf{halt}}\ b\ c_1\ c_2\ k=k.
\]
For an instruction $\mathsf{INC}(i);\mathsf{goto}\ q$, first replace $b,c_1,c_2$
by $b^+,c_1^+,c_2^+$, multiply the $i$-th counter once more by $b^+$, and
continue with $Q_q$.  For an instruction
\[
 \mathsf{if}\ C_i=0\ \mathsf{goto}\ q_0\
 \mathsf{else}\ \mathsf{DEC}(i);\mathsf{goto}\ q_1,
\]
use
\begin{align*}
 Q_\ell\ b\ c_1\ c_2\ k
 ={}&\mathsf{Div}(c_i^+,b)\\
 &\bigl(Q_{q_1}\ b^+\ \cdots\
       \mathsf{Div}(c_i^+,b^+)\ \cdots\ k\bigr)\\
 &\bigl(Q_{q_0}\ b^+\ \cdots\ c_i^+\ \cdots\ k\bigr).
\end{align*}
Here the first continuation corresponds to the nonzero branch and the second
to the zero branch.  The abbreviations involving $b^+$ and $c_i^+$ can be
expanded by introducing auxiliary nonterminals.  Finally, let
\[
 S=Q_{\mathsf{start}}\ \mathsf{Base}\ \mathsf{Zero}\ \mathsf{Zero}\ e.
\]
Since $C$ has order $1$, while $\mathsf{Mul}$, $\mathsf{Div}$, and the
nonterminals $Q_\ell$ have order $2$, the resulting PHORS $\mathcal G_M$ has
order $2$.

It remains to estimate its termination probability.  By construction,
\[
 \beta_s=\left(\frac1{16}\right)^{2^s}.
\]
Conditioned on all previous zero-test decisions being correct, the probability
of an incorrect decision at step $s$ is less than $\beta_s$.  Hence, by a
union bound over the first incorrect decision, the probability $E$ that some
incorrect zero-test decision occurs satisfies
\[
 E
 <\sum_{s\geq0}\left(\frac1{16}\right)^{2^s}
 \leq\sum_{s\geq0}\left(\frac1{16}\right)^{s+1}
 =\frac1{15}.
\]
Every invocation of $\mathsf{Mul}$ and $\mathsf{Div}$ exits to one of its two
arguments almost surely.

If $M$ halts, then whenever no zero-test error occurs, the simulation follows
the unique computation of $M$ and eventually reaches $e$.  Therefore
\[
 \mathrm{Pr}(\mathcal G_M)\geq1-E>\frac{14}{15}.
\]
If $M$ does not halt, a terminating run of $\mathcal G_M$ must contain at least
one incorrect zero-test decision; otherwise it follows the infinite computation
of $M$.  Thus
\[
 \mathrm{Pr}(\mathcal G_M)\leq E<\frac1{15}.
\]
This proves the theorem.
\end{proof}

\begin{corollary}
\label{cor:remark34-iv-v}
For every rational number $r\in(0,1)$, neither $\Psi_{<r}$ nor
$\Psi_{\leq r}$ is recursively enumerable.  In particular, propositions
(iv) and (v) in Remark~\ref{rem:approximation} are false.
\end{corollary}

\begin{proof}
First take $r=1/2$.  By Theorem~\ref{thm:minsky-gap},
\[
 M\text{ does not halt}
 \quad\Longleftrightarrow\quad
 \mathrm{Pr}(\mathcal G_M)<\frac12
 \quad\Longleftrightarrow\quad
 \mathrm{Pr}(\mathcal G_M)\leq\frac12.
\]
Since the set of nonhalting deterministic two-counter Minsky machines is not
recursively enumerable, neither $\Psi_{<1/2}$ nor $\Psi_{\leq1/2}$ is
recursively enumerable.

For an arbitrary rational $r\in(0,1)$, choose a rational $v$ such that
\[
 0<v<2\min(r,1-r),
\]
and put $u=r-v/2$.  From $\mathcal G_M$, construct an order-2 PHORS
$\mathcal G'_M$ by a finite probabilistic wrapper such that
\[
 \mathrm{Pr}(\mathcal G'_M)=u+v\mathrm{Pr}(\mathcal G_M).
\]
Such a wrapper is obtained by probabilistically choosing among immediate
termination $e$, the start symbol of $\mathcal G_M$, and divergence $\Omega$,
with respective probabilities $u$, $v$, and $1-u-v$.

If $M$ does not halt, then
\[
 \mathrm{Pr}(\mathcal G'_M)
 <u+\frac{v}{15}
 =r-\frac{13v}{30}
 <r,
\]
whereas if $M$ halts, then
\[
 \mathrm{Pr}(\mathcal G'_M)
 >u+\frac{14v}{15}
 =r+\frac{13v}{30}
 >r.
\]
Consequently,
\[
 M\text{ does not halt}
 \quad\Longleftrightarrow\quad
 \mathcal G'_M\in\Psi_{<r}
 \quad\Longleftrightarrow\quad
 \mathcal G'_M\in\Psi_{\leq r},
\]
and the claim follows.
\end{proof}

\begin{corollary}
\label{cor:remark34-vi}
There is no algorithm that, given an order-2 PHORS $\mathcal G$ and a rational
number $\epsilon>0$, always returns a rational number $a$ satisfying
\[
 |\mathrm{Pr}(\mathcal G)-a|<\epsilon.
\]
Thus proposition (vi) in Remark~\ref{rem:approximation} is false.
\end{corollary}

\begin{proof}
Suppose that such an algorithm existed.  Apply it to $\mathcal G_M$ of
Theorem~\ref{thm:minsky-gap} with $\epsilon=1/6$.  If $M$ halts, the returned
number $a$ satisfies
\[
 a>\frac{14}{15}-\frac16=\frac{23}{30}>\frac12,
\]
whereas if $M$ does not halt, it satisfies
\[
 a<\frac1{15}+\frac16=\frac7{30}<\frac12.
\]
Thus comparison of $a$ with $1/2$ would decide whether $M$ halts, a
contradiction.
\end{proof}

\begin{remark}
Theorem~\ref{thm:minsky-gap} does not contradict the decidability of
qualitative reachability for PHORS.  Even if $M$ does not halt, a zero-test can
be answered incorrectly with positive probability, and the resulting erroneous
simulation may subsequently reach $e$.  Hence the constructed PHORS may have a
terminating path even in the nonhalting case.  The reduction exploits the gap
in the \emph{total probability} of termination, rather than the existence of a
terminating path.
\end{remark}

\subsection{Almost-Sure Termination of Order-2 PHORS is \(\boldsymbol{\Pi^0_2}\)-Hard}
\label{app:ast-pi2-hardness}

The gap construction in the preceding subsection is sufficient for the
negative approximation results, but it does not by itself characterize
almost-sure termination: an erroneous simulation may run forever.  We now
refine the construction so that every simulation containing an erroneous
zero-test terminates almost surely, while the faithful infinite simulation
still has positive probability.  This yields the missing lower bound for
almost-sure termination.

We continue to use the type $C=o\to o\to o$ and the notation $L_h$ from the
proof of Theorem~\ref{thm:minsky-gap} for the probability with which a total
chooser $h:C$ selects its left argument.  When both outcomes of $h$ have
positive probability, this agrees with the odds notation via
\[
 L_h=\frac{1}{1+\omega(h)}.
\]
All choosers used below are total: they select one of their arguments with
probability $1$.  Notice that $L_h$ is also defined for a deterministic
chooser, although its odds need not be defined under the convention of the
preceding proof.

We first describe two operations on confidence choosers.  Define
\[
 \mathsf{Swap}\ h\ x\ y=h\ y\ x
\]
and
\[
 \mathsf{Conj}\ h\ q\ x\ y=h\,(q\ x\ y)\,y.
\]
Then
\begin{equation}
\label{eq:confidence-operations}
 L_{\mathsf{Swap}(h)}=1-L_h,
 \qquad
 L_{\mathsf{Conj}(h,q)}=L_hL_q.
\end{equation}
Let $\mathsf{Amp}(h)$ invoke $h$ independently five times and select its left
argument if and only if at least three invocations select their left
argument.  This is a finite decision tree and is therefore definable by
order-2 nonterminals.  More explicitly, we may introduce nonterminals
$A_{i,j}:C\to o\to o\to o$ with
\[
 A_{i,j}\ h\ x\ y
 =h\,(A_{i-1,j-1}\ h\ x\ y)\,(A_{i-1,j}\ h\ x\ y),
\]
where $A_{i,0}\ h\ x\ y=x$ and $A_{i,j}\ h\ x\ y=y$ for $j>i$, and put
$\mathsf{Amp}(h)=A_{5,3}h$.  If $p=L_h$, then
\begin{equation}
\label{eq:amplification-polynomial}
 L_{\mathsf{Amp}(h)}=A(p):=10p^3-15p^4+6p^5.
\end{equation}
The following elementary estimates will be used repeatedly:
\begin{equation}
\label{eq:amplification-estimates}
 1-A(1-\delta)\leq 10\delta^3,
 \qquad
 A(p)\leq 10p^3
 \quad(0\leq\delta,p\leq1).
\end{equation}

\begin{lemma}
\label{lem:audited-simulation}
Given a deterministic two-counter Minsky machine $M$ and an input $n$, one
can effectively construct an order-2 PHORS $\mathcal H_{M,n}$ such that
\begin{equation}
\label{eq:one-input-equivalence}
 \mathcal H_{M,n}\text{ is almost surely terminating}
 \quad\Longleftrightarrow\quad
 M\text{ halts on }n.
\end{equation}
\end{lemma}

\begin{proof}
We modify the simulation in the proof of
Theorem~\ref{thm:minsky-gap}.  Put
\[
 \beta_0=\frac1{256},
 \qquad
 \beta_{s+1}=\beta_s^2.
\]
At simulation time $s$, the base $b$ and counter choosers $c_1,c_2$ satisfy
\begin{equation}
\label{eq:ast-counter-invariant}
 \omega(b)=\beta_s,
 \qquad
 \omega(c_i)=\beta_s^{n_i}.
\end{equation}
The initial base is the rational chooser
$\mathsf{Base}_{\mathrm{ast}}\ x\ y=x\oplus_{256/257}y$; the chooser
$\mathsf{Zero}$ from the preceding proof represents counter value $0$.
For the given input $n$, let $C_n$ be the finite term obtained from
$\mathsf{Zero}$ by multiplying its odds by
$\mathsf{Base}_{\mathrm{ast}}$ exactly $n$ times.  Thus
\[
 \omega(C_n)=\beta_0^n.
\]
The initial configuration uses $C_n$ and $\mathsf{Zero}$ as the first and
second counter encodings, respectively.  (As usual, the machine is assumed
to receive its input in the first counter.)
As before, at the beginning of a simulated step we set
\[
 b^+=\mathsf{Mul}(b,b),
 \qquad
 c_i^+=\mathsf{Mul}(c_i,c_i).
\]
An unchanged counter is represented by $c_i^+$, an increment by
$\mathsf{Mul}(c_i^+,b^+)$, and a decrement by
$\mathsf{Div}(c_i^+,b^+)$.  Thus invariant
\eqref{eq:ast-counter-invariant} is preserved along every faithful simulation
prefix.  After the first erroneous zero-test outcome, a counter encoding need
not represent a natural number; the argument below uses only the fact that
all the resulting choosers remain total.

For a zero-test of counter $i$, let
\[
 z_i=\mathsf{Div}(c_i^+,b).
\]
As in the preceding proof,
\begin{equation}
\label{eq:ast-zero-test-odds}
 \omega(z_i)=\beta_s^{2n_i-1}.
\end{equation}
The left outcome of $z_i$ is interpreted as ``nonzero'' and the right outcome
as ``zero.''  If $n_i=0$, the right outcome has probability
$1/(1+\beta_s)>1-\beta_s$; if $n_i>0$, the left outcome has probability at
least $1/(1+\beta_s)>1-\beta_s$.  Hence the correct outcome has probability
at least $1-\beta_s$, and the incorrect outcome has probability at most
$\beta_s$.

In addition to the base and the two counter encodings, each control state
carries a total chooser $h:C$, called its \emph{confidence}; initially $h$
is the deterministic chooser $\mathsf{Top}\ x\ y=x$.  Thus the start term of
$\mathcal H_{M,n}$ is
\[
 Q_{\mathsf{start}}\ \mathsf{Base}_{\mathrm{ast}}\ C_n\
 \mathsf{Zero}\ \mathsf{Top}.
\]
The halting control state immediately returns $e$.  Every nonhalting state
first audits the current confidence:
\begin{equation}
\label{eq:confidence-audit}
 Q_\ell\ b\ c_1\ c_2\ h
 =h\,(B_\ell\ b\ c_1\ c_2\ h)\,e,
\end{equation}
where $B_\ell$ implements the instruction at label $\ell$.  Thus a failed
audit leads to harmless termination.  Within $B_\ell$, at every simulated
step, first replace $h$ by $\mathsf{Amp}(h)$.  For an increment instruction, this amplified
chooser is passed unchanged to the next control state.  At a zero-test, the
confidence for the selected branch is conjoined with the probability that
the selected outcome of $z_i$ was the correct one.  Accordingly, the
nonzero and zero branches receive, respectively,
\begin{equation}
\label{eq:confidence-update}
 \mathsf{Conj}(\mathsf{Amp}(h),z_i)
 \quad\text{and}\quad
 \mathsf{Conj}(\mathsf{Amp}(h),\mathsf{Swap}(z_i)).
\end{equation}
For example, for an instruction
\[
 \mathsf{if}\ C_i=0\ \mathsf{goto}\ q_0\
 \mathsf{else}\ \mathsf{DEC}(i);\mathsf{goto}\ q_1,
\]
the instruction body is
\begin{equation}
\label{eq:audited-zero-test-rule}
\begin{aligned}
 B_\ell\ b\ c_1\ c_2\ h
 ={}&z_i\\
 &\bigl(Q_{q_1}\ b^+\ \cdots\
       \mathsf{Div}(c_i^+,b^+)\ \cdots\
       \mathsf{Conj}(\mathsf{Amp}(h),z_i)\bigr)\\
 &\bigl(Q_{q_0}\ b^+\ \cdots\ c_i^+\ \cdots\
       \mathsf{Conj}(\mathsf{Amp}(h),\mathsf{Swap}(z_i))\bigr).
\end{aligned}
\end{equation}
Here the first continuation is the nonzero branch and the second is the zero
branch; the omitted counter is replaced by its squared encoding.  The
occurrence of $z_i$ that makes the outer control-flow choice and the occurrence
of $z_i$ stored in the selected confidence are distinct.  Under the
call-by-name rewriting semantics of PHORS, they are evaluated independently.
The same applies to every later occurrence of a chooser stored in a
confidence term.

The scheme need not determine whether the selected branch is correct.  If it
is correct, the second factor in \eqref{eq:confidence-update} has
left-selection probability at least
$1-\beta_s$; if it is incorrect, that probability is at most $\beta_s$.

Consider first the faithful computation of $M$.  Let $a_s$ be the
left-selection probability of the confidence passed to the state at time
$s$, and put $\delta_s=1-a_s$.  Let $e_s$ be the error probability of the
zero-test at step $s$, with $e_s=0$ when that step contains no zero-test.
Along the faithful computation, \eqref{eq:confidence-operations},
\eqref{eq:amplification-polynomial}, and
\eqref{eq:confidence-update} give
\[
 a_{s+1}=A(a_s)(1-e_s),
 \qquad 0\leq e_s\leq\beta_s.
\]
It follows from \eqref{eq:amplification-estimates} that
\begin{equation}
\label{eq:confidence-error-recurrence}
 \delta_{s+1}\leq10\delta_s^3+\beta_s.
\end{equation}
Since $\delta_0=0$, induction yields
\begin{equation}
\label{eq:confidence-error-bound}
 \delta_{s+1}\leq2\beta_s.
\end{equation}
For $s=0$, this follows directly from
$\delta_1\leq\beta_0$.  For $s\geq1$, using
$\beta_s=\beta_{s-1}^2$, $\beta_0<1/80$, and the induction hypothesis, we
obtain
\[
 \delta_{s+1}
 \leq10(2\beta_{s-1})^3+\beta_s
 =(80\beta_{s-1}+1)\beta_s
 \leq2\beta_s.
\]
Since $(\beta_s)_{s\geq0}$ is summable,
\eqref{eq:confidence-error-bound} implies
\[
 \sum_s(1-a_s)<\infty,
 \qquad
 \sum_s e_s<\infty.
\]
Therefore the standard criterion for infinite products gives
\begin{equation}
\label{eq:faithful-positive-product}
 \prod_s a_s(1-e_s)>0.
\end{equation}
For $N\geq1$, let $F_N$ be the event that the first $N$ audits succeed and
that every zero-test among the first $N$ simulated steps takes its correct
branch.  Each chooser occurrence is freshly evaluated under the call-by-name
semantics, and conditioning successively on the preceding faithful prefix
therefore gives
\[
 \Prb(F_N)=\prod_{s<N}a_s(1-e_s).
\]
If the computation of $M$ is infinite, the events $F_N$ form a decreasing
sequence, and continuity from above together with
\eqref{eq:faithful-positive-product} yields
\[
 \Prb\Bigl(\bigcap_{N\geq1}F_N\Bigr)
 =\lim_{N\to\infty}\prod_{s<N}a_s(1-e_s)>0.
\]
On this event the PHORS does not terminate.  Thus $\mathcal H_{M,n}$ is not
almost surely terminating.

It remains to show that erroneous infinite simulations have probability $0$.
Suppose that the first incorrect zero-test outcome is selected at time $t$.
By \eqref{eq:confidence-update}, the confidence passed to the next state has
left-selection
probability
\[
 a_{t+1}\leq\beta_t\leq\frac1{256}.
\]
Every later confidence factor is at most $1$, so
\eqref{eq:amplification-estimates} gives
\begin{equation}
\label{eq:wrong-confidence-recurrence}
 a_{s+1}\leq A(a_s)\leq10a_s^3
 \qquad(s>t).
\end{equation}
In particular, $a_s$ tends to $0$.  Continuing the simulation forever would
require every subsequent audit in \eqref{eq:confidence-audit} to succeed.  To
make the resulting estimate uniform over all continuations after the first
error, define
\[
 q_0=\beta_t,
 \qquad
 q_{k+1}=10q_k^3.
\]
By induction from \eqref{eq:wrong-confidence-recurrence}, after surviving $k$
further simulated steps, the left-selection probability of the current
confidence is at most $q_k$, independently of the control-flow and zero-test
outcomes taken after time $t$.  Since $q_0\leq1/256$, the sequence $(q_k)$
tends to $0$.

Let $D_N$ be the event that, after the first error at time $t$, the simulation
survives the next $N$ audits without reaching the halting state.  Conditioning
successively on the histories before those audits gives
\[
 \Prb(D_N\mid\text{the first error occurs at time }t)
 \leq\prod_{k=0}^{N-1}q_k
 \leq q_{N-1}.
\]
The right-hand side tends to $0$.  By continuity from above, the probability
of surviving all subsequent audits is therefore $0$.  Hence, conditioned on
any fixed first error time $t$, the erroneous simulation either reaches the
machine's halting state or is sent to $e$ almost surely.

If $M$ halts on $n$, a simulation with no zero-test error reaches the halting
state after finitely many steps.  Every infinite run must therefore have a
first error at some finite time.  By the preceding paragraph, for each fixed
first error time the probability of an infinite continuation is $0$; the
countable union over all possible $t$ also has probability $0$.  Thus
$\mathcal H_{M,n}$ is almost surely terminating.  This proves
\eqref{eq:one-input-equivalence}.

Finally, all data arguments $b,c_1,c_2,h,z_i$ have type $C$, of order $1$.
The nonterminals that receive them, including $\mathsf{Mul}$,
$\mathsf{Div}$, $\mathsf{Conj}$, $\mathsf{Amp}$, and the control-state
nonterminals, have order at most $2$.  Thus the construction is indeed an
order-2 PHORS.
\end{proof}

\begin{theorem}
\label{thm:ast-pi2-hard}
Almost-sure termination of order-2 PHORS is $\Pi^0_2$-hard.
\end{theorem}

\begin{proof}
We reduce from the totality problem
\[
 \mathsf{TOT}
 =\{M\mid \text{$M$ halts on every input $n\in\mathbb N$}\},
\]
which is $\Pi^0_2$-complete.  The construction in
Lemma~\ref{lem:audited-simulation} is uniform in the input.  For any chooser
$c:C$, let $H_M(c)$ denote the start term
\[
 Q_{\mathsf{start}}\ \mathsf{Base}_{\mathrm{ast}}\ c\
 \mathsf{Zero}\ \mathsf{Top}
\]
of the audited simulation.  If $\omega(c)=\beta_0^n$, the proof of the lemma
applies to $H_M(c)$ with input $n$.

Starting with $c=\mathsf{Zero}$, repeatedly make a fair probabilistic choice
between running the simulation for the currently encoded input and proceeding
to the next input:
\begin{equation}
\label{eq:input-enumeration}
E_M(c)
 =H_M(c)
   \mathbin{\oplus_{1/2}}
   E_M(\mathsf{Mul}(c,\mathsf{Base}_{\mathrm{ast}})).
\end{equation}
Input $n$ is selected with probability $2^{-(n+1)}>0$, while the probability
of postponing the selection forever is $0$.

If $M$ fails to halt on some input $n$,
\eqref{eq:input-enumeration} selects that input with positive
probability, and Lemma~\ref{lem:audited-simulation} gives a positive
conditional probability of divergence.  Hence $E_M(\mathsf{Zero})$ is not
almost surely terminating.  Conversely, if $M$ halts on every input, every
selected simulation terminates almost surely by the lemma; the only remaining
infinite behavior is to postpone the selection forever, an event of
probability $0$.  Consequently,
\[
 E_M(\mathsf{Zero})\text{ is almost surely terminating}
 \quad\Longleftrightarrow\quad
 M\in\mathsf{TOT}.
\]
This is an effective many-one reduction from $\mathsf{TOT}$.
\end{proof}

\begin{corollary}
\label{cor:ast-pi2-complete}
Almost-sure termination of order-2 PHORS is $\Pi^0_2$-complete.
\end{corollary}

\begin{proof}
The lower bound is Theorem~\ref{thm:ast-pi2-hard}; the matching $\Pi^0_2$
upper bound was established in the main body of the paper.
\end{proof}
 \end{document}